\documentclass{article}

\usepackage[margin=1.3in]{geometry}
\usepackage{enumitem}
\usepackage[utf8]{inputenc} 
\usepackage[T1]{fontenc}    
\usepackage{hyperref}       
\usepackage{url}            
\usepackage{booktabs}       
\usepackage{amsfonts}       
\usepackage{nicefrac}       
\usepackage{microtype}      
\usepackage{lipsum}		
\usepackage{graphicx}
\usepackage[authoryear,round]{natbib}
\usepackage{doi}
\usepackage{amsmath, amsthm, bm}
\usepackage[plain]{algorithm}
\usepackage{algpseudocode}
\usepackage{tabularx}
\usepackage[dvipsnames]{xcolor}
\usepackage{tikz}

\usepackage[figurename=Fig., labelfont=bf, labelsep=space]{caption}

\newtheorem{theorem}{Theorem}
\newtheorem{prop}{Proposition}
\newtheorem{corollary}{Corollary}
\newtheorem{example}{Example}
\newtheorem{lemma}{Lemma}
\theoremstyle{remark}
\newtheorem{remark}{Remark}

\providecommand{\keywords}[1]
{
  \small	
  \textbf{\textit{Keywords---}} #1
}

\title{Adaptive random neighbourhood informed Markov chain Monte Carlo for high-dimensional Bayesian variable selection}

\author{
Xitong Liang,
Samuel Livingstone \& Jim Griffin
\\ 
\small \textit{Department of Statistical Science, University College London, UK}
\\
\small \texttt{xitong.liang.18@ucl.ac.uk}; \texttt{samuel.livingstone@ucl.ac.uk}; \texttt{j.griffin@ucl.ac.uk}
}

\hypersetup{
pdftitle={A template for the arxiv style},
pdfsubject={q-bio.NC, q-bio.QM},
pdfauthor={David S.~Hippocampus, Elias D.~Striatum},
pdfkeywords={First keyword, Second keyword, More},
}

\begin{document}
\maketitle

\begin{abstract} 
We introduce a framework for efficient Markov Chain Monte Carlo (MCMC) algorithms targeting discrete-valued high-dimensional distributions, such as posterior distributions in Bayesian variable selection (BVS) problems. We show that many recently introduced algorithms, such as the locally informed sampler of \cite{zanella2020informed} and the Adaptively Scaled Individual Adaptation sampler (ASI) of \cite{griffin2021search}, can be viewed as particular cases within the framework. We then describe a novel algorithm, the \textit{Adaptive Random Neighbourhood Informed} sampler (ARNI), by combining ideas from both of these existing approaches.  We show using several examples of both real and simulated datasets that a computationally efficient point-wise implementation (PARNI) leads to relatively more reliable inferences on a range of variable selection problems, particularly in the very large $p$ setting. 
\end{abstract}

\noindent \keywords{Bayesian computation; variable selection;  spike-and-slab priors; Markov chain Monte Carlo; random neighbourhood samplers; locally informed Metropolis-Hastings schemes}

\section{Introduction}
\label{sec:introduction}

Despite their long history, linear regression models remain a key building block of many present-day statistical analyses. In the modern setting, practitioners not only show interest in making good predictions but also intend to investigate underlying low-dimensional structure based on the belief that only a small subset of predictors play a crucial role in predicting the response. These problems can be addressed by \emph{variable selection}. A variable selection method is an automatic procedure that selects the best (small) subset of covariates that explains most of the variation in the response \citep{chipman2001practical}. Frequentist approaches focus on model comparisons through information criteria or point estimates, using e.g. maximum penalised likelihood under sparsity assumptions \citep{hastie2015statistical}. Alternatively the Bayesian approach can be taken by imposing an appropriate prior on all possible models and computing the posterior.

We consider Bayesian variable selection (BVS) with spike-and-slab priors \citep{mitchell1988bayesian}, which lead to natural uncertainty measures such as posterior model probabilities and marginal posterior variable inclusion probabilities.  For a linear regression model with $p$ candidate covariates, it has been shown that spike-and-slab priors often lead to \emph{posterior consistency} in the sense that the posterior collapses to a Dirac measure on the true model as more observations are gathered \citep{fernandez2001benchmark,liang2008mixtures,yang2016computational}, even in high-dimensional setting where $p$ grows with $n$ \citep{shang2011consistency,narisetty2014bayesian}.  Another approach is to employ continuous shrinkage priors \citep[\textit {e.g.}][]{polson2010shrink,griffin2021bayesian}, which only give posterior inference on regression coefficients but can result in a more computationally tractable posterior distribution.

The exact posterior distribution when using a spike-and-slab prior is challenging to compute, and when $p>30$ Markov chain Monte Carlo (MCMC) algorithms are typically used to estimate posterior summaries of interest \citep{george1993variable,chipman2001practical}. \cite{garcia2013sampling} discuss the properties of classical methods such as the Gibbs sampler in the context of variable selection. Other MCMC algorithms such as $\text{MC}^3$ \citep{madigan1995bayesian} and the Add-Delete-Swap method \citep{brown1998bayesian} are also well-studied. \cite{fernandez2001benchmark}, for instance, use the $\text{MC}^3$ sampler to explore the benchmark hyperparameter in a $g$-prior. 
\cite{yang2016computational} use Add-Delete-Swap to provide the conditions for \textit{rapid mixing} in the sense that the mixing time grows at
most polynomially in $p$. These approaches do, however, suffer from an unexpectedly long mixing time and therefore slow convergence when $p$ is large. For this reason, alternative approaches have also been developed. \cite{hans2007shotgun} describe a novel \emph{Shotgun Stochastic Search} (SSS) algorithm to identify a subset of probable models and then approximate the posterior distribution by restricting calculations to this subset. \cite{papaspiliopoulos2017bayesian} propose an efficient deterministic iterative procedure of model search under block-diagonal design. \cite{ray2021variational} introduce a mean-field spike and slab variational Bayes approximation to the BVS posterior, and show that such an approximation converges to the sparse truth at the optimal rate and gives optimal prediction of the response vector under mild conditions.

MCMC methods for variable selection can be improved using adaptive proposals within a Metropolis--Hastings algorithm \citep{metropolis1953equation,hastings1970monte}. Adaptive MCMC is a sub-class of algorithms in which tuning parameters are automatically updated ``on the fly'' \citep[\textit{e.g.}][]{andrieu2008tutorial}. Several adaptive methods have been developed in the context of BVS. \cite{ji2013adaptive} suggest using a multivariate Gaussian proposal kernel for the regression coefficient in which the mean and covariance matrix are learned as the algorithm runs. \cite{lamnisos2009transdimensional} and \cite{lamnisos2013adaptive} develop adaptive Gibbs, $\text{MC}^3$ and Add-Delete-Swap samplers. More recently \cite{griffin2021search} developed the \textit{Adaptively-Scaled Individual Adaptation} sampler (ASI), which is able to adapt the importance of each individual predictor variable and propose multiple swaps per iteration in high-dimensional settings. The algorithm works by learning the posterior inclusion probabilities for each covariate adaptively in a Rao--Blackwellised manner, and using these probabilities to construct a Metropolis--Hasting proposal under the assumption of an orthogonal design matrix.

Recently \textit{informed} MCMC schemes have gained popularity for problems with discrete parameter spaces (having already achieved prominence in the continuous setting). Informed MCMC schemes are those in which the Metropolis--Hastings proposal exploits some information about the posterior distribution. Intuitively, the success of informed proposals relies on avoiding models with low posterior model probabilities \citep{zhou2021dimension}. \cite{titsias2017hamming} propose the Hamming ball sampler (HBS) in which models are proposed in proportion to their locally-truncated posterior probability within a Hamming ball neighbourhood. \cite{zanella2019scalable} consider a Tempered Gibbs sampler (TGS), which involves importance sampling and more frequently updates components with lower conditional distributions. A more general class of locally informed proposals is introduced by \cite{zanella2020informed}. These locally informed proposals can be obtained by weighting a base kernel using a \emph{balancing function}, which is a function of the posterior distribution that satisfies a certain functional property. The base kernel is typically concentrated on a neighbourhood of the current state, resulting in a proposal that is informed or \emph{balanced} using ``local'' information about the posterior. The author shows that a random walk proposal is asymptotically dominated by its locally balanced counterpart in the Peskun sense as dimensionality increases \citep{peskun1973optimum,tierney1998note}. It has also been shown by \cite{zhou2021dimension} that a dimension-free mixing time bound can be constructed for a properly-designed locally informed MCMC algorithm for BVS. For other developments concerning locally informed proposals, see \textit{e.g.} \cite{livingstone2019barker,gagnon2021informed,power2019accelerated}. 

Informed MCMC schemes often mix quickly and have good convergence properties, but the computation of each transition can be prohibitively expensive. For example, in Bayesian variable selection, informed proposals usually involve calculating multiple posterior model probabilities within a neighbourhood of the current model (such as all models within Hamming distance one). The size of this neighbourhood will be at least linear in $p$ and will tend to include large numbers of unimportant variables under standard sparsity assumptions.  In this paper we propose a solution to this problem by introducing a framework for constructing flexible and efficient MCMC algorithms based on locally balanced proposals within \textit{random} neighbourhoods. In large $p$ settings, this randomness can be exploited to control the size and the number of unimportant variables in generated neighbourhood.  This leads to  Markov chains with good convergence properties and controlled computational cost per iteration. To illustrate the power of the framework we develop a new MCMC algorithm for Bayesian variable selection in linear regression, namely the \textit{Point-wise Adaptive Random Neighbourhood Informed} (PARNI) sampler. We illustrate with extensive empirical results on both real and simulated datasets that the PARNI sampler yields good estimates for posterior quantities of interest and performs particularly well for well-known large-$p$ examples such as the PCR ($p=22,575$) and SNP ($p=79,748$) datasets. Random neighbourhoods have also been considered previously in the context of stochastic search \citep{chen2016paired}.

The rest of this paper is structured as follows. In Section \ref{sec:background}, we review BVS for the linear model along with prior specification. We also briefly describe both the ASI scheme of \cite{griffin2021search} and the locally informed method of \cite{zanella2020informed}. In Section \ref{sec:random_neighbourhood_samplers}, we characterise a construction of random neighbourhood proposals  and illustrate that locally informed proposal and the ASI scheme fall in this framework. Section \ref{sec:adaptive_informed} presents the construction of adaptive random neighbourhood and informed samplers. Following this structure, we present the ARNI and PARNI samplers. In addition, we establish both the ergodicity and a strong law of large numbers for the PARNI algorithm. We implement the PARNI sampler in Section \ref{sec:numerical_studies} on both simulated and real data. Comparisons between the PARNI sampler and other state-of-the-art MCMC algorithms are carried out to showcase its capacity and efficiency. In Section \ref{sec:discussion_future} we discuss limitations and possible future work. Detailed explanations and proofs are provided in the supplement.

\section{Background}
\label{sec:background}

\subsection{Bayesian variable selection for the linear regression model}

Consider a dataset $\{(y_i,x_{i1},...,x_{ip})\}_{i=1}^n$, where the vector $y = (y_1,...,y_n) \in \mathbb{R}^n$ is called the response variable and each $x_j = (x_{1j},...,x_{nj})$ is one of $p$ predictor variables or covariates.  The variable selection problem is concerned with finding the best $q \ll p$ covariates that are most associated with the response.  Assuming that each regression includes an intercept, then there are $2^p$ possible models that can be formulated to predict the response. We refer to each model as $M_{\gamma}$ where the models are indexed by the indicator variable $\gamma = (\gamma_1, \dots, \gamma_p) \in \Gamma = \{0,1\}^p$, where $\gamma_j = 1$ if the $j$-th variable is included in model $M_\gamma$ and $\gamma_j = 0$ otherwise. We refer to $\Gamma$ as model space and let $p_\gamma:= \sum_j \gamma_j$. The model $M_\gamma$ associated with $\gamma$ is then
\begin{align}
    y = \alpha \textbf{1}_n + X_\gamma \beta_\gamma + \epsilon
\end{align}
where $\epsilon \sim N_n(0, \sigma^2 I_n)$, $y$ is an $n$-dimensional response vector, $X_\gamma$ is an $(n \times p_\gamma)$ design matrix which consists of the ``active'' variables in $\gamma$ (those for which $\gamma_j = 1$), $\alpha$ is an intercept term and $\beta_\gamma \in \mathbb{R}^{p_\gamma}$. In the Bayesian framework, we consider a commonly-used conjugate prior specification
\begin{align*}
    p(\alpha)  \propto 1, \quad
    \beta_\gamma|\gamma, \sigma^2  \sim N(0, g \sigma^2 V_\gamma), \quad
    p(\sigma^2) \propto \sigma^{-2}, \quad
    p(\gamma)  = h^{p_{\gamma}} (1-h)^{p-p_{\gamma}}.
\end{align*}
For simplicity, we can remove the intercept term $\alpha$ by centering $y$ and $X_j$ for all $j$. \cite{chipman2001practical} highlight that this method can be motivated from a formal Bayesian perspective by integrating out the coefficients corresponding to those fixed regressors with respect to an improper uniform prior. The covariance matrix $V_\gamma$ is often chosen as $(X_\gamma^T X_\gamma)^{-1}$ (a $g$-prior) or identity matrix $I_{p_\gamma}$ (an independent prior). For both of these choices, the marginal likelihood $p(y|\gamma)$ is analytically tractable. Suitable values for the global scale parameter $g$ are suggested in \cite{fernandez2001benchmark}. It can also be driven by a hyperprior, yielding a fully Bayesian model (see \cite{liang2008mixtures} for details). The hyperparameter $h \in (0,1)$ is the prior probability that each variable is included in the model. \cite{steel2007effect} suggest against using fixed $h$ unless strong information is given, and instead placing a hyperprior on it such as a Beta prior $h \sim \text{Beta}(a,b)$, leading to a Beta-binomial prior on the model size. In the following sections, we will develop efficient sampling schemes targeting the posterior distribution $\pi(\gamma) \propto p(y|\gamma)p(\gamma)$.

\subsection{Adaptively Scaled Individual Adaptation algorithm}

\cite{griffin2021search} introduce a scalable adaptive MCMC algorithm targeting high-dimensional BVS posterior distributions together with a method that automatically updates the tuning parameters. They consider the class of proposal kernels
\begin{align}
    q_{\eta} (\gamma, \gamma^\prime) = \prod_{j=1}^p q_{\eta, j}(\gamma_j, \gamma_j^\prime) \label{mat:asi_prop}
\end{align}
where $\eta = (A, D) = (A_1, \dots, A_p, D_1, \dots, D_p)$, $q_{\eta, j}(\gamma_j=0, \gamma_j^\prime=1) = A_j$ and $q_{\eta, j}(\gamma_j = 1, \gamma_j^\prime = 0) = D_j$, with Metropolis-Hastings acceptance probability
\begin{align}
    \alpha_\eta(\gamma, \gamma^\prime) = \left\{1, \frac{\pi(\gamma^\prime)q_{\eta}(\gamma^\prime, \gamma)}{\pi(\gamma)q_{\eta}(\gamma, \gamma^\prime)}\right\}. \label{mat:asi_accrate}
\end{align}
This proposal mainly benefits from two aspects. Firstly, the flexibility offered by $2p$ tuning parameters allows the proposal to be tailored to the data. Secondly, this form of proposal also allows multiple variables to be added or deleted from the model in a single iteration, which in turn allows the algorithm to make large jumps in model space.

\begin{algorithm}[t]
\caption{Adaptively Scaled Individual Adaptation (ASI)}
\label{alg:ASI}

\begin{algorithmic}
\For{$i = 1$ to $i = N$}
\For{$l = 1$ to $l = L$}
\State Sample $\gamma^{l,\prime} \sim q_{\zeta^{(i)}  \eta^{(i)}}(\gamma^{l,(i)}, \cdot)$ as in (\ref{mat:asi_prop}) and $U \sim U(0,1)$;
\State If $U < \alpha_{\zeta^{(i)}  \eta^{(i)}}(\gamma^{l,(i)}, \gamma^{l,\prime})$ as in (\ref{mat:asi_accrate}), then $\gamma^{l,(i+1)} = \gamma^{l,\prime}$, else $\gamma^{l,(i+1)} = \gamma^{l,(i)}$;
\EndFor
\State Update $\hat{\pi}^{(i+1)}_j$ as in (\ref{mat:rb_pips}), set $\tilde{\pi}^{(i+1)}_j = \pi_0 + (1-2\pi_0)\hat{\pi}^{(i+1)}_j$ for $j=1, \dots, p$;
\State Update $\zeta^{(i+1)}$ as in (\ref{mat:update_zeta});
\State Update $A^{(i+1)}_j = \min \left\{1, \tilde{\pi}^{(i+1)}_j/(1-\tilde{\pi}^{(i+1)}_j) \right\}$;
\State Update $D^{(i+1)}_j = \min \left\{1, (1-\tilde{\pi}^{(i+1)}_j)/\tilde{\pi}^{(i+1)}_j \right\}$;
\State Set $\eta^{(i+1)} = (A^{(i+1)}_j, D^{(i+1)}_j)$
\EndFor
\end{algorithmic}
\end{algorithm}

\cite{griffin2021search} suggest an optimal choice of $\eta = (A, D)$ in Peskun sense while assuming that all variables are independent. If $\pi_j$ denotes the posterior inclusion probability of the $j$-th regressor, the optimal choice of $\eta^{\text{opt}} = (A^{\text{opt}}, D^{\text{opt}})$ is given as
\begin{align}
    A^{\text{opt}}_j = \min\left\{1, \frac{\pi_j}{1-\pi_j}\right\}, \quad
    D^{\text{opt}}_j = \min\left\{1, \frac{1-\pi_j}{\pi_j}\right\}. \label{mat:ASI_AD_opt}
\end{align}
The independence assumption is usually violated due to the correlation between regressors and therefore a scaled proposal with parameters $\eta = \zeta \eta^{\text{opt}}$ for a scaling parameter $\zeta \in (0,1)$ is suggested. This scaling parameter $\zeta$ controls the number of variables that differ between the current state $\gamma$ and the proposed state $\gamma'$. Smaller values of $\eta$ can be used to avoid overly ambitious moves with low probabilities of acceptance and so control the average acceptance rate. They also suggest multiple chain acceleration with common adaptive parameters since running multiple independent chains with shared adaptive parameters can facilitate the convergence of the adaptive parameters \citep{craiu2009learn}. This phenomenon is guaranteed in their simulation studies where the schemes with 25 multiple chains outperform the schemes only with 5 multiple chains in terms of relative efficiency especially for large $p$ datasets. Suppose $L$ chains are used and let $\gamma^{l,(i)}$ and $\gamma^{l,\prime}$ denote the current state and proposal respectively for the $l$-th chain. The tuning parameters of the proposal are updated on the fly using 
a Rao-Blackwellised estimate of the posterior inclusion probability of the $j$-th regressor which, at the $N$-th iteration, is
\begin{align}
\hat{\pi}^{(N)}_j = \frac{1}{NL} \sum_{i=1}^N \sum_{l=1}^L \frac{\pi(\gamma_j = 1, \gamma^{l,(i)}_{-j}|y)}{\pi(\gamma_j = 1, \gamma^{l,(i)}_{-j}|y) + \pi(\gamma_j = 0, \gamma^{l,(i)}_{-j}|y)}. \label{mat:rb_pips}
\end{align}
The use of the Rao-Blackwellised estimates of the posterior inclusion probabilities can swiftly distinguish unimportant variables. \cite{griffin2021search} show how this can be calculated in $\mathcal{O}(p)$ operations which leads to a scalable MCMC scheme in large-$p$ BVS problems. At the $i$-th iteration, the proposal parameters are $\eta = \zeta^{(i)} \times \eta^{(i)}$ where $\eta^{(i)} = (A^{(i)}, D^{(i)})$,
\begin{align}
    A^{(i)}_j = \min\left\{1, \frac{\hat{\pi}^{(i)}_j}{1-\hat{\pi}^{(i)}_j}\right\}, \quad
    D^{(i)}_j = \min\left\{1, \frac{1-\hat{\pi}^{(i)}_j}{\hat{\pi}^{(i)}_j}\right\} \label{mat:ASI_paraform}
\end{align}
and 
the scaling parameter $\zeta^{(i)}$ is tuned using the Robbins-Monro scheme
\begin{align}
    \text{logit}_\epsilon \zeta^{(i+1)} = \text{logit}_\epsilon \zeta^{(i)} + \frac{\phi_i}{L} \sum_{l=1}^L (\alpha_{\zeta^{(i)} \eta^{(i)}}(\gamma^{l,(i)},\gamma^{l,\prime}) - \tau) \label{mat:update_zeta}
\end{align}
for a target rate of acceptance
$\tau$  and the mapping $\text{logit}_\epsilon:(\epsilon, 1-\epsilon) \to \mathbb{R}$ is a modified logistic function (or logit function) defined by 
\begin{align}
    \text{logit}_\epsilon(x) = \log(x-\epsilon) - \log(1-x-\epsilon)  \label{mat:logit_eps}
\end{align}
for some small $\epsilon \in (0, 1/2)$. The full description of the sampler is given in Algorithm \ref{alg:ASI}. 
 The resulting algorithm is called \emph{Adaptively Scaled Individual Adaptation} (ASI).
 \cite{griffin2021search}  establish the $\pi$-ergodicity and a strong law of large numbers for the ASI sampler.

\subsection{Locally informed proposals}
\label{subsec:locally_informed}

In continuous sample space, MCMC algorithms often utilise gradients of the target distribution, \textit{e.g.} the \emph{Metropolis-adjusted Langevin algorithm} \citep{grenander1994representations} and \emph{Hamiltonian Monte Carlo} \citep{duane1987hybrid}. These methods are defined on continuous spaces but \cite{zanella2020informed} develop a class of locally informed proposals as an analog for discrete spaces. The approach assumes that we can define a random walk Metropolis proposal kernel $Q$ on a neighbourhood $N \subset \Gamma$ with mass function $q$. \cite{zanella2020informed} considers a class of pointwise informed proposals $Q_g$ defined as follows
\begin{align}
    q_g(\gamma, \gamma^\prime) =
    \begin{cases}
    \frac{g\left(\frac{\pi(\gamma^\prime)}{\pi(\gamma)}\right)q(\gamma, \gamma^\prime)}{Z_g(\gamma)}, \quad & \gamma \in N \\
    0, \quad & \text{otherwise}
    \end{cases}
\end{align}
where $g:[0,\infty)\to[0,\infty)$ is a continuous function and $Z_g(\gamma)$ is a normalising constant such that
\begin{align}
    Z_g(\gamma) = \sum_{\gamma^\prime \in N} g\left(\frac{\pi(\gamma^\prime)}{\pi(\gamma)}\right)q(\gamma, \gamma^\prime).
\end{align}
The choice of the function $g$ is crucial for the performance of $Q_g$ since it determines how the target distribution $\pi$ drives the proposal. \emph{Locally balanced proposals} are defined to be those proposals $Q_g$ which are approximately $\pi$-reversible if $Q$ only proposes local moves. \cite{zanella2020informed}  shows that this occurs if $g(t) = tg(1/t)$ for all $t > 0$ and calls these \emph{balancing functions}. The neighbourhoods are often chosen to be $N = \mathcal{H}_m(\gamma) := \{\gamma^\prime \in \Gamma| d_H(\gamma^\prime, \gamma) \leq m\}$ for which $d_H(\cdot,\cdot)$ denotes the measure of Hamming distance (\textit{i.e.} $d_H(\gamma, \gamma^\prime) = \sum_{j=1}^p |\gamma_j - \gamma_j^\prime|$) and the proposal kernel $Q$ would be a uniform distribution on the neighbourhood $N$. When $m$ is taken to be 1, $Q$ is identical to the $\text{MC}^3$ sampler.

Both analytical and empirical results suggest that a locally balanced proposal outperforms other alternatives such as the original non-informative kernel (when $g(t) = 1$) or the globally balanced kernel (when $g(t) = t$). Theorem 5 of \cite{zanella2020informed} shows that using a uniform based kernel on neighbourhood $\mathcal{H}_m(\gamma)$ will be asymptotically optimal, in terms of Peskun ordering, as the dimensionality goes to infinity under mild conditions.

\section{Random neighbourhood samplers}
\label{sec:random_neighbourhood_samplers}

We consider a framework for constructing Metropolis--Hastings proposals to sample from $\pi(\gamma)$ in which a new state is proposed within a \textit{neighbourhood} around the current state. The neighbourhoods can be \text{random}, and are generated using an auxiliary variable $k$ as a neighbourhood indicator. The auxiliary variable $k$ is a discrete random variable defined on a countable set $\mathcal{K}$ such that a neighbourhood $N = N(\gamma, k)$ is generated as the same probability of generating $k$ (\textit{i.e.} $p(N|\gamma) = p(k|\gamma)$). Suppose $\gamma$ is the current state and $Q_k$ is a Metropolis-Hastings proposal kernel (conditioned on $k$) with mass function $q_k$. A new state $\gamma^\prime$ is drawn from kernel $Q_k$ after a $k$ is generated. In updating $k$ at each iteration, we usually consider sending $k$ to a new state $k^\prime \in \mathcal{K}$ depends on a bijection $\rho:\mathcal{K} \to \mathcal{K}$ which is an involution (\textit{i.e.} a self-inverse function which satisfies $\rho(\rho(k)) = k$). We call an MCMC algorithm that uses the above construction to generate Metropolis--Hastings proposals a \emph{random neighbourhood sampler}. The followings are some examples of random neighbourhood samplers.

\begin{example}[Samplers with non-stochastic neighbourhoods]

In fact, samplers with non-stochastic neighbourhoods are also random neighbourhood samplers where the specific neighbourhoods are generated with constant probability of 1 at each state $\gamma$. In such cases, the choices of $k$ and $\rho$ can be arbitrary. For instance, the $\text{MC}^3$ sampler can be viewed as a random neighbourhood sampler for which the neighbourhood $N$ consists of models that are $1$-Hamming distance from $\gamma$. In particular, the locally informed samplers of \cite{zanella2020informed} also belong to this class with neighbourhood $N$ as defined in Section \ref{subsec:locally_informed}. 
\end{example}

\begin{example}[Add-Delete-Swap sampler]

In each iteration of an Add-Delete-Swap sampler, a strategy from ``addition'', ``deletion'' and ``swap'' is uniformly chosen which implies that the auxiliary variable $k$ is uniformed distributed over the sample space $\mathcal{K} = \{\text{``addition''}, \text{``deletion''}, \text{``swap''}\}$ and therefore construct a neighbourhood $N(\gamma, k)$ as in \cite{yang2016computational}. A new state $\gamma^\prime$ is uniformly proposed from $N(\gamma, k)$. The corresponding mapping $\rho$ is then a function that sends the auxiliary variable to an opposite strategy, \textit{e.g.} it sends ``addition'' to ``deletion'' and vice versa. Note that the opposite strategy of ``swap'' is itself.

\end{example}

\begin{example}[Hamming ball sampler]

A Hamming ball sampler with radius $m$ is described by \cite{titsias2017hamming}. This algorithm is literally evolved from a neighbourhood construction $\mathcal{H}_m(\gamma) \subset \Gamma$ which produces a subset of states at most $m$-Hamming distance away from $\gamma$. The auxiliary variable $k$ is equivalent to $U$ in their design in which $k$ is uniformly distributed over the set $\mathcal{K} = \mathcal{H}_m(\gamma)$ and a neighbourhood $N(\gamma, k) = \mathcal{H}_m(k)$ is used to draw new state. The Hamming ball sampler proposes a new state according to the truncated posterior model probability in the neighbourhood $N(\gamma, k)$. In this scheme, the mapping $\rho$ is an identity function which indicates the same auxiliary variable is used in reversed moves.
\end{example}

The full update of a random neighbourhood sampler uses the three stages below:
\begin{enumerate}[label= (\roman*)]
    \item (\emph{Neighbourhood construction}) Sample a neighbourhood indicator $k$ from $p(\cdot|\gamma)$, and construct the corresponding neighbourhood $N(\gamma, k)$;
    \item (\emph{Within-neighbourhood proposal}) Propose a new model $\gamma^\prime$ in $N(\gamma, k)$ according to $Q_k(\gamma, \cdot)$;
    \item (\emph{Accept/reject step}) Calculate the probability of the reverse move, $q_{\rho(k)}(\gamma^\prime, \gamma)$, by constructing the reverse neighbourhood $N(\gamma^\prime, \rho(k))$. Move to the new state $\gamma^\prime$ with probability $\alpha_k(\gamma, \gamma^\prime)$ where $\alpha_k(\gamma, \gamma^\prime)$ is the Metropolis-Hastings acceptance probability
    \begin{align}
        \alpha_k(\gamma, \gamma^\prime) = \min\left\{1, \frac{\pi(\gamma^\prime)p(\rho(k)|\gamma^\prime)q_{\rho(k)}(\gamma^\prime, \gamma)}{\pi(\gamma)p(k|\gamma)q_k(\gamma, \gamma^\prime)}\right\}. \label{mat:RN_accrate}
    \end{align}
\end{enumerate}

Throughout this article, we refer to the above three stages as \emph{neighbourhood construction}, \emph{within-neighbourhood proposal} and \emph{accept/reject step} respectively. To preserve the reversibility of the chain, it is better to design a neighbourhood generation scheme where the law
\begin{align}
    \gamma^\prime \in N(\gamma, k) \iff \gamma \in N(\gamma^\prime, \rho(k)) \label{con:neighbourhood_construction_1}
\end{align}
holds for any $\gamma$, $\gamma^\prime$ and $k$. Upon this law, we assume that the condition
\begin{align}
    p(k|\gamma)q_k(\gamma, \gamma^\prime) > 0  \iff p(\rho(k)|\gamma)q_{\rho(k)}(\gamma^\prime, \gamma) > 0 \label{con:neighbourhood_construction_2}
\end{align}
is satisfied. This assumption is a generalisation of the paired-move strategy in \cite{chen2016paired} and it results in the correctness and reversibility of such a scheme through the following proposition.

\begin{prop} \label{prop:RN_correct}
A random neighbourhood sampler is $\pi$-reversible provided that condition (\ref{con:neighbourhood_construction_2}) holds, $p(k|\gamma)$ is a valid probability measure on $\mathcal{K}$ and $q_k(\gamma, \gamma^\prime)$ is a valid probability measure on neighbourhood $N(\gamma, k)$ for all $\gamma \in \Gamma$ and $k \in \mathcal{K}$.
\end{prop}

\begin{remark}
To generalise the framework of random neighbourhood samplers, it is possible to use a continuous auxiliary variable $k$. In such a case, the acceptance probability in (\ref{mat:RN_accrate}) should include the Jacobian term.
\end{remark}

We show in the next part that the ASI sampler is also random neighbourhood sampler but it takes a strategy unlike the locally balanced proposals. Instead of paying more attention to the within-neighbourhood proposal but using a naive neighbourhood construction, the ASI sampler more focuses on constructing sophisticated random neighbourhoods which is more likely to contain promising models and employs a random walk within-neighbourhood proposal.

\subsection{Another take on the ASI scheme}

It is not straightforward to observe that the ASI sampler is a random neighbourhood sampler, however we show below that in fact it can be. To do so, we introduce a random neighbourhood sampler, the \textit{Adaptive Random Neighbourhood} (ARN) sampler, and prove that the ARN and ASI samplers are equivalent if they share some common adaptive parameters. We also illustrate that the ARN sampler differs from the locally informed approach, instead it puts more efforts on neighbourhood construction but its within-neighbourhood proposal follows a random walk.

We consider a random neighbourhood sampler with algorithmic tuning parameter $\theta = (\xi \eta^{\text{opt}}, \omega) \in (\epsilon, 1-\epsilon)^{2p+1} := \Delta_\epsilon^{2p+1}$, where $\eta^{\text{opt}}$ is given in \eqref{mat:ASI_AD_opt}, and the tuning parameters $\xi$ and $\omega$ are used in the random neighbourhood construction and the  within-neighbourhood proposal respectively. In the random neighbourhood construction, the neighbourhood indicator variable $k = (k_1, \dots, k_p) \in \mathcal{K} = \{0,1\}^p$ is generated from the distribution
\begin{align}
    p^{\text{RN}}_{\xi \eta^{\text{opt}}}(k| \gamma) = & \prod_{j = 1}^p p^{\text{RN}}_{\xi \eta^{\text{opt}}, j}(k_j| \gamma_j) \label{mat:RN_prop1}
\end{align}
where $p^{\text{RN}}_{\xi \eta^{\text{opt}},j}(k_j=1|\gamma_j=0) = \xi A_j^{\text{opt}}$ 
and $p^{\text{RN}}_{\xi \eta^{\text{opt}},j}(k_j=1|\gamma_j=1) = \xi D_j^{\text{opt}}$.
This is equivalent to the ASI proposal in \eqref{mat:asi_prop} where $k_j=1$ if and only if $\gamma_j \neq \gamma_j^\prime$. A neighbourhood $N(\gamma, k)$ is obtained from $\gamma$ and $k$ for which $\gamma$ is the ``centre'' of $N(\gamma, k)$ and $k$ indicates the possible indices altered from $\gamma$. These tuning parameters $\xi$ and $\eta^{\text{opt}}$ are abortively updated on the fly. For any $\gamma^* \in N(\gamma, k)$, $k_j = 0$ implies that $\gamma^*_j = \gamma_j$. This identity can be used to state a formal definition of the neighbourhood $N(\gamma, k)$ as
\begin{align*}
N(\gamma, k) = \{ \gamma^* \in \Gamma |  \gamma_j = \gamma_j^*, ~ \forall k_j = 0  \}.
\end{align*}
The neighbourhood contains $2^{p_k}$ models where $p_k$ is number of 1s in $k$ (\textit{i.e.} $p_k := \sum_{j=1}^p k_j$). The parameter $\xi$ affects the $p_k$ and therefore controls neighbourhood size. So we call $\xi$ the neighbourhood scaling parameter.

The mapping $\rho$ is chosen to be an identity function. The within-neighbourhood proposal within this \emph{adaptive} random neighbourhood scheme is also based on the same proposal in \eqref{mat:asi_prop} over the neighbourhood $N(\gamma, k)$. It can be characterised as choosing the variables to be added or deleted from the model by thinning from within the set $\{j | ~ k_j = 1 \}$ with the thinning probability set to be $\omega \in (0,1)$. We refer to this parameter as  the unique \emph{within-neighbourhood proposal tuning parameter}. A larger value of $\omega$ increases the probability of proposing $\gamma^\prime$ further away from $\gamma$ in Hamming distance.
This can be written formally as the proposal in \eqref{mat:asi_prop}  with tuning parameter $\eta^{\text{THIN}} = (A^{\text{THIN}}, D^{\text{THIN}}) = (\omega k,\omega k)$, that is $A^{\text{THIN}}_j = D^{\text{THIN}}_j = \omega$ for $k_j = 1$ and $A^{\text{THIN}}_j = D^{\text{THIN}}_j = 0$ otherwise. The resulting proposal is termed as $q^{\text{THIN}}_{\omega, k}$ which is formulated as
\begin{align}
q_{\omega, k}^{\text{THIN}}(\gamma, \gamma^\prime) = \prod_{j=1}^p q_{\omega, k_j}^{\text{THIN}}(\gamma_j, \gamma^\prime_j), \label{mat:THIN_prop}
\end{align}
where $q_{\omega, 1}^{\text{THIN}}(\gamma_j, 1-\gamma_j) = \omega$ and $q_{\omega, 0}^{\text{THIN}}(\gamma_j, 1-\gamma_j) = 0$. The proposal $q_{\omega, k}^{\text{THIN}}$ is symmetric and only generates new states inside the neighbourhood $N(\gamma, k)$. This is because the probabilities of proposing flips on coordinates other than $j$ such that $k_j = 1$ are 0. The scheme is completed by accepting or rejecting the proposal using a standard Metropolis-Hastings acceptance probability 
\begin{align}
\alpha_{\theta, k}^{\text{ARN}} (\gamma, \gamma^\prime) = \left\{ 1,  \frac{\pi(\gamma^\prime) p^{\text{RN}}_{\xi \eta^{\text{opt}}}(k|\gamma^\prime) q^{\text{THIN}}_{\omega, k}(\gamma^\prime, \gamma) }{\pi(\gamma) p^{\text{RN}}_{\xi \eta^{\text{opt}}}(k|\gamma) q^{\text{THIN}}_{\omega, k}(\gamma, \gamma^\prime) }  \right\} \label{mat:ARN_acc}.
\end{align}

\begin{remark} \label{remark:thin_proposal}
An alternative formulation to (\ref{mat:THIN_prop}) in terms of Hamming distance between $\gamma$ and $\gamma^\prime$ is
\begin{align}
q_{\omega, k}^{\text{THIN}}(\gamma, \gamma^\prime) & = \omega^{d_H(\gamma, \gamma^\prime)} (1-\omega)^{p_k - d_H(\gamma, \gamma^\prime)} \mathbb{I}\{ \gamma^\prime \in N(\gamma, k)\} \nonumber \\
& = \left(\frac{\omega}{1-\omega}\right)^{d_H(\gamma, \gamma^\prime)} (1-\omega)^{p_k} \mathbb{I}\{ \gamma^\prime \in N(\gamma, k)\} \label{mat:ARN_prop3}
\end{align}
where $d_H(\gamma, \gamma^\prime)$ is the measure of Hamming distance between two models $\gamma$ and $\gamma^\prime$. 
\end{remark}

\begin{remark} \label{remark:thin_proposal_unif}
   When $\omega$ is chosen to be $1/2$, the within-neighbourhood proposal $q^{\text{THIN}}_{\omega = 1/2, k}$ is uniformly distributed over the local neighbourhood $N(\gamma, k)$.
\end{remark}

\begin{algorithm}[t]
\caption{Adaptive Random Neighbourhood sampler (ARN)} \label{alg:ARN}
\begin{algorithmic}
\State Sample $k \sim p^{\text{RN}}_{\xi \eta^{\text{opt}}}(\cdot|\gamma)$ as in (\ref{mat:RN_prop1});
\State Sample $\gamma^\prime \sim q_{\omega, k}^{\text{THIN}}(\gamma, \cdot)$ as in (\ref{mat:THIN_prop}) and $U \sim U(0,1)$;
\State if $U < \alpha_{\theta, k}^{\text{ARN}}(\gamma, \gamma^\prime)$ as in (\ref{mat:ARN_acc}), then accept $\gamma^\prime$ 
\end{algorithmic}
\end{algorithm}

Algorithm \ref{alg:ARN} describe how a new state $\gamma^\prime$ is proposed using the ARN scheme. We indicate the transition kernel by $p^{\text{ARN}}_{\theta}$ and 
the corresponding sub-transition kernel conditional on $k$ by $p^{\text{ARN}}_{\theta, k}$. They obey the relationship
\begin{align*}
    p^{\text{ARN}}_{\theta}(\gamma, \gamma^\prime) = \sum_{k \in \mathcal{K}} p^{\text{ARN}}_{\theta, k}(\gamma, \gamma^\prime).
\end{align*}

The following proposition helps to show that the ARN sampler is $\pi$-reversible.

\begin{prop} \label{prop:ARN_correct} For any tuning parameter $\theta = (\eta, \omega) \in \Delta_\epsilon^{2p+1}$, the condition (\ref{con:neighbourhood_construction_2}) holds, the conditional distribution of $k$, $p^{\text{RN}}_\eta(k|\gamma)$, within the ARN sampler is a valid distribution on $\mathcal{K} = \{0,1\}^p$. In addition, for any $\gamma \in \Gamma$ and $k \in \mathcal{K}$, the  within-neighbourhood proposal of the ARN sampler $q^{\text{THIN}}_{\omega, k}(\gamma,\gamma^\prime)$ is also a valid probability distribution on $N(\gamma,k)$.
\end{prop}

Proposition \ref{prop:RN_correct} together with Proposition \ref{prop:ARN_correct} show that the ARN transition kernel is $\pi$-reversible and therefore generates samples that preserve the target distribution $\pi$. In fact ARN and ASI are mathematically equivalent provided that the tuning parameter choices are made in a prescribed manner. To see this suppose that the tuning parameters of both the ARN and ASI schemes are fixed and share the same tuning parameter $\eta$. The following theorem shows that their transition probabilities from $\gamma$ to $\gamma^\prime$ are equal when $\zeta = \xi \times \omega$ holds. 

\begin{theorem} \label{them:asi_arn_equal}
Suppose that $\eta \in \Delta_\epsilon^{2p}$ and $\zeta$, $\xi$, $\omega \in \Delta_\epsilon$ for small $\epsilon \in (0, 1/2)$, $p^{\mathrm{ARN}}_{(\xi \eta,\omega)}$ and $p^{\mathrm{ASI}}_{\zeta \eta}$ are transition kernels of the ARN and ASI schemes respectively. If $\zeta = \xi \times \omega$ and, then
\begin{align}
p^{\mathrm{ARN}}_{(\xi \eta,\omega)} (\gamma, \gamma^\prime) = p^{\mathrm{ASI}}_{\zeta \eta} (\gamma, \gamma^\prime)
\end{align}
holds for any $\gamma$ and $\gamma^\prime \in \Gamma$.
\end{theorem}

In addition we deduce the following corollary.

\begin{corollary}
\label{coro:arns_equal}
Setting $\xi_1 \times \omega_1 = \xi_2 \times \omega_2$ implies
\begin{align*}
    p^{\mathrm{ARN}}_{(\xi_1\eta, \omega_1)} (\gamma, \gamma^\prime) = p^{\mathrm{ARN}}_{(\xi_2\eta, \omega_2)} (\gamma, \gamma^\prime)
\end{align*}
for any $\gamma$ and $\gamma^\prime \in \Gamma$.
\end{corollary}

Corollary \ref{coro:arns_equal} shows that two ARN kernels with different tuning parameters coincide in probability if the products of the neighbourhood scaling parameter $\xi$ and proposal thinning parameter $\omega$ are equal. This corollary also suggests that magnitudes of $\xi$ and $\omega$ can shift to each other without modifying the resulting proposal as long as their product preserves.

\section{Adaptive random neighbourhood and informed samplers}
\label{sec:adaptive_informed}

It should be clear from the above discussion that both the locally informed and ASI schemes can be viewed as random neighbourhood samplers, and that in the former much work is done to select proposals within a neighbourhood, while in the latter most of the work is done in constructing a neighbourhood where many of the candidate models represent proposals that are likely to be accepted by the algorithm. Our main methodological contribution is to design a random neighbourhood sampler for which both the neighbourhood construction and within-neighbourhood proposal are designed in an informed way. We therefore consider using an adaptive random neighbourhood approach to construct neighbourhoods, followed by a locally informed approach to select a proposal from within this neighbourhood.

The advantages of combining the two schemes in this manner are worth highlighting. A key strength of ASI is that generating proposals is computationally cheap, but when components of the posterior distribution are highly correlated then the assumption of independence that is embedded into the proposal generation can lead to overly ambitious moves that will be rejected. To combat this, the scaling parameter must be used to control the acceptance rate, but in the presence of high correlation this can lead to small moves and slow mixing. The locally informed sampler can cope well with high levels of correlation in the posterior distribution, but in high-dimensions often the (uninformed) neighbourhood will often either contain no sensible models, or be so large that the cost of computing all of the balancing function values within it becomes prohibitive.  Combining the two schemes is therefore an attractive proposition, as an intelligent neighbourhood that is also not too large can be constructed using ASI, and then correlation can be controlled for at the second stage by choosing the within-neighbourhood proposal using the locally informed approach.

We give the details of this informed and adaptive random neighbourhood sampler below, which we call the \textit{Adaptive Random Neighbourhood Informed} (ARNI) sampler. After this we define the point-wise ARNI (PARNI) scheme, which enjoys the benefits of ARNI but with much lower computational cost.

\subsection{Adaptive Random Neighbourhood Informed algorithm}

We first describe a construction of \textit{random neighbourhood informed proposals}. Suppose a random neighbourhood sampler is given with neighbourhood indicator variable $k \in \mathcal{K}$ and a update mapping $\rho$ together with a within-neighbourhood proposal kernel $Q_k$. The variable $k$ follows a conditional distribution $p(k|\gamma)$ whereas the proposal $Q_k$ produces a new state $\gamma^\prime$ within neighbourhood $N(\gamma, k)$ in an uninformed manner.

We consider the class of \emph{informed proposals} $Q_{g,k}$ with mass function
\begin{align}
    q_{g,k}(\gamma, \gamma^\prime) = 
    \begin{cases}
        \frac{g\left(\frac{\pi(\gamma^\prime)p(\rho(k)|\gamma^\prime)q_{\rho(k)}(\gamma^\prime, \gamma)}{\pi(\gamma)p(k|\gamma)q_k(\gamma, \gamma^\prime)}\right) q_k(\gamma, \gamma^\prime)}{Z_{g, k}(\gamma)}, & \gamma \in N(\gamma, k) \\
        0, & \text{otherwise}
    \end{cases}
\end{align}
where $g:[0,\infty)\to[0,\infty)$ is a continuous function, and $Z_{g, k}(\gamma)$ is a normalising constant defined by
\begin{align}
    Z_{g, k}(\gamma)=\sum_{\gamma^*\in N(\gamma, k)} g\left(\frac{\pi(\gamma^*)p(\rho(k)|\gamma^*)q_{\rho(k)}(\gamma^*, \gamma)}{\pi(\gamma)p(k|\gamma)q_k(\gamma, \gamma^*)}\right)q_k(\gamma, \gamma^*). \label{mat:locban_normalising}
\end{align}
The generated new state $\gamma^\prime$ is accepted using the Metropolis-Hastings rule
\begin{align}
\alpha_{g, k}(\gamma, \gamma^\prime) & =\min \left\{1, \frac{\pi(\gamma^\prime)p(\rho(k)|\gamma^\prime)q_{g,\rho(k)}(\gamma^\prime, \gamma)}{\pi(\gamma)p(k|\gamma)q_{g,k}(\gamma, \gamma^\prime)}\right\}.
\end{align}
The proposal collapses to the locally balanced proposal of \cite{zanella2020informed} when the neighbourhood is non-stochastic and the within-neighbourhood proposal is symmetric. According to the locally balanced proposals, using a balancing function $g$ that satisfies $g(t)=tg(1/t)$ is optimal in Peskun sense. In random neighbourhood informed proposals, we also recommend using a balancing function $g$.

In what follows, we combine the above random neighbourhood informed proposal with the ARN scheme and develop an \textit{Adaptive Random Neighbourhood Informed} (ARNI) proposal that uses an informed proposal at the within-neighbourhood proposal stage. In the ARNI scheme, the within-neighbourhood proposal in Algorithm 2 is replaced by 
\begin{align}
    q^{\text{ARNI}}_{\theta, k}(\gamma, \gamma^\prime) & \propto g\left(\frac{\pi(\gamma^\prime)p^{\text{RN}}_{\xi \eta^{\text{opt}}}(k|\gamma^\prime)q^{\text{THIN}}_{\omega, k}(\gamma^\prime, \gamma)}{\pi(\gamma)p^{\text{RN}}_{\xi \eta^{\text{opt}}}(k|\gamma)q^{\text{THIN}}_{\omega, k}(\gamma, \gamma^\prime)}\right) q^{\text{THIN}}_{\omega, k}(\gamma, \gamma^\prime) \notag \\
    & = g\left(\frac{\pi(\gamma^\prime)p^{\text{RN}}_{\xi \eta^{\text{opt}}}(k|\gamma^\prime)}{\pi(\gamma)p^{\text{RN}}_{\xi \eta^{\text{opt}}}(k|\gamma)}\right) q^{\text{THIN}}_{\omega, k}(\gamma, \gamma^\prime). \label{mat:ARNI_prop}
\end{align}
where $g$ is a balancing function such that $g(t) = tg(1/t)$ holds and $\theta = (\xi \eta^{\text{opt}}, \omega) \in \Delta_\epsilon^{2p+1}$. The last equation follows since the within-neighbourhood proposal $q^{\text{THIN}}_{\omega, k}$ is symmetric and therefore $q^{\text{THIN}}_{\omega, k}(\gamma^\prime, \gamma)/q^{\text{THIN}}_{\omega, k}(\gamma, \gamma^\prime) = 1$ for all $\gamma^\prime \in N(\gamma, k)$. The Metropolis-Hastings acceptance probability is tailored to the new informed proposal as
\begin{align}
    \alpha_{\theta, k}^{\text{ARNI}}(\gamma, \gamma^\prime) = \min \left\{1, \frac{\pi(\gamma^\prime)p^{\text{RN}}_{\xi \eta^{\text{opt}}}(k|\gamma^\prime)q^{\text{ARNI}}_{\theta, k}(\gamma^\prime, \gamma)}{\pi(\gamma)p^{\text{RN}}_{\xi \eta^{\text{opt}}}(k|\gamma)q^{\text{ARNI}}_{\theta, k}(\gamma, \gamma^\prime)}\right\} \label{mat:ARNI_acc}.
\end{align}

To boost the convergence of these adaptive tuning parameters, the same  multiple chain strategy as ASI should be implemented. In addition to the notations used in ASI, $k^{l,(i)}$ denotes the neighbourhood indicator variable for the $l$-chain at time $i$. For $L$ multiple chains, the tuning parameters $\eta^{\text{opt}}$ are updated following the same scheme of ASI as in (\ref{mat:ASI_paraform}) and (\ref{mat:update_zeta}). Two scaling parameters $\xi$ and $\omega$ can be updated using the Robbins-Monro schemes
\begin{align}
    \text{logit}_\epsilon\xi^{(i+1)} & = \text{logit}_\epsilon \xi^{(i)} + \frac{\phi_i}{L} \sum_{l=1}^L (p_{k^{l,(i)}} - s) \\
    \text{logit}_\epsilon\omega^{(i+1)} & = \text{logit}_\epsilon \omega^{(i)} + \frac{\phi_i}{L} \sum_{l=1}^L (\alpha^l_i - \tau)
\end{align}
where $p_k$ is the size of $k$ as mentioned previously, $s$ is the target size of $k$, $\alpha^l_i$ is the acceptance probability at the $i$th iteration for the $l$-th chain and $\tau$ is the target average acceptance rate. We recommend using the diminishing sequence $\phi_i = i^{-0.7}$ for both updating schemes.

While the informed proposal is powerful in accelerating the convergence of the chains, it also introduces extra computational costs since the posterior probabilities of all models in a neighbourhood are required. Given a $k$ of size $p_k$, the resulting neighbourhood $N(\gamma, k)$  consists of $2^{p_k}$ models. Although it is possible to speed up the posterior calculations using Gray codes as introduced in \cite{george1997approaches}, evaluating $2^{p_k}$ models is still computationally expensive when $p_k$ is very large and leads to an inefficient scheme. One way to address the issue is to tune the neighbourhood scaling parameter to generate neighbourhoods with a desired size, say let $s$ be 5. In our experience, such control of the size of $k$ comes at the cost of reduced exploration of the model space and the ARNI scheme fails to achieve better performance than ASI. This motivated us to develop a more efficient implementation of this approach that controls computational cost but maintains good exploration properties.

\subsection{The PARNI sampler}

We consider a point-wise implementation of the ARNI scheme (for short, the PARNI scheme). This approach is motivated by the block-wise implementation in \cite{zanella2020informed} and the block design strategy in \cite{titsias2017hamming}. The main idea is that the variables with neighbourhood indicator equal to 1 are divided into a series of smaller blocks and the new model is  sequentially proposed by sequentially adding or deleting variables in each block.
 The block design can lead to a significant reduction in the total number of models considered and so require less computational effort. For instance, suppose that there are $p_k$ non-zero neighbourhood indicator variables, which are divided into $m$ equally sized blocks. The neighbourhoods generated by each block will have $2^m$ models. Working through each block to propose a new state requires evaluating $2^m p_k/m$ posterior probabilities. As the computational cost is proportional to the total number of models considered, the computational cost is largest when $m = p_k$ where the only building block is the entire neighbourhood of $N(\gamma, k)$. In contrast, the smallest computational cost occurs when $m=1$ where each block has  one variable and therefore contains two models only. Throughout the section, we consider the latter block design when $m=1$ and the resulting algorithm is the PARNI sampler. 

\subsubsection{Main Algorithm}

We now formally present the PARNI algorithm and show how a new state $\gamma^\prime$ is proposed from a current state $\gamma$. We use the same random neighbourhood construction as the ARNI scheme, in addition, the neighbourhood scaling parameter $\xi$ is set to be fixed at 1 to indicate that neighbourhood sizes are not reduced at this stage. In other words the neighbourhoods are generated with the optimal values $\eta^{\text{opt}}$ as in (\ref{mat:ASI_AD_opt}). After a neighbourhood $N(\gamma, k)$ is sampled, we sequentially propose a sequence of new models with only 1-Hamming distance differences inside $N(\gamma, k)$. 
We define $K = \{K_1, \dots, K_{p_k}\} = \{j|k_j = 1\}$ to be the set of variables for which $k_j=1$ (the order of variables is random). We also define a sequence of models, $\gamma(1), \dots, \gamma(p_k)$ and  neighbourhoods, $N(1),\dots, N(p_k)$ to select the final model. Let $e(1),\dots,e(p)$ be the basis vector of a $p$-dimensional Cartesian space where $e(j)_j = 1$ and $e(j)_{j^\prime} = 0$ whenever $j^\prime \neq j$. We consider the neighbourhoods constructed according to $\gamma(j)$ and $e(K_r)$ for $r$ from 1 to $p_k$. The first neighbourhood is $N(1)=N(\gamma, e(K_1))$ and propose a model $\gamma(1)$ from
\begin{align}
    q^{\text{PARNI}}_{\theta,K_1}(\gamma, \gamma(1)) \propto 
    \begin{cases}
    g\left(\frac{\pi(\gamma(1))p^{\text{RN}}_{\eta^{\text{opt}}}(e(K_1)|\gamma(1))}{\pi(\gamma)p^{\text{RN}}_{\eta^{\text{opt}}}(e(K_1)|\gamma)}\right) q^{\text{THIN}}_{\omega, e(K_1)}(\gamma, \gamma(1)), \quad & \text{if $\gamma(1) \in N(1)$ } \\
    0, & \text{otherwise}
    \end{cases}
\end{align}
where $g$ is a balancing function and $\theta = (\eta^{\text{opt}}, \omega)$.
We repeat this process to construct the second neighbourhood $N(2) = N(\gamma(1), e(K_2))$ and propose the model $\gamma(2)$ from $N(2)$. In general, at time $r$,
we defined $N(r) = N(\gamma(r-1), e(K_r))$ and propose a model $\gamma(r)$ from
\begin{align}
\label{mat:PARNI_prop}
    q^{\text{PARNI}}_{\theta,K_r}(\gamma(r-1), \gamma(r)) \propto 
    \begin{cases}
    g\left(\frac{\pi(\gamma(r))p^{\text{RN}}_{\eta^{\text{opt}}}(e(K_r)|\gamma(r))}{\pi(\gamma(r-1))p^{\text{RN}}_{\eta^{\text{opt}}}(e(K_r)|\gamma(r-1))}\right) q^{\text{THIN}}_{\omega, e(K_r)}(\gamma(r-1), \gamma(r)), \quad & \text{if $\gamma(r) \in N(r)$ } \\
    0, & \text{otherwise.}
    \end{cases}
\end{align}
 We call this a position-wise probability measure because it only allows to alter or preserve position $K_r$ for each $r$. Fig. \ref{fig:PARNI_flowchart} provides a flowchart of the PARNI scheme which only involves enumerating at most
$2p_k$ models rather than $2^{p_k}$ models in the ARNI proposal. The parameters of the proposal are
$\theta = (\eta^{\text{opt}}, \omega)$.

To the construct a $\pi$-reversible chain, the probability of the reverse moves is required. These reverse moves use $K^\prime = \rho(K)$ as their auxiliary variables. The mapping $\rho$ reverses the order of elements in $K$ so that the variable $K^\prime$ contains the same elements in $K$ but with reverse order. The typical benefit is that it leads to identical intermediate models of forward and reverse proposals and the posterior probabilities of $p_k$ models are required instead of $2 p_k$. Suppose that $\gamma^\prime(r)$ for $r = 0 ,\dots, p_k$ are consecutive intermediate models used in the reverse move and $N^\prime(r)$ for $r = 0 ,\dots, p_k$ are the neighbourhoods used in the reverse move. These models and neighbourhoods are identical to models and neighbourhoods used in the proposal move but with opposite order, in particular $\gamma^{\prime}(r) = \gamma(p_k-r)$ for $r = 0 ,\dots, p_k$ and $N^\prime(r) = N(p_k-r+1)$ for $r = 1 ,\dots, p_k$. The second benefit is that the design leads to a reduced form of Metropolis-Hastings probability of acceptance. Let $Z(r)$ be the normalising constant of the $r$-th position-wise probability measure, $q^{\text{PARNI}}_{\theta,K_r}(\gamma(r-1), \gamma(r))$, and $Z^\prime(j)$ denote the normalising constant of $r$-th position-wise probability measure in reversed move, $q^{\text{PARNI}}_{\theta,K^\prime_r}(\gamma^{\prime}(r-1), \gamma^{\prime}(r))$. We have that
\begin{align}
    \begin{split}
    Z(r) = & \sum_{\gamma^* \in N(r)} g\left(\frac{\pi(\gamma^*)p^{\text{RN}}_{\eta^{\text{opt}}}(e(K_r)|\gamma^*)}{\pi(\gamma(r-1))p^{\text{RN}}_{\eta^{\text{opt}}}(e(K_r)|\gamma(j-1))}\right) q^{\text{THIN}}_{\omega, e(K_r)}(\gamma(j-1), \gamma^*) \\
    Z^\prime(r) = & \sum_{\gamma^* \in N^\prime(r)} g\left(\frac{\pi(\gamma^*)p^{\text{RN}}_{\eta^{\text{opt}}}(e(K^\prime_r)|\gamma^*)}{\pi(\gamma^\prime(r-1))p^{\text{RN}}_{\eta^{\text{opt}}}(e(K^\prime_r)|\gamma^\prime(j-1))}\right) q^{\text{THIN}}_{\omega, e(K^\prime_r)}(\gamma^\prime(j-1), \gamma^*).
    \end{split}
    \label{mat:PARNInorcon}
\end{align}
We let $q^{\text{PARNI}}_{\theta, k}(\gamma^\prime, \gamma)$ be the full proposal kernel that satisfies
\begin{align}
    q^{\text{PARNI}}_{\theta, k}(\gamma, \gamma^\prime) = \prod_{r=1}^{p_k} q^{\text{PARNI}}_{\theta, K_j}(\gamma(r-1), \gamma(r))
\end{align}
where $\gamma(0)$ is current state $\gamma$ and $\gamma(p_k)$ is the final proposal $\gamma^\prime$. The Metropolis-Hastings acceptance probability of the PARNI proposal is given as
\begin{align}
\alpha^{\text{PARNI}}_{\theta, k}(\gamma, \gamma^\prime) = \left\{ 1, \frac{\pi(\gamma^\prime)p^{\text{RN}}_{\eta^{\text{opt}}}(k|\gamma^\prime)q^{\text{PARNI}}_{\theta, k}(\gamma^\prime, \gamma)}{\pi(\gamma)p^{\text{RN}}_{\eta^{\text{opt}}}(k|\gamma)q^{\text{PARNI}}_{\theta, k}(\gamma, \gamma^\prime)} \right\}.
\label{mat:PARNIacc}
\end{align}
The reduced form of the Metropolis-Hastings acceptance probability is illustrated by the following proposition:
\begin{prop} \label{prop:PARNI_acc} 
Suppose $\gamma$, $\gamma^\prime \in \Gamma$ are fixed.  For any $\theta = (\eta, \omega) \in \Delta_\epsilon^{2p+1}$ and $k$ such that $\gamma^\prime \in N(\gamma, k)$, the Metropolis-Hastings acceptance probability in (\ref{mat:PARNIacc}) can be written
\begin{align}
        \alpha^{\text{PARNI}}_{\theta, k}(\gamma, \gamma^\prime) = \min\left\{1,  \prod_{r=1}^{p_k} \frac{Z(r)}{Z^\prime(r)} \right\} \label{mat:PARNIacc2}
\end{align}
where $Z(r)$, $Z^\prime(r)$ for $r = 1, \cdots, p_k$ are the normalising constants given in (\ref{mat:PARNInorcon}).
\end{prop}

\subsubsection{Adaptation schemes for algorithmic parameters}

The last building block to complete the PARNI sampler is the adaptation mechanism of the tuning parameters. Suppose a total number of $L$ are used. The posterior inclusion probabilities $\pi_j$ are updated as the ASI scheme in (\ref{mat:rb_pips}). The magnitude of the proposal thinning parameter $\omega$ is crucial in the mixing time and convergence rate of chains.
Therefore, we consider two adaptation schemes for updating $\omega$, the Robins-Monro adaptation scheme (RM) and the Kiefer-Wolfowitz adaptation scheme (KW). For the rest of section, we assume $L$ multiple chains are used for the PARNI sampler.

The Robbins-Monro adaptation scheme is widely used in updating tuning parameters of adaptive MCMC algorithms. \cite{andrieu2008tutorial} review
several adaptive MCMC algorithms using variants of the Robbins-Monro process. Given a specified probability of acceptance $\tau$, the Robbins-Monro adaptation scheme automatically adjusts $\omega$ according to the comparison between the current probability of acceptance and $\tau$. It is generally considered to be a robust adaption scheme. Given the acceptance probability of the $l$-th chain at the $i$-th iteration $\alpha^l_i$, the tuning parameter $\omega$ is updated through the law
\begin{align}
    \text{logit}_\epsilon\omega^{(i+1)} = \omega^{(i)} + \frac{\phi_i}{L}\sum_{l=1}^L(\alpha^l_i - \tau). \label{mat:update_omega_RM}
\end{align}
The theoretical optimal value of $\tau$ may not exist for every candidate proposal kernel and choice of posterior distribution. We recommend using a target acceptance rate of $0.65$ based on a large number of experiments that will be illustrated in Section \ref{sec:numerical_studies} and using the diminishing sequence $\phi_i = i^{-0.7}$.

The Kiefer-Wolfowitz scheme is a stochastic approximation algorithm and modification of the Robbins-Monro scheme in which a finite difference approximation to the derivative is used. In this scheme the tuning parameter is updated to target the optimiser of an objective function of interest. Following \cite{pasarica2010adaptively}, we use the Kiefer-Wolfowitz scheme to adapt $\omega$ and use the \textit{expected squared jumping distance} as the objective function. The expected squared jumping distance is closely linked to the mixing and convergence properties of a Markov chain. To estimate the expected squared jumping distance at each iteration, we use the \textit{average squared jumping distance} as an estimator. An alternative objective function can be the generalised speed measure introduced in \cite{titsias2019gradient}.

To estimate the finite difference approximation to the derivative of the average squared jumping distance, we exploit the multiple chain implementation of PARNI. The multiple independent chains naturally provide independent samples which fits the Kiefer-Wolfowitz approximation. Our implementation of the Kiefer-Wolfowitz adaption scheme proceeds as follows. We first evenly divide $L$ multiple chains into two equally sized batches, $L^+$ and $L^-$. Let $c_i$ be a diminishing sequence, new proposals are generated using $\omega^+ = \omega^{(i)} + c_i$ for chains in $L^+$ and $\omega^- = \omega^{(i)} - c_i$ for chains in $L^-$. The average squared jumping distances for these batches (\textit{i.e.} $\text{ASJD}^{+,(i)}$ and $\text{ASJD}^{-,(i)}$) are estimated using the new proposals and their corresponding probabilities of acceptance. The tuning parameter $\omega$ is then updated according to the rule
    \begin{align}
    \text{logit}_\epsilon \omega^{(i+1)} = \text{logit}_\epsilon \omega^{(i)} + a_i \left(\frac{\text{ASJD}^{+,(i)} - \text{ASJD}^{-,(i)}}{2c_i}\right). \label{mat:update_omega_KW}
\end{align}
We suggest using $a_i = i^{-1}$ and $c_i = i^{-0.5}$ in the Kiefer-Wolfowitz scheme. Further details of the Kiefer-Wolfowitz adaption scheme are given in \ref{apx:KW_scheme} of the supplementary material and a feasibility analysis of the Kiefer-Wolfowitz adaption scheme is carried out in \ref{apx:KW_feasibility} of the supplementary material. 

Pseudocode of the PARNI sampler is given in Algorithm \ref{alg:PARNI}. The corresponding transition kernel is referred to as $p_{\theta}^{\text{PARNI-$\bullet$}}$ for $\bullet = \text{RM}$ or $\text{KW}$. In the next section we show that the PARNI sampler is $\pi$-ergodic and satisfy a strong law of large numbers.

\begin{algorithm}[t]
\caption{Pointwise Adaptive Random Neighbourhood Sampler with Informed proposal (PARNI-$\bullet$)} \label{alg:PARNI}
\begin{algorithmic}
\For{$i = 1$ to $i = N$}
\State If $\bullet$ is KW, then divide $L$-chains into $L^+$ and $L^-$ and compute $\omega^+$ and $\omega^-$;
\For{$l = 1$ to $l = L$}
\State Sample $k \sim p^{\text{RN}}_{\eta^{(i)}}(\cdot|\gamma^{l,(i)})$ as in (\ref{mat:RN_prop1});
\State Set $\gamma(0) = \gamma^{(i)}$, $p_k = \sum_{j=1}^p k_j$,  $K = \{j|k_j=1\}$;
\For{$r = 1$ to $r = p_k$}
\State Sample $\gamma(r) \sim q_{\theta^{(i)},  K_j}^{\text{PARNI}}(\gamma(r-1), \cdot)$ as in (\ref{mat:PARNI_prop});
\State Calculate $Z(r)$ and $Z^\prime(p_k-r+1)$ as in (\ref{mat:PARNInorcon});
\EndFor
\State Set $\gamma^\prime = \gamma(p_k)$ Sample $U \sim U(0,1)$;
\State If $U < \alpha^{\text{PARNI}}_{\theta^{(i)}, k}(\gamma^{(i)}, \gamma^\prime)$ as in (\ref{mat:PARNIacc2}), then $\gamma^{(i+1)} = \gamma^\prime$, else $\gamma^{(i+1)} = \gamma^{(i)}$;
\State Update $\hat{\pi}^{(i+1)}_j$ as in (\ref{mat:rb_pips}), set $\tilde{\pi}^{(i+1)}_j = \pi_0 + (1-2\pi_0)\hat{\pi}^{(i+1)}_j$ for $j=1, \dots, p$;
\EndFor
\State Update $A^{(i+1)}_j = \min \left\{1, \tilde{\pi}^{(i+1)}_j/(1-\tilde{\pi}^{(i+1)}_j) \right\}$;
\State Update $D^{(i+1)}_j = \min \left\{1, (1-\tilde{\pi}^{(i+1)}_j)/\tilde{\pi}^{(i+1)}_j \right\}$;
\State If $\bullet$ is RM, then update $\omega^{(i+1)}$ as in (\ref{mat:update_omega_RM}); else if $\bullet$ is KW, then update $\omega^{(i+1)}$ as in (\ref{mat:update_omega_KW});
\State Set $\eta^{(i+1)} = (A^{(i+1)}, D^{(i+1)})$ and $\theta^{(i+1)} = (\eta^{(i+1)}, \omega^{(i+1)})$
\EndFor
\end{algorithmic}
\end{algorithm}

\subsubsection{Ergodicity and Strong Law of Large Numbers}

The multiple chain acceleration can be thought of as the realisation of $L$ runs on a product space $\Gamma^{\otimes L}$ with joint variable $\gamma^{\otimes L}  = (\gamma^{1}, \dots, \gamma^{L})  \in \Gamma^{\otimes L}$. The number of independent chains $L$ is satisfied $L \geq 1$ if Robbins-Monro adaptation is used and $L \geq 2$ if Kiefer-Wolfowitz adaptation is used. We consider a posterior distribution $\pi$ on the space $\Gamma$ which is of the form
\begin{align}
    \pi(\gamma) \propto p(y|\gamma) p(\gamma) \label{mat:target_post}
\end{align} 
where both $p(y|\gamma)$ and $p(\gamma)$ are analytically available. In addition, the joint posterior distribution $\pi^{\otimes L}$ on the product set $\Gamma^{\otimes L}$ is given as
\begin{align}
    \pi^{\otimes L}(\gamma^{\otimes L}) = \prod_{l=1}^L \pi(\gamma^{l}).
\end{align} 
In this section, the symbol $\bullet$ represents either KW or RM. The sub-proposal mass function of the $\text{PARNI}-\bullet$ sampler given neighbourhood indicator variable $k$ and tuning parameter $\theta = (\eta, \omega)$ is defined by
\begin{align}
     \psi^{\text{PARNI}-\bullet}_{\theta, k}(\gamma, \gamma^\prime)= p^{\text{RN}}_\eta(k|\gamma)q^{\text{PARNI}-\bullet}_{\theta, k}(\gamma, \gamma^\prime). \label{mat:PARNI_subpropkernel}
\end{align}
The full transition kernel of the PARNI sampler is marginalised over all possible $k$
\begin{align}
    P^{\text{PARNI}-\bullet}_{\theta}(\gamma, S) = \sum_{k \in \mathcal{K}} P^{\text{PARNI}-\bullet}_{\theta, k}(\gamma, S) \label{mat:PARNI_trankernel}
\end{align} 
where the sub-transition kernels given $k$ are
\begin{align}
    P^{\text{PARNI}-\bullet}_{(\theta, k)}(\gamma, S) = & \sum_{\gamma^\prime \in S}  p^{\text{PARNI}-\bullet}_{\theta, k}(\gamma, \gamma^\prime) \notag \\
    = & \sum_{\gamma^\prime \in S} \psi^{\text{PARNI}-\bullet}_{\theta, k}(\gamma, \gamma^\prime) \alpha^{\text{PARNI}-\bullet}_{\theta, k}(\gamma, \gamma^\prime) \notag  \\
    & + \mathbb{I}\{\gamma \in S\} \sum_{\gamma^\prime \in \Gamma} \psi^{\text{PARNI}-\bullet}_{\theta, k}(\gamma, \gamma^\prime)(1- \alpha^{\text{PARNI}-\bullet}_{\theta, k}(\gamma, \gamma^\prime)) \label{mat:PARNI_sub_trankernel}
\end{align}
and $\alpha^{\text{PARNI}}_{\theta, k}$ are Metropolis-Hastings acceptance rates in (\ref{mat:PARNIacc2}). The Markov chain transitional kernel that works on the product space $\Gamma^{\otimes L}$ is given as
\begin{align}
    P^{\text{PARNI}-\bullet}_{(\theta, k^{\otimes L})}(\gamma^{\otimes L}, S^{\otimes L}) = \prod_{l=1}^L P^{\text{PARNI}-\bullet}_{(\theta, k^{l})}(\gamma^{l}, S^{l}).
\end{align}

To establish the ergodicity and a SLLN of the PARNI sampler and its multiple chain acceleration, we require the following assumptions:
\begin{enumerate}[label=(A.\arabic*)]
\item For some small constant $\varepsilon > 0$ and constant $C_g$, the balancing function $g$ satisfies
\begin{align}
    g(t_2) - g(t_1) \leq C_g (t_2 - t_1) \label{mat:bal_fun_inequality}
\end{align}
where $\varepsilon < t_1 < t_2$. This is a common condition for the proper choice of balancing functions. For example, Hastings' choice $g_\text{H}(t) = \min\{1,t\}$ follows (\ref{mat:bal_fun_inequality}) immediately for $C_g = 1$ and Barker's choice $g_\text{B} = t/(1+t)$ also follows (\ref{mat:bal_fun_inequality}) when $C_g = 1$ (\textit{i.e.} the maximum derivative).
\item The posterior distribution $\pi$ is everywhere positive and bounded, that is, there exists a positive $\Pi \in (1, \infty)$ such that
\begin{align*}
    \frac{1}{\Pi} \leq \frac{\pi(\gamma^\prime)}{\pi(\gamma)} \leq \Pi
\end{align*}
for all $\gamma$, $\gamma^\prime \in \Gamma$.
\item The tuning parameters $\theta^{(i)} = (\eta^{(i)}, \omega^{(i)})$ are bounded away from 0 and 1, and lie in the set
    \begin{align}
        \theta^{(i)} \in \Delta_\epsilon^{2p+1} \label{mat:delta_epsilon}
    \end{align}
    for some small $\epsilon \in (0, 1/2)$.
\end{enumerate}

The analysis of convergence and ergodicity often relies on the distribution of the Markov chain at time $i$ along with its associated total variation distance $\|\cdot\|_{TV}$ at an arbitrary starting point. Given $\{\gamma^{l,(i)}\}_{i=0}^\infty$ these are defined as
\begin{align}
    \mathcal{L}^{l,(i)} [(\gamma^{l}, \theta), S] := \Pr\left[ \gamma^{l, (i)} \in S| \gamma^{l, (0)} = \gamma^{l}, \theta^{0} = \theta \right], \\
    \lim_{i \to \infty} T^l(\gamma^l, \theta, i) := \| \mathcal{L}^{l,(i)} [(\gamma^{l}, \theta), \cdot] - \pi(\cdot) \|_{TV}.
\end{align}
We show here that the PARNI sampler is ergodic and satisfies a strong law of large numbers (SLLN). In mathematical terms for any starting point $\gamma^{\otimes L} \in \Gamma^{\otimes L}$ and $\theta \in \Delta_\epsilon^{2p+1}$ ergodicity means that
\begin{equation}
 \lim_{i \to \infty} T^l(\gamma^l, \theta, i) \to 0, \quad  \text{} \label{mat:PARNI_ergodicity}
\end{equation}
for any $l = 1, \dots, L$, while a strong law of large numbers (SLLN) implies that
\begin{equation}
\frac{1}{NL}   \sum_{i=0}^{N-1} \sum_{l=1}^L f(\gamma^{l,(i)}) \to \pi(f) \label{mat:PARNI_SLLN}
\end{equation}
almost surely, for any $f:\Gamma \to \mathbb{R}$.  We first establish two technical results before presenting the main theorem of this section.

\begin{lemma}[Simultaneous Uniform Ergodicity]
\label{lemma:PARNI_simerg}
The MCMC transition kernel $P_\theta^{\text{PARNI}-\bullet}$ in (\ref{mat:PARNI_trankernel}) with target distribution $\pi$ in (\ref{mat:target_post}) is simultaneously uniformly ergodic for any choice of $\epsilon \in (0,1/2)$ in (\ref{mat:delta_epsilon}). i.e. for any $\delta>0$, there exists $N = N(\delta, \epsilon)$ such that 
\begin{align*}
    \left\|\left(P^{\text{PARNI}-\bullet}_\theta(\gamma^{\otimes L},\cdot)\right)^N - \pi^{\otimes L}(\cdot)\right\|_{TV} \leq \delta
\end{align*}
holds for any any starting point $\gamma^{\otimes L} \in \Gamma^{\otimes L}$ and any value $\theta \in \Delta_\epsilon^{2p+1}$.
\end{lemma}

\begin{lemma}[Diminishing adaptation]
\label{lemma:PARNI_dimadap} 
Let $\lambda = 0.5$ be the constant of diminishing rate, for any $\epsilon \in(0, 1/2)$ and $\pi_0 \in (0, 1)$, the PARNI sampler satisfies diminishing adaptation, that is, its transition kernel satisfies
\begin{align}
    \sup_{\gamma \in \Gamma} \left\|P^{\text{PARNI}-\bullet}_{\theta^{(i+1)}}(\gamma, \cdot) - P^{\text{PARNI}-\bullet}_{\theta^{(i)}}(\gamma, \cdot) \right\|_{TV} \leq C i^{-\lambda} \label{mat:PARNI_dimadap_tv}
\end{align}
for some constant $C < \infty$.
\end{lemma}

\begin{theorem}[Ergodicity and SLLN]  \label{theo:PARNI_ergslln}
Consider a target distribution $\pi(\gamma)$ in (\ref{mat:target_post}), tuning rate $\lambda = 0.5$ and $\epsilon \in (0,1/2)$ that lead to a diminishing rate $\mathcal{O}(i^{-\lambda})$, and the parameter $\pi_0 > 0$ in Algorithm \ref{alg:PARNI}. Then ergodicity (\ref{mat:PARNI_ergodicity}) and a strong law of large numbers (\ref{mat:PARNI_SLLN}) hold for both the $\text{PARNI-KW}$ and $\text{PARNI-RM}$ samplers as described in Algorithm \ref{alg:PARNI} and its corresponding multiple chain acceleration versions.
\end{theorem}

\section{Numerical studies}
\label{sec:numerical_studies}

\subsection{Simulated data}
\label{subsec:simulated_dataset}

We consider the data generation model introduced by \cite{yang2016computational}, and replicated in simulation studies conducted by \cite{griffin2021search} and \cite{zanella2019scalable}. Suppose a linear model with $n$ observations and $p$ covariates is needed, data are generated from the model specification 
\begin{align*}
    y = X^* \beta^* + \epsilon
\end{align*}
where $\epsilon \sim N_n(0, \sigma^2 I_n)$ for pre-specified residual variance $\sigma^2$ and $\beta^* = \text{SNR}\times \Tilde{\beta}\sqrt{(\sigma^2\log p)/n}$ in which $\text{SNR}$ represents the signal-to-noise ratios. Let $\Tilde{\beta} = (2, -3, 2,2,-3,3,-2,3,-2,3,0,\cdots,0)$ and each row of the design matrix $X^*_i$ follow a multivariate normal distribution with mean zero and covariance $\Sigma$ with entries $\Sigma_{j j} = 1$ for all $j$ and $\Sigma_{ij} = 0.6^{|i-j|}$ for $i \ne j$. We consider four choices of SNR, namely 0.5, 1, 2 and 3, two choices of $n$, namely 500 and 1,000 and three choices of $p$, namely 500, 5,000 and 50,000. 

We use the prior parameter values $V_\gamma = I_{p_\gamma}$, $g = 9$ and $h = 10/p$. \cite{griffin2021search} give a detailed description of the resulting posterior distributions. In the presence of a low SNR ($\text{SNR} = 0.5$), there is too much noise to detect the true non-zero variables and the resulting posterior is rather flat, with no variables having posterior inclusion probabilities larger than $0.1$. The posterior distributions are completely different when the SNR is large ($\text{SNR} = 2$ and $\text{SNR} = 3$). In these cases all of the true non-zero variables have inclusion probabilities close to $1$ as the posterior distributions are more concentrated. In the intermediate case $\text{SNR} = 1$ slightly less than half of the true non-zero variables have inclusion probabilities above $0.8$. In general the problem of finding the true non-zero variables becomes more difficult in the cases with lower SNR, smaller $n$ and larger $p$. 

We are interested in comparing the performance of the ASI and PARNI schemes relative to an Add-Delete-Swap sampler because the ASI scheme has been compared with several other state-of-the-art MCMC algorithms in \cite{griffin2021search}. The adaptive algorithms are run with 25 multiple chains. In addition, to reduce the computational budget, all the adaptations terminate after the period of burn-in.

The choice of balancing function mainly focuses on three particular candidates: square root function $g_{\text{sq}} (t) = \sqrt{t}$, Hastings' choice $g_{\text{H}}(t)=\min\{1,t\}$ and Barker's choice $g_{\text{B}} = t/(1+t)$. The comparisons of these balancing functions in Supplement B.1.3 of \cite{zanella2020informed} illustrate two major findings. The Hastings' and Barker's choices only differ by at most a factor of 2 due to their similar asymptotic behaviors. The square root function mixes worsen outside the burn-in phase. Therefore, we consider the Hastings' choice throughout the section. Similar results are also expected for the Barker's choice.

Trace plots of chains are a straightforward way to visualise convergence. Fig. \ref{fig:yang_trace} are the trace plots of posterior model probabilities from the Add-Delete-Swap, ASI, PARNI-KW and PARNI-RM algorithms for the first 1,500 iterations when the SNR = 2. The Add-Delete-Swap scheme fails to converge for all choices of $n$ and $p$ and in particular becomes trapped at the empty model for a long period of time when $p=50,000$. The ASI scheme converges reasonably quickly when $p$ is $500$ or $5,000$, but takes longer to reach high probability regions when $p = 50,000$. This suggests that ASI mixes worse and converges slower in high-dimensional datasets. On the other hand, both the PARNI-KW and PARNI-RM samplers mix rapidly in this setting.

The trace plots are not truly a fair comparison as they do not take into account running time. To better address the issue of computational efficiency we repeatedly ran all of the algorithms for 15 minutes and stored the estimates of posterior inclusion probabilities. We calculated mean squared errors of these estimates compared to ``gold standard'' estimates taken from a weighted tempered Gibbs sampler that has ran for any extremely long time. We show results in the form of performance relative to the Add-Delete-Swap scheme in Table \ref{table:yang_mse}. Smaller values always indicate better performance of the scheme. The value of -1 indicates the scheme yields $10$ times smaller mean squared errors compared to those from the Add-Delete-Swap scheme in this specific dataset. Generally speaking, the mean squared errors for important variables are greater than those for important variables for almost every dataset and scheme. The choice of $n$ does not significantly affect the performance of the samplers.  Concentrating on the results for important variables, the ASI scheme leads to an order of magnitude improvement in efficiency over the Add-Delete-Swap sampler, which match the results in \cite{griffin2021search}. The two PARNI algorithms with different adaptations lead to similar levels of accuracy and dominate both the ASI and Add-Delete-Swap schemes in every case when $p>500$. In particular, the PARNI schemes result in roughly $10^5$ times improvements over Add-Delete-Swap and $10$ times improvements over ASI when $p=50,000$ and SNR$ = 2$. On the other hand, the Add-Delete-Swap scheme is quite adept at removing the unimportant variables when the true model size is small compared to the number of covariates. When $p = 50,000$ and SNR$>1$ the ASI scheme struggles with unimportant variables and has yet to give estimates better than Add-Delete-Swap, but the PARNI algorithms produce better estimates even for these unimportant variables. Overall, the results suggest both PARNI samplers are more computationally efficient than alternatives when $p$ is large. More results from simulated data are provided in  Section \ref{apx:yang_results} of the supplementary material.

\subsection{Real data}
\label{subsec:resl_dataset}

We consider eight real datasets implemented in \cite{griffin2021search}, four of them with moderate $p$ and four with larger $p$. 

The first dataset is the Tecator dataset, which is previously analysed by \cite{brown2010inference} in Bayesian linear regression and implemented by \cite{lamnisos2013adaptive} and  \cite{griffin2021search} in the context of Bayesian variable selection. It contains 172 observations and 100 explanatory variables. We also consider three small-$p$ data sets constructed by \cite{schafer2013sequential} to illustrate the performance of sequential Monte Carlo algorithms on Bayesian variable selection problems, the Boston Housing data $(n = 506, p=104)$, the Concrete data $(n = 1030, p = 79)$ and the Protein data $(n = 96, p = 88)$.  These data sets are extended by squared and interaction terms which lead to high dependencies and multicollinearity.

The last four data sets are high-dimensional problems with very large-$p$. Three of them come from an experiment conducted by \cite{lan2006combined} to examine the genetics of two inbred mouse populations. The experiment resulted in a set of data with 60 observations in total that were used to monitor the expression levels of $22,575$ genes of $31$ female and $29$ male mice. \cite{bondell2012consistent} first considered this dataset in the context of variable selection. Three physiological phenotypes are also measured by quantitative real-time  polymerase chain
reaction (PCR), they are used as possible responses and are named $\text{PCR}_i$ for $i=1,2,3$ respectively. For more details, see \cite{lan2006combined,bondell2012consistent}. The last dataset concerns genome-wide mapping of a complex trait. The data are illustrated in \cite{carbonetto2017varbvs}. They are body and testis weight measurements recorded for $993$ outbred mice, and genotypes at $79,748$ single nucleotide polymorphisms (SNPs) for the same mice. The main purpose of the study is to identify genetic variants contributing to variation in testis weight. Thus, we consider the testis weight as response, the body weight as a regressor that is always included in the model and
variable selection is performed on the 79,748 SNPs.

Before analysing the performance of MCMC algorithms on the above datasets, it is worth discussing the selection of an optimal acceptance rate for the PARNI-RM sampler. The optimal scaling property of a Gaussian random walk proposal on some specific forms of target distribution is a well-studied problem. The most commonly used guideline is to seek an average acceptance rate of 0.234 \citep{gelman1997weak}. The optimal acceptance rates for sophisticated \textit{informed} proposals involving gradient information are typically larger, \textit{e.g.} 0.57 for the \textit{Metropolis-adjusted Langevin algorithm} \citep{grenander1994representations,roberts1998optimal} and 0.65 for \textit{Hamiltonian Monte Carlo} \citep{duane1987hybrid, beskos2013optimal}. As our balanced random neighbourhood proposals can be viewed as a discrete analog to these gradient-based algorithms, it is natural to think that the PANRI sampler will have a larger optimal acceptance rate than a random walk Metropolis. To test this we ran the PARNI-RM scheme targeting different rates of acceptance on the above datasets. Fig. \ref{fig:opt_acc_rate} shows the effect of the average acceptance rate on the expected squared jumping distance and average mean squared errors of the sampler. Parts (a) and (b) of the figure illustrate the relation between the thinning parameter $\omega$ and the average acceptance rate. Bigger values of $\omega$ are synonymous with larger jumps and therefore can lead to a smaller average acceptance rate.  Parts (c) and (d) of the figure suggest that the maximum average squared jumping distance occurs when the acceptance rate is around $0.65$ for all datasets. Parts (e) and (f) show that the average mean squared error is minimised when acceptance is around the same value. Therefore, we recommend targeting an average acceptance rate of $0.65$ in most problems. Similar results for the simulated datasets of \ref{subsec:simulated_dataset} are presented in \ref{apx:yang_omega} of the supplementary material. We stress that the PARNI-KW scheme does not require a target acceptance rate to be chosen, so users who are uncomfortable with having to choose this quantity for a particular dataset are recommended to use this version of the sampler.

We consider a total of seven different MCMC schemes for these sets of data. In addition to the four schemes used in the simulation study (Add-Delete-Swap, ASI, PARNI-KW and PARNI-RM), we also implement 3 state-of-the-art algorithms, the Hamming ball sampler (HBS) with radius of 1 of \cite{titsias2017hamming}, and also both the tempered Gibbs sampler (TGS) and weighted tempered Gibbs sampler (WTGS) of \cite{zanella2019scalable}. The prior specification for each dataset is given in Table \ref{table:real_prior}.

Fig. \ref{fig:real_trace} shows trace plots of posterior model probabilities from the Add-Delete-Swap, ASI, PARNI-KW and PARNI-RM algorithms for the first 1,500 iterations in all eight real datasets. It is clear that the Add-Delete-Swap scheme does not mix well since it struggles to explore model space.   All algorithms do reach high probability regions for datasets with moderate $p$ in roughly the same number of iterations, however both the PARNI schemes appear to mix better than both Add-Delete-Swap and ASI and explore more models around the model space. In the large-$p$ datasets, these algorithms lead to different behaviour.  The Add-Delete-Swap scheme gets trapped in the empty model whereas the ASI algorithm does not converge properly for the first 1,500 iterations. The PARNI schemes, by contrast, accept almost every proposed states and mix very quickly. The figure suggests that both the PARNI schemes mix no worse than Add-Delete-Swap and ASI in moderate $p$ problems, and better in large $p$ problems.

We next turn attention to the average mean squared errors on these eight real datasets. The PARNI samplers do not dominate other schemes in moderate-$p$ datasets, but they still lead to impressive results that are not close to optimal. The worst performance is for the Boston Housing and Concrete datasets, which are multi-modal and contain intricately correlated covariates. For large-$p$ problems both the PARNI schemes significantly outperform other samplers. Surprisingly, the HBS and TGS schemes lead to worse estimates than Add-Delete-Swap. This can be explained by the computational cost per iteration of the HBS, TGS and WTGS algorithms, which is linear in $p$. These large computational costs combined with the issue of rarely exploring important variables lead to low efficiencies for HBS and TGS. The WTGS algorithm still outperforms TGS, which coincides with the conclusions gathered in \cite{zanella2019scalable} where the WTGS algorithm is shown to have smaller relaxation time than TGS. The ASI algorithm gives competitive estimates to WTGS in high-dimension but is eventually outperformed by the PARNI schemes. Among the PARNI schemes, the PARNI sampler with Kiefer-Wolfowitz adaption generally performs better than the Robbins-Monro version, but only by a small margin. This is due to the fact that the optimal acceptance rates are problem-specific and not exactly 0.65  for every dataset.

\section{Discussion and future work}
\label{sec:discussion_future}

In this paper we present a framework for neighbourhood based MCMC algorithms, and propose a new scheme as an \textit{informed} counterpart to the ASI algorithm in \cite{griffin2021search}, using elements from locally balanced Metropolis--Hastings introduced in \cite{zanella2020informed}. To address the expensive computational costs introduced by the informed proposal, we introduce two less computationally costly algorithms, the PARNI schemes, which can lead to a dramatic improvement in performance. The two PARNI schemes employ different adaptation schemes for the thinning parameter, the Kiefer-Wolfowitz and Robbins Moron schemes. The success of these new samplers is attributed to two aspects. Firstly the adaptation helps to explore the areas of interest (mainly with high posterior probabilities), and secondly the locally informed proposals are able to stabilise random walk behaviour in high-dimensions and lead to rapidly mixing samplers in practice. From the numerical studies on both simulated and real datasets, we recommend using a PARNI sampler with the Kiefer-Wolfowitz scheme for tackling high-dimensional (or large-$p$) Bayesian variable selection problems. We note that it can still be challenging for the PARNI samplers to move across low probabilistic regions, which could affect performance when the posterior has very  isolated modes. This phenomenon is due to the fact that the PARNI samplers propose models sequentially where each sub-proposal can alter only 1 position at most. On the other hand, the original ARNI scheme can take larger jumps and is more able to explore well-separated modes, albeit with a substantial increase in computational costs. In summary, new schemes like PARNI show the potential of combining adaptive, random neighbourhood and informed proposals. We look forward to adding more theoretical support to the numerical evidence shown here in future work. In addition, the code to run the PARNI samplers and aforementioned numerical studies is given in
\url{https://github.com/XitongLiang/The-PARNI-scheme.git}.

The paper provides many directions for extensions and future work. Some recent work has shed light on the issues of extra computational costs that come with informed proposals. \cite{grathwohl2021oops} develop an accelerated locally informed proposal that uses derivatives with respect to the log mass functions. It is possible to derive the gradient of the posterior mass function with respect to $\gamma$ with minor modifications on representations of the posterior distribution $\pi(\gamma)$. To address the lack of mode jumping in the PARNI schemes, we can first try to construct larger blocks intelligently so that separated models are covered in one single block. This solution can be achieved by introducing basis vectors beyond the Cartesian case in the block construction. One can also use the sequential Monte Carlo methods of \cite{ji2013adaptive}, \cite{schafer2013sequential} and \cite{ma2015scalable}, which are more able to handle multimodality. Combining them with PARNI yields the chance of producing efficient methods on highly multimodal posterior distributions with well-separated modes. Another option in this direction is the JAMS algorithm of \cite{pompe2020framework} that first locates each individual mode and then produces a mixture proposal that involves jumps within and between modes.

We intend to also study the performance of the PARNI schemes in generalised linear models as in \cite{wan2021adaptive} or a more flexible Bayesian variable selection model such as that suggested by \cite{rossell2018tractable}. In these cases, regression coefficients and residual variance are no longer integrated out analytically and the likelihood of $\gamma$ is not available in closed form. Informed proposals for such models are computationally challenging because the proposals involve the evaluations of these likelihood but the required approximations and estimates of the marginal likelihood are computationally intensive. One possible approach is the data-augmentation method using the P\'{o}lya-gamma distribution as described in \cite{polson2013bayesian}. The design does however require some care to avoid inefficiency causing by introducing a large number of auxiliary variables in large-$n$ problems. We also believe that their is potential to use random neighbourhood samplers beyond variable selection, and aim to consider to other discrete-valued sampling problems in future work.

\section*{Acknowledgements}

XL thanks Dr. Krzysztof {\L}atuszy{\'n}ski for the discussion and help for the proof of Lemma \ref{lemma:PARNI_dimadap} and related results.

\begin{figure}[t]
    \centering
    \begin{tikzpicture}
            \draw  plot[smooth, tension=.8] coordinates {(-6.5, 0) (-5,3.5) (0, 4.6) (5.5,3.7) (5.7, -1) (5,-3.5) (1, -3.2) (-1, -3.3)  (-5.4,-3.5) (-6.5, 0)};
            \draw [red] (-3.7,-2) ellipse (2.3cm and 1.4cm);
            \draw [rotate=-10, color=red] (-2.5,-0.5) ellipse (1.4cm and 2.3cm);
            \draw[rotate=50, color=red] (0,3) ellipse (2.3cm and 1.4cm);
            \draw[rotate=-60, color=red] (-0.5,3.1) ellipse (2.3cm and 1.4cm);
            \draw[rotate=-70, color=red] (1.6,3.3) ellipse (2.3cm and 1.4cm);
            \draw[fill, orange] (-5,-2) circle [radius=0.07];
            \node[below, orange] at (-5,-2) {\Large $\gamma = \gamma(0)$};
            \draw[fill, blue] (-3,-1.5) circle [radius=0.07];
            \node[below right, blue] at (-3,-1.5) {\Large $\gamma(1)$};
            \draw[fill, blue] (-2.5,1.3) circle [radius=0.07];
            \node[above left, blue] at (-2.5,1.3) {\Large $\gamma(2)$};
            \node[blue] at (0.2,3) {\Large $\cdots$};
            \draw[fill, blue] (3,1) circle [radius=0.07];
            \node[above right, blue] at (3,1) {\Large $\gamma(p_k-1)$};
            \draw[fill, purple] (4,-1.3) circle [radius=0.07];
            \node[below, purple] at (4,-1.3) {\Large $\gamma(p_k) = \gamma^\prime$};
            \draw [->, thick, cyan] (-4.8,
            -1.95) -- (-3.2,-1.55);
            \draw [->, thick, cyan] (-2.95,
            -1.22) -- (-2.55,1.02);
            \draw [->, thick, cyan] (-2.27,1.47) -- (-0.43,3);
            \draw [->, thick, cyan] (0.84,3) -- (2.76,1.2);
            \draw [->, thick, cyan] (3.1,0.77) -- (3.9,-1.07);
            \node[below, rotate=35] at (-5,3) {\Large $N(\gamma, k)$};
            \node[above, red] at (-5.5,-1) {\Large $N(1)$};
            \node[left, red] at (-4,0.5) {\Large $N(2)$};
            \node[red] at (4.2,3) {\Large $N(p_k-1)$};
            \node[below, red] at (4,-2.5) {\Large $N(p_k)$};
    \end{tikzpicture}
    \begin{tikzpicture}
            \draw  plot[smooth, tension=.8] coordinates {(-6.5, 0) (-5,3.5) (0, 4.6) (5.5,3.7) (5.7, -1) (5,-3.5) (1, -3.2) (-1, -3.3)  (-5.4,-3.5) (-6.5, 0)};
            \draw [red] (-3.7,-2) ellipse (2.3cm and 1.4cm);
            \draw [rotate=-10, color=red] (-2.5,-0.5) ellipse (1.4cm and 2.3cm);
            \draw[rotate=50, color=red] (0,3) ellipse (2.3cm and 1.4cm);
            \draw[rotate=-60, color=red] (-0.5,3.1) ellipse (2.3cm and 1.4cm);
            \draw[rotate=-70, color=red] (1.6,3.3) ellipse (2.3cm and 1.4cm);
            \draw[fill, orange] (-5,-2) circle [radius=0.07];
            \node[below, orange] at (-5,-2) {\Large $\gamma^{\prime}(p_k) = \gamma$};
            \draw[fill, blue] (-3,-1.5) circle [radius=0.07];
            \node[below right, blue] at (-3,-1.5) {\Large $\gamma^{\prime}(p_k-1)$};
            \draw[fill, blue] (-2.5,1.3) circle [radius=0.07];
            \node[above left, blue] at (-2.5,1.3) {\Large $\gamma^{\prime}(p_k-2)$};
            \node[blue] at (0.2,3) {\Large $\cdots$};
            \draw[fill, blue] (3,1) circle [radius=0.07];
            \node[above right, blue] at (3,1) {\Large $\gamma^{\prime}(1)$};
            \draw[fill, purple] (4,-1.3) circle [radius=0.07];
            \node[below, purple] at (4,-1.3) {\Large $\gamma^\prime = \gamma^{\prime}(0)$};
            \draw [<-, thick, cyan] (-4.8,
            -1.95) -- (-3.2,-1.55);
            \draw [<-, thick, cyan] (-2.95,
            -1.22) -- (-2.55,1.02);
            \draw [<-, thick, cyan] (-2.27,1.47) -- (-0.43,3);
            \draw [<-, thick, cyan] (0.84,3) -- (2.76,1.2);
            \draw [<-, thick, cyan] (3.1,0.77) -- (3.9,-1.07);
            \node[below, rotate=35] at (-5,3) {\Large $N(\gamma^\prime, k)$};
            \node[above, red] at (-5.5,-1) {\Large $N^\prime(p_k)$};
            \node[left, red] at (-4,0.5) {\Large $N^\prime(p_\gamma-1)$};
            \node[red] at (4.2,3) {\Large $N^\prime(2)$};
            \node[below, red] at (4,-2.5) {\Large $N^\prime(1)$};
    \end{tikzpicture}
    \caption{Flowcharts of the Pointwise implementation of Adaptive Random Neighbourhood Informed proposal in one iteration. Top panel: proposed direction. Bottom panel: reversed direction. The black neighbourhoods $N(\gamma, k)$ and $N(\gamma^\prime, k)$ are the original large neighbourhoods. The red neighbourhoods $N(r)$ and $N^\prime(r)$ are subsequent small neighbourhoods used for each intermediate proposals. The orange model $\gamma$ is the current state and the cerise model $\gamma^\prime$ is the final proposal. The blue models $\gamma(r)$ and $\gamma^{\prime}(r)$ are intermediate models. The light blue arrows indicate the position-wise proposals.}
    \label{fig:PARNI_flowchart}
\end{figure}
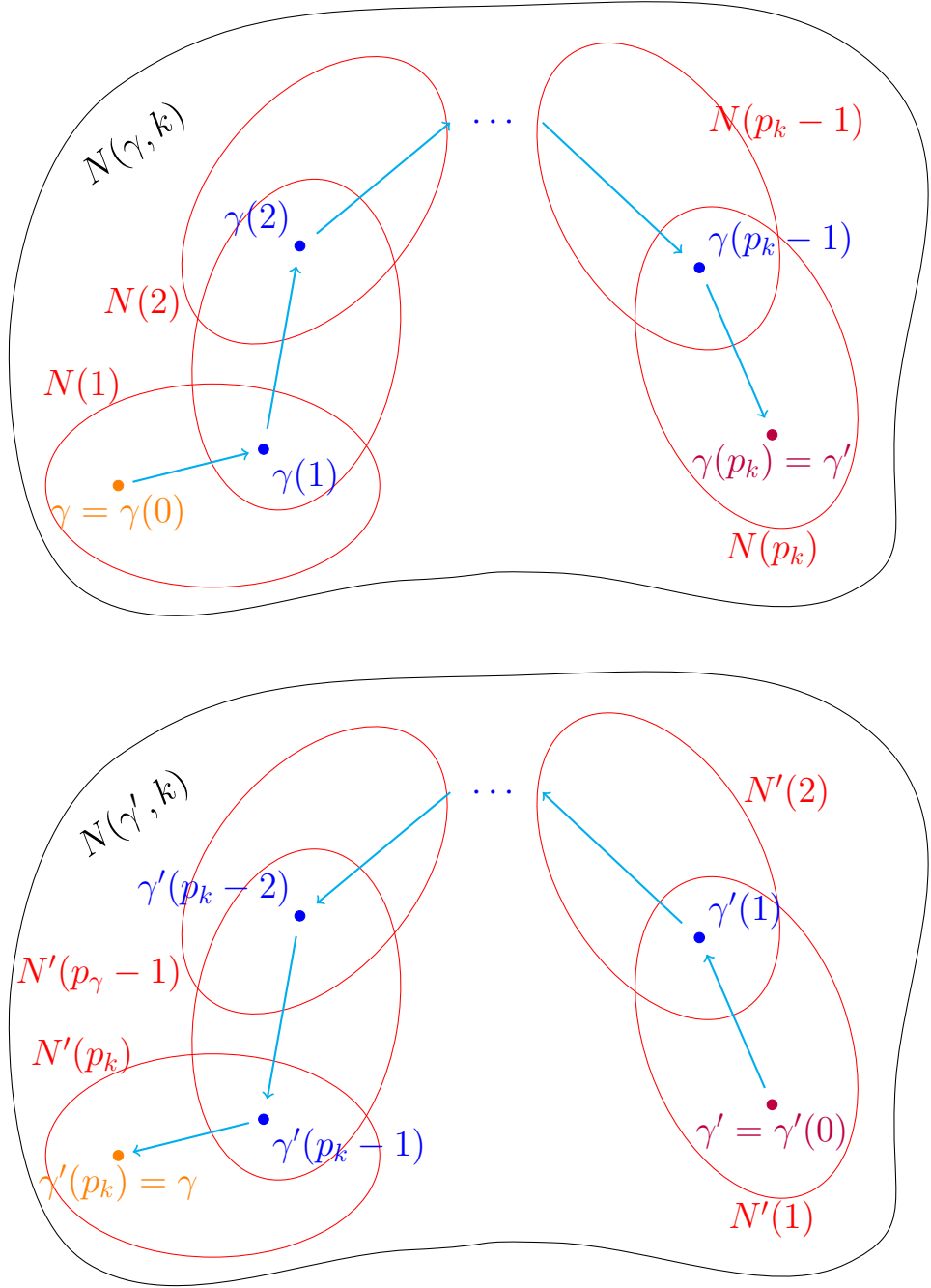

\begin{figure}[!p]
\includegraphics[width=\columnwidth]{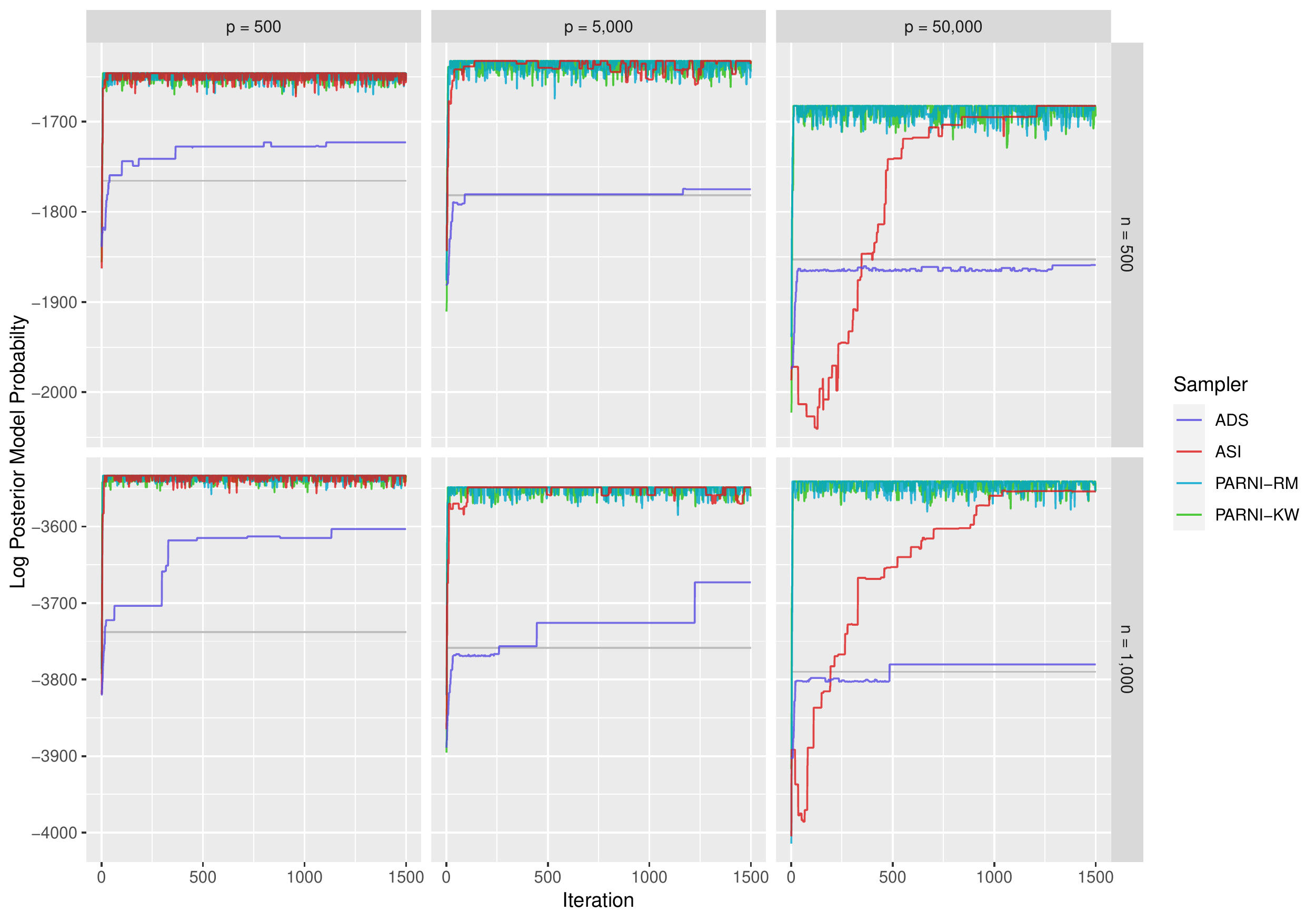}
\centering
\caption{Simulated data: trace plots of log posterior model probability from the Add-Delete-Swap (ADS), Adaptively Scaled Individual (ASI) adaptation, Pointwise implementation of Adaptive Random Neighbourhood Informed proposal with Robbins-Monro update (PARNI-RM) and Pointwise implementation of Adaptive Random Neighbourhood Informed proposal with Kiefer-Wolfowitz update (PARNI-KW) samplers for the first 1,500 iterations on simulated datasets with signal-to-noise ratio of 2.}
\label{fig:yang_trace}
\end{figure}

\begin{table}[!p]
\centering
\renewcommand{\arraystretch}{1.5}
\begin{tabular}{llcccc}
\hline\hline
 $(n,p)$ & Samplers & \multicolumn{4}{c}{SNR}      \\ 
\cline{3-6}
 &   &   0.5   &  1  &   2  &  3        \\   
 \hline $(500, 500)$  
  &   ASI   &    -4.08(\textbf{-2.59})    &  -3.94(-1.77)  & -4.84(\textbf{-2.93})  &  \textbf{-4.18}(-3.06)  \\
      &  PARNI-KW   &    \textbf{-4.12}(-2.44)    &  \textbf{-4.02}(-2.88)  &  \textbf{-5.04}(-2.81)  &  -3.90(\textbf{-3.07})  \\
      &  PARNI-RM   &    -3.90(-2.50)    &  -3.91(\textbf{-2.89})  &  -4.85(-2.74)  &  -3.67(-2.99)  \\
\hline $(1000, 500)$  
  &   ASI   &   -4.04(-3.69)     &  -1.51(-1.20)  & -4.11(\textbf{-2.84})   &  -3.89(\textbf{-2.88})  \\
    &  PARNI-KW   &    -4.01(\textbf{-3.74})    &  \textbf{-3.18}(\textbf{-2.86})  &  -3.85(-2.80)   &  \textbf{-4.02}(-2.78)  \\
      &  PARNI-RM   &   \textbf{-4.14}(-3.69)     &  -3.05(-2.53)  &  \textbf{-4.12}(-2.77)   &  -3.87(-2.72)  \\
      \hline $(500, 5000)$  
  &   ASI   &    -3.32(-1.46)    &  -3.00(-0.96)  & -4.52(-1.35)  &  \textbf{-3.5}(-1.36)  \\
    &  PARNI-KW   &    -3.62(\textbf{-1.93})    &  \textbf{-3.67}(-1.02)  &  \textbf{-5.09}(-1.74)  &  -3.46(-1.65)  \\
      &  PARNI-RM   &    \textbf{-4.01}(-1.91)    &  \textbf{-3.66}(\textbf{-1.10})  &  -4.98(\textbf{-1.79})  &  -3.38(\textbf{-1.69})  \\
      \hline $(1000, 5000)$   
  &   ASI   &    -3.65(-1.82)    &   -1.45(-1.23)   &  -6.18(-1.25)   &  -2.82(-1.16)  \\
    &  PARNI-KW   &    -3.97(-2.43)    &  \textbf{-1.76}(-1.68)  &  \textbf{-7.28}(-1.62)  &  \textbf{-3.95}(-1.51)  \\
      &  PARNI-RM   &    \textbf{-4.64}(\textbf{-2.46})    &  -1.73(\textbf{-1.76})  &   -5.86(\textbf{-1.63})   &  -3.33(\textbf{-1.56})  \\
      \hline $(500, 50000)$   
  &   ASI   &    -0.95(0.71)    &  -3.22(0.78)  & -4.30(0.93)  &  -3.70(1.07)  \\
    &  PARNI-KW   &    -1.90(\textbf{-0.59})    &  \textbf{-5.24}(\textbf{-0.45})  &  \textbf{-5.68}(\textbf{-0.50}) &  -5.26(-0.51)  \\
      &  PARNI-RM   &   \textbf{-2.35}(\textbf{-0.60})     &  -5.18(\textbf{-0.46})  &  -5.54(-0.46)  &  \textbf{-5.38}(\textbf{-0.56})  \\
      \hline $(1000, 50000)$   
  &   ASI   &    -1.33(-0.28)    &  -3.28(0.90)  &  -4.40(1.93)   &  -2.90(2.33)  \\
    &  PARNI-KW   &    -1.89(\textbf{-1.49})    &  -4.36(\textbf{-0.20})  &  \textbf{-5.24}(-0.35) &  -4.85(-0.24)  \\
      &  PARNI-RM   &    \textbf{-2.41}(-1.48)    &  \textbf{-4.47}(-0.17)  &  -5.07(\textbf{-0.47})  &  \textbf{-5.07}(\textbf{-0.27})  \\
\hline\hline
\end{tabular}
\caption{Simulated data: relative average mean squared errors for the Adaptively scaled individual (ASI), Pointwise implementation of Adaptive Random Neighbourhood Informed proposal with Kiefer-Wolfowitz update (PARNI-KW) and Pointwise implementation of Adaptive Random Neighbourhood Informed proposal with Robbins-Monro update (PARNI-RM) schemes on estimating posterior inclusion probabilities over important and unimportant variables respectively against a standard Add-Delete-Swap algorithm. The quantities outside the brackets are for important variables which have posterior inclusion probabilities greater than 0.01 whereas the quantities inside the brackets are for unimportant variables which have posterior inclusion probabilities less than 0.01. Values are presented in logarithm to base 10. Smaller values always indicate better estimates. Values in bold are those methods which have the best performance for each simulated dataset.}
\label{table:yang_mse}
\end{table}

\begin{table}[!p]
\centering
\renewcommand{\arraystretch}{1.5}
\begin{tabular}{ccc}
\hline \hline
Dataset  & Prior on $\beta_\gamma$ & Prior on $\gamma$ \\ \hline
\begin{tabular}[c]{@{}c@{}}Tecator, Concrete,\\ Boston Housing, Protein\end{tabular} &     $\beta_\gamma \sim N_{p_\gamma}(0, 100 I_{p_\gamma})$                          & \begin{tabular}[c]{@{}c@{}} $p(\gamma) = h^{p_\gamma} \left(1 - h\right)^{p - p_\gamma}$, \\ $h = 5/100$  \end{tabular}    \\ \hline
PCR1, PCR2, PCR3           &        $\beta_\gamma \sim N_{p_\gamma}(0, 1/2 \times I_{p_\gamma})$                         &     \begin{tabular}[c]{@{}c@{}} $p(\gamma) = h^{p_\gamma} \left(1 - h\right)^{p - p_\gamma}$, \\ $h \sim Be(1,(p - 5)/5)$  \end{tabular}            \\ \hline
SNP    &   $\beta_\gamma \sim N_{p_\gamma}(0, 1/4 \times I_{p_\gamma})$      &     \begin{tabular}[c]{@{}c@{}} $p(\gamma) = h^{p_\gamma} \left(1 - h\right)^{p - p_\gamma}$, \\ $h = 5/p$  \end{tabular}   \\ \hline \hline
\end{tabular}
\caption{Table of prior specifications of 8 real dataset.}
\label{table:real_prior}
\end{table}

\begin{figure}[!p]
\includegraphics[width=\textwidth]{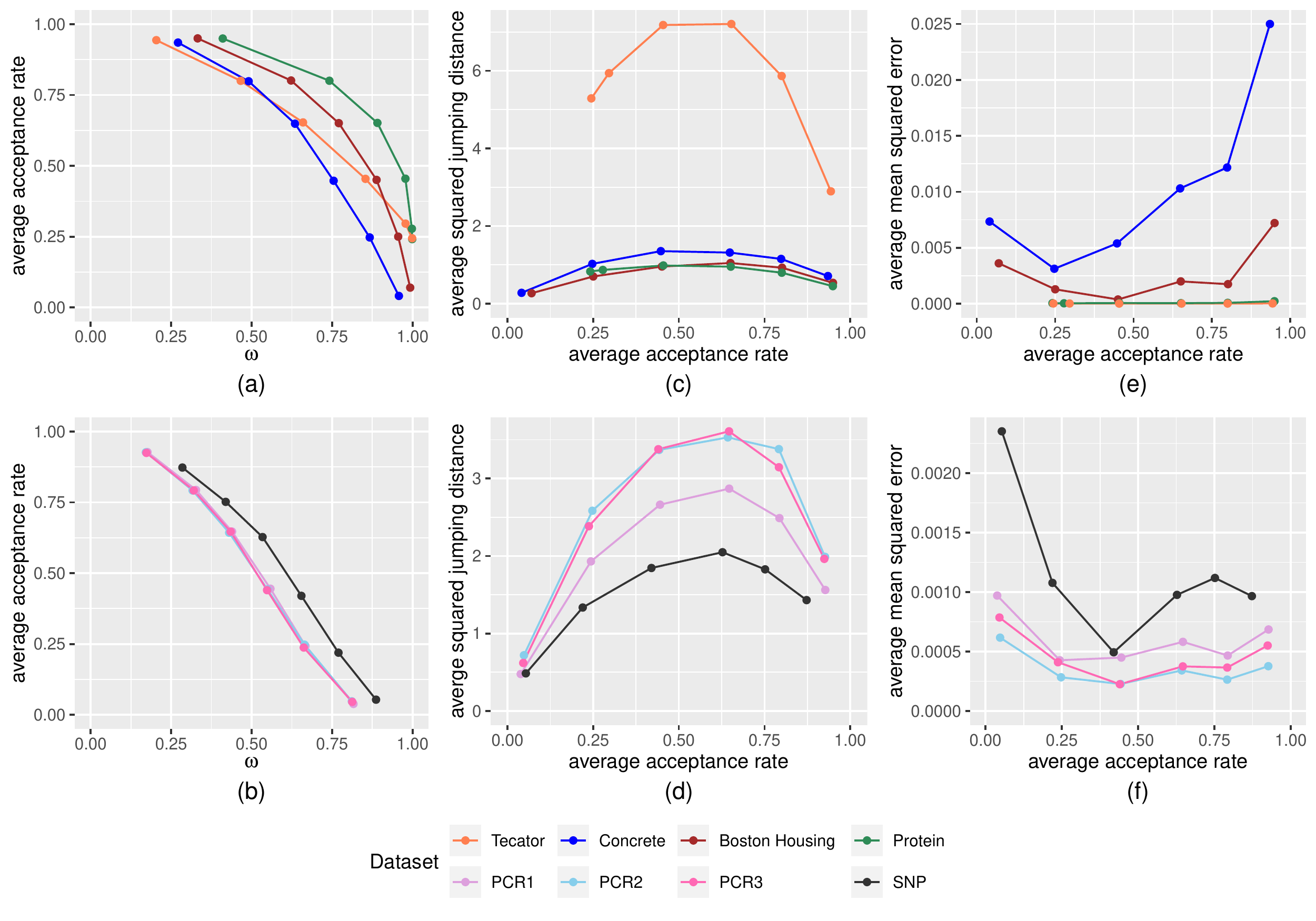}
\centering
\caption{Real data: plots of expected squared jumping distance and average mean square error again average acceptance rate and $\omega$. (a) average acceptance rate against $\omega$ for 4 small-$p$ real datasets; (b) average acceptance rate against $\omega$ for 4 large-$p$ real datasets; (c) expected squared jumping distance against average acceptance rate for 4 small-$p$ real datasets; (d) expected squared jumping distance against average acceptance rate for 4 large-$p$ real datasets; (e) average mean squared error against average acceptance rate for 4 small-$p$ real datasets; (f) average mean squared error against average acceptance rate for 4 large-$p$ real datasets.}
\label{fig:opt_acc_rate}
\end{figure}

\begin{figure}[!p]
\includegraphics[width=\columnwidth]{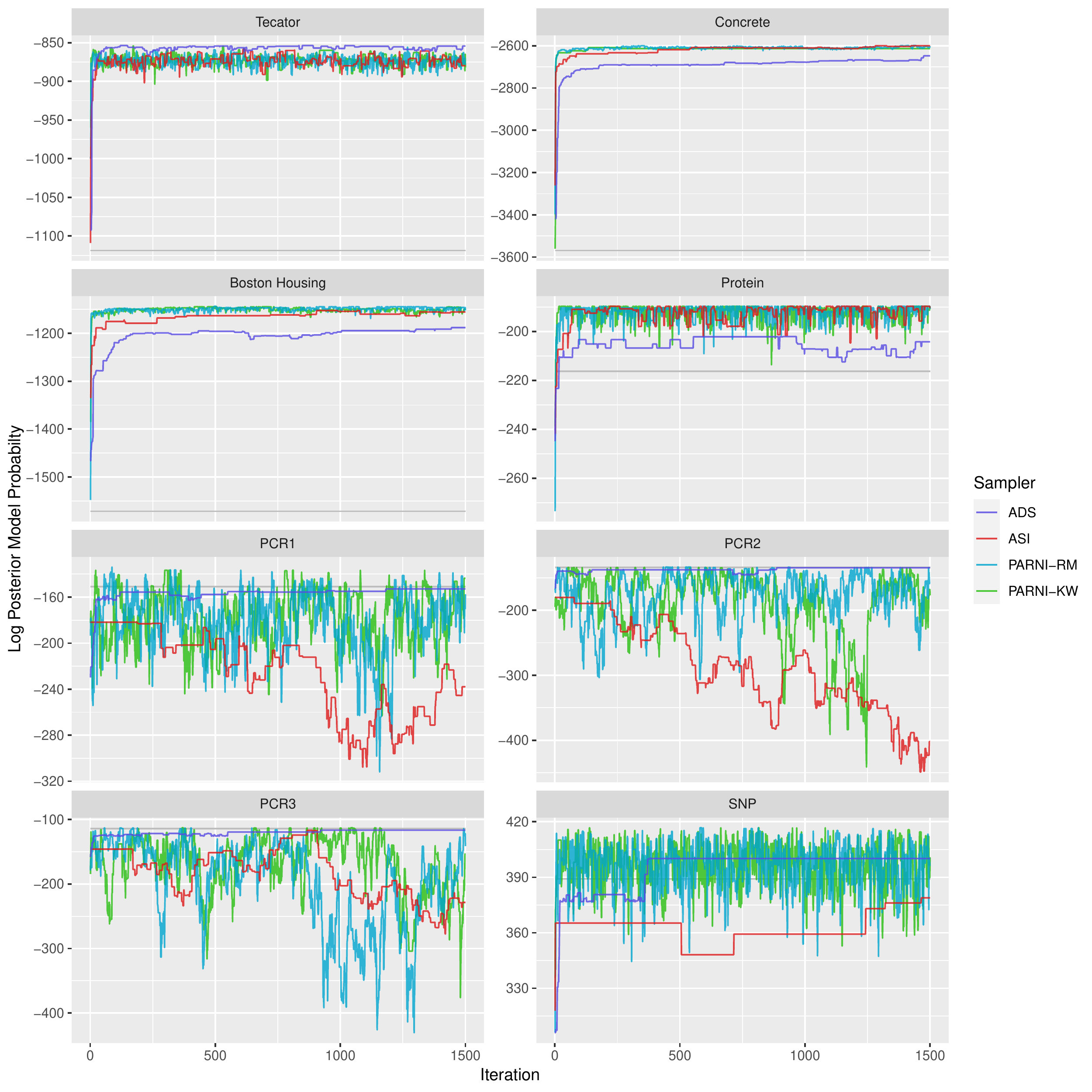}
\centering
\caption{Real data: trace plots of log posterior model probability from the Add-Delete-Swap (ADS), Adaptively Scaled Individual (ASI) adaptation, Pointwise implementation of Adaptive Random Neighbourhood Informed proposal with Robbins-Monro update (PARNI-RM) and Pointwise implementation of Adaptive Random Neighbourhood Informed proposal with Kiefer-Wolfowitz update (PARNI-KW) samplers for the first 1,500 iterations on 8 real datasets.}
\label{fig:real_trace}
\end{figure}

\begin{table}[!p]
\centering
\renewcommand{\arraystretch}{1.5}
\begin{tabular}{lllcccccc} 
\hline\hline
 dataset & $n$ & $p$ & \multicolumn{6}{c}{Samplers}      \\ 
\cline{4-9}
 &   &    &   ASI   &  HBS  &  TGS  &  WTGS  &  PARNI-KW  &  PARNI-RM        \\   
 \hline   Tecator  &  172  &  100  &  \textbf{-4.04}  &  -3.23  &  -3.13  &  -3.28  &  -3.44 &   -3.45 \\ 
 Boston Housing  &  506  &  104  &  0.51  &  \textbf{-1.43}  &  -1.20  &  -1.01  &  -1.19  & -1.18 \\ 
 Concrete  &  1030  &  79  &  -1.63  &  -1.27  &  -1.84  &  \textbf{-1.89}  &  -1.26   &   -0.72  \\ 
 Protein  &  96  &  88  &  -3.64  &  -1.61  &  -3.66  &  -4.21  &  -3.62   &  \textbf{-3.79}  \\ 
 \hline PCR1  &  60  &  22,575  &  -1.28  &  0.16  &  0.38  &  -1.51  &  \textbf{-1.53}   &   -1.34  \\ 
 PCR2  &  60  &  22,575  &  -0.75  &  0.67  &  0.77  &  -1.02  &    \textbf{-1.25}    &   -1.24  \\ 
 PCR3  &  60  &  22,575  &  -0.72  &  0.76  &  0.69  &  -1.02  &  \textbf{-1.20}   &   -0.98   \\ 
 SNP  &  993  &  79,748  &  -0.48  &  0.47  &  0.49  &  -0.75  &    \textbf{-1.37}    &  -1.23  \\ 
\hline\hline
\end{tabular}
\caption{Real data: relative average mean squared errors for the Adaptively Scaled Individual (ASI), Hamming Ball Sampler (HBS), Tempered Gibbs Sampler (TGS), Weighted Tempered Gibbs Sampler (WTGS), Pointwise implementation of Adaptive Random Neighbourhood Informed proposal with Kiefer-Wolfowitz update (PARNI-KW) and Pointwise implementation of Adaptive Random Neighbourhood Informed proposal with Robbins-Monro update (PARNI-RM) schemes on estimating posterior inclusion probabilities over important variables against a standard Add-Delete-Swap algorithm. Important variables are those variables have posterior inclusion probabilities greater than 0.01. Values are presented in logarithm to base 10. Smaller values always indicate better estimates. Values in bold are those methods which have the best performance for each real dataset.}
\label{table:real_mse}
\end{table}

\bibliographystyle{agsm}
\bibliography{references}  

\addcontentsline{toc}{chapter}{Appendices}

\appendix

\renewcommand*{\theequation}{%
  \thesection
 .\arabic{equation}%
}
\counterwithin*{equation}{section}

\newpage

\section{Additional materials}
\label{appendixlabel1}

\subsection{The Kiefer-Wolfowitz adaption scheme}
\label{apx:KW_scheme}

The optimal scaling property of a Gaussian random walk proposal on some specific forms of target distribution is well-studied. The most commonly used way to achieve optimal mixing time is to tune scaling parameters, which leads to an average acceptance rate of 0.234. In practice, even in those cases where the posterior distribution does not strictly obey the assumptions, the average acceptance rate of 0.234 is often a suitable guide and results in good practical performance. For those proposals beyond random walks, the guidelines for optimal tuning are often unknown. Due to this fact, we develop an adaptation scheme which is able to adapt the tuning parameters in which the mixing time and convergence rate are optimised without knowing any theoretical results in advance.

We design a scheme that maximises the \emph{Expected
Squared Jumping Distance} (ESJD). The ESJD is an efficiency measure which accounts for the jumping distances between two consecutive states from a Markov chain which is highly related to first order autocorrelation \citep{pasarica2010adaptively}. Suppose that $\omega \in \mathbb{R}$ is a continuous tuning parameter of a $\pi$-reversible transition kernel $p_\omega$. The definition of ESJD given parameter $\omega$ is given as follows
\begin{align}
    \text{ESJD}(\omega)
    & = \sum_{\gamma \in \Gamma} \sum_{\gamma^\prime \in \Gamma} \left(  \sum_{j = 1}^p (\gamma_j - \gamma_j^\prime)^2 \right)\pi(\gamma) p_\omega(\gamma, \gamma^\prime).
\end{align}
If $p_\omega$ is a Metropolis--Hastings transition kernel and can be decomposed into a product of a proposal kernel $Q_\omega$ with mass function $q_\omega$ and a term of the Metropolis-Hasting acceptance probability $\alpha_\omega(\gamma, \gamma^\prime)$, the definition of the ESJD above is equivalent to
\begin{align}
    \text{ESJD}(\omega)
    & = \sum_{\gamma \in \Gamma} \sum_{\gamma^\prime \in \Gamma} \left(  \sum_{j = 1}^p (\gamma_j - \gamma_j^\prime)^2 \right)\pi(\gamma) q_\omega(\gamma, \gamma^\prime)\alpha_\omega(\gamma, \gamma^\prime).
\end{align}
The ESJD is often infeasible to compute since it involves double sum over the sample space. To access the value of ESJD, we consider an estimator \textit{Average Squared Jumping Distance} (ASJD) which depends on the past chain and ASJD is defined as follows:
\begin{align}
    \text{ASJD}(\omega)
    & = \frac{1}{N} \sum_{i=0}^{N-1} \left(  \sum_{j = 1}^p (\gamma^{(i)}_j - \gamma^{(i+1)}_j)^2 \right)\pi(\gamma) 
\end{align}
or alternatively
\begin{align}
    \text{ASJD}(\omega)
    & = \frac{1}{N} \sum_{i=0}^{N-1} \left(  \sum_{j = 1}^p (\gamma^{(i)}_j - (\gamma^{(i)})^\prime_j)^2 \right)\alpha_\omega(\gamma^{(i)}, (\gamma^{(i)})^\prime)
\end{align}
where $(\gamma^{(i)})^\prime$ is the proposal of $\gamma^{(i)}$ through $Q_\omega$. From above, the main advance of using ASJD is that ASJD can be easily estimated in each individual iteration.

The objective is to locate the value of $\omega$ that leads to the largest ASJD. This is equivalent to solving the following optimisation problem of the tuning parameter
\begin{align}
    \omega^* := \arg \max_{\omega} ~ \text{ASJD}(\omega).
\end{align}
If objective function $\text{ESJD}(\omega)$ is unimodal and smooth, $\omega^*$ can be found by solving the first order ordinary differential equation
\begin{align}
    \frac{d}{d \omega}  \text{ASJD}(\omega) = 0.
\end{align} 
The Robbins-Monro scheme can be applied here to adaptively update the optimal $\theta$ when those derivatives exist analytically. In most cases, however, the derivatives are not available analytically, which makes the Robbins-Monro scheme impossible to use. The Kiefer-Wolfowitz scheme \citep{kiefer1952stochastic}, on the other hand, is an alternative to the Robbins-Monro algorithm where the derivatives are estimated using a finite difference method. 

The following is how a Kiefer-Wolfowitz scheme proceeds. Let $M(\omega)$ be an objective function with a maximum $\theta^*$. If $M(\omega)$ is assumed to be unknown but some random observations $\mathcal{M}(\omega)$ are given such that $M(\omega) = \mathbb{E}[\mathcal{M}(\omega)]$, $\omega$ is updated following an iterative algorithm as follows
\begin{align}
    \omega_{i+1} = \omega_i + a_i \left( \frac{\mathcal{M}(\omega_i + c_i) - \mathcal{M}(\omega_i - c_i)}{2c_i} \right)
\end{align}
where $a_i$ and $c_i$ are two diminishing sequences of forward positive step sizes and finite difference widths used respectively. In each iteration, we need two independent observations, $\mathcal{M}(\omega_i + c_i)$ and $\mathcal{M}(\omega_i - c_i)$, with tuning parameters $\omega_i + c_i$ and $\omega_i - c_i$ respectively. If the objective function $M(\omega)$ satisfies certain regularity conditions, it can be shown that $\omega_i$ will converge to the optimal value $\omega^*$ as $n \to \infty$. \cite{blum1954approximation} show that this convergence is almost sure provided that some other conditions hold, most importantly that the diminishing sequences $a_i$ and $c_i$ satisfy
\begin{enumerate}
    \item $c_i \to 0$ as $i \to \infty$;
    \item $\sum_{i=0}^\infty a_i = \infty$;
    \item $\sum_{i=0}^\infty a_i c_i < \infty$;
    \item $\sum_{i=0}^\infty a_i^2c_i^{-2} = \infty$.
\end{enumerate}

We design an adaptive MCMC sampler which combines the Kiefer-Wolfowitz scheme and parallel chain implementation. The parallel chain implementation involves a number of independent chains which only share the same tuning parameters and provides independent observations as the Kiefer-Wolfowitz scheme requires. We consider a sampler which involves $L$ parallel chains. In the sampler a new state is proposed through the kernel $Q_\omega$, which is then accepted with probability $\alpha_\omega$. An adaptation scheme with the Kiefer-Wolfowitz updates is given as follows: compute $a_i = i^{-\phi_a}$ and $c_i = i^{-\phi_c}$; calculate $\omega^+ = \omega_i + c_i$ and $\omega^- = \omega_i - c_n$; separate the $N$ chains into two groups, $L^-$ and $L^+$,; for each $l \in L^{+}$, propose new state $(\gamma^{(l,i)})^{\prime}$ using $Q_{\omega^+}(\gamma^{(l,i)}, \cdot)$, accept $(\gamma^{(l,i)})^{\prime}$ with probability $\alpha_{\omega^+}(\gamma^{(l,i)}, (\gamma^{(l,i)})^{\prime}$; for each $l \in L^{-}$, propose new state $(\gamma^{(l,i)})^{\prime}$ using $Q_{\omega^-}(\gamma^{(l,i)}, \cdot)$, accept $(\gamma^{(l,i)})^{\prime}$ with probability $\alpha_{\omega^-}(\gamma^{(l,i)}, (\gamma^{(l,i)})^{\prime})$; compute the ASJD for the current iteration by averaging over the chains in groups $L^+$ and $L^-$ respectively as follows: 
\begin{align}
    \text{ASJD}^{\bullet,(i)} \approx \frac{1}{|L^{\bullet}|} \sum_{l \in L^\bullet} \left(\sum_{j = 1}^p (\gamma^{(l,i)}_j - (\gamma^{(l,i)})^{\prime}_j)^2 \right) \alpha_{\theta^\bullet} (\gamma^{(l,i)}, (\gamma^{(l,i)})^{\prime}) \label{ESJD2}
\end{align}
where $\bullet$ is either $+$ or $-$; update the tuning parameter for the next iteration by
\begin{align}
    \omega^{(i+1)} =  \omega^{(i)} + a_i \left(\frac{\text{ASJD}^{+,(i)} - \text{ASJD}^{-,(i)}}{2c_i}\right) \label{ESJDupdateKW}.
\end{align}

\newpage

\section{Proofs}
\label{appendixlabel2}

\subsection{Proof of Proposition \ref{prop:RN_correct}}

\begin{proof}[Proof of Proposition \ref{prop:RN_correct} (Alternative)]
Define the probability mass function $\mu(\gamma,\gamma',k) := \pi(\gamma)p(k|\gamma)q_k(\gamma,\gamma')$. Then the algorithm can be viewed as a Markov chain on the larger state space $(\gamma,\gamma',k)$, which alternates between the following steps:
\begin{enumerate}
    \item Re-sample $k,\gamma'|\gamma$ from its conditional distribution with mass function $p(k|\gamma)q_k(\gamma,\gamma')$
    \item Perform a Metropolis step with deterministic proposal
    $$
    \phi \left( \begin{array}{c} \gamma \\ \gamma' \\ k \end{array} \right) = \left( \begin{array}{c} \gamma' \\ \gamma \\ \rho(k) \end{array} \right)
    $$
    and acceptance probability $\min\left(1,\mu \circ \phi(\gamma,\gamma',k)/\mu(\gamma,\gamma',k)\right)$.
\end{enumerate}
Note that $\phi \circ \phi(\gamma,\gamma',k) = (\gamma,\gamma',k)$, meaning $\phi$ is an involution. Therefore by Proposition 1 in \cite{andrieu2020general} the triple $(\mu,\phi,a)$ can be defined with $a(r) = \min(1,r)$, which combined with step 1 shows that the process on $\gamma$-space defined by a random neighbourhood sampler is a $\pi$-reversible Markov chain.
\end{proof}

\subsection{Proof of Proposition \ref{prop:ARN_correct}}

\begin{proof}[Proof of Proposition \ref{prop:ARN_correct}]

    Since $\eta \in \Delta_\epsilon^{2p}$, it is clear that $p^{\text{RN}}_\eta(k|\gamma) > \epsilon^p$ for all $\gamma \in \Gamma$ and $k \in \mathcal{K}$. Recall that $q^{\text{THIN}}_{\omega, k}$ is symmetric, so we have $q^{\text{THIN}}_{\omega, k}(\gamma, \gamma^\prime) = q^{\text{THIN}}_{\omega, k}(\gamma^\prime, \gamma)$. The condition
    \begin{align}
        p^{\text{RN}}_\eta(k|\gamma) q^{\text{THIN}}_{\omega, k}(\gamma, \gamma^\prime) > 0 \iff p^{\text{RN}}_\eta(k|\gamma^\prime) q^{\text{THIN}}_{\omega, k}(\gamma^\prime, \gamma) > 0
    \end{align}
    follows immediately.
    
    To show that $p^{\text{RN}}_\eta(\cdot|\gamma)$ is a probability measure on $\mathcal{K}$ for any $\gamma \in \Gamma$, we also need to show that 
    \begin{align}
        \sum_{k \in \mathcal{K}} p^{\text{RN}}_\eta(k|\gamma) = 1.
    \end{align}
    Starting from the right hand side gives
    \begin{align*}
        \sum_{k \in \mathcal{K}} p^{\text{RN}}_\eta(k|\gamma) & = \sum_{k \in \mathcal{K}} \prod_{j = 1}^p  p^{\text{RN}}_{\eta,j}(k_j|\gamma_j) \\
        & =  \prod_{j = 1}^p   \left( p^{\text{RN}}_{\eta,j}(k_j = 0|\gamma_j) + p^{\text{RN}}_{\eta,j}(k_j = 1|\gamma_j) \right). \\
        & = 1
    \end{align*}
    as required.
    
    We then show $q_{\omega,k}(\gamma, \cdot)$ is a probability measure on set $N(\gamma, k)$ for any $\gamma \in \Gamma$ and $k \in \mathcal{K}$. Let $J$ be a projection of $k$ to a set $J(k)$ consisting of the indices $j$ for which $k_j = 1$ (\textit{i.e.} $J(k) = \{j|k_j = 1\}$). Starting from equation \eqref{mat:ARN_prop3} of \emph{Remark} \ref{remark:thin_proposal} together with the identity used above
	$\int_{\Gamma}p(d\gamma) = \prod_j \int_{\Gamma_j}p_j(d\gamma_j)$, we obtain
	\begin{align*}
	\sum_{\gamma^\prime \in N(\gamma,k)} q^{\text{THIN}}_{\omega, k}(\gamma, \gamma^\prime)
	& = \prod_{j \in J(k)} \sum_{\gamma^\prime_j \in \{0,1\}} \left(\frac{\omega}{1-\omega}\right)^{d_H(\gamma_j,\gamma_j^\prime)} (1-\omega) \\
	& = \prod_{j \in J(k)} \sum_{\gamma_j^\prime \in \{\gamma_j,1-\gamma_j\}} \left(\frac{\omega}{1-\omega}\right)^{|\gamma_j - \gamma_j^\prime|} (1-\omega) \\
	& = \prod_{j \in J(k)} \left( (1-\omega) + \omega \right) \\
	& = 1
	\end{align*}
	as required.
\end{proof}

\subsection{Proof of Theorem \ref{them:asi_arn_equal}}

Before proving Theorem \ref{them:asi_arn_equal}, we first draw some conclusions on acceptance probability of the ASI and ARN proposals.

\begin{prop} \label{prop:ARN_acc_rate_equ}
Suppose $\gamma$ is the current state, $\gamma^\prime$ is the proposed state and $\eta \in \Delta_\epsilon^{2p}$ is fixed parameter. For any $k$ that constructs neighbourhood containing $\gamma$ and $\gamma^\prime$ and any choices of $\xi \in \Delta_\epsilon$ and $\omega \in \Delta_\epsilon$, the Metropolis-Hastings acceptance probability of the ARN proposal, $\alpha_{(\xi\eta,\omega), k}^{\text{ARN}}$, as in (\ref{mat:ARN_acc}) is fixed. In addition, this term is also the acceptance probability of the ASI proposal, \textit{i.e.}
\begin{align*}
\alpha_{(\xi\eta,\omega), k}^{\text{ARN}} (\gamma, \gamma^\prime) = \alpha_{\zeta \eta}^{\text{ASI}} (\gamma, \gamma^\prime).
\end{align*}
for any $\zeta \in \Delta_\epsilon$.
\end{prop}

\begin{proof}[Proof of Proposition \ref{prop:ARN_acc_rate_equ}]
	
	Suppose that $\gamma$, $\gamma^\prime \in \Gamma$, $\eta \in \Delta_\epsilon^{2p}$ are given and fixed. We consider all the $k \in \mathcal{K}$ such that $\gamma^\prime \in N(\gamma, k)$. We are going to show that for any $\xi$ and $\omega \in \Delta_\epsilon$, the acceptance probability $\alpha_{\theta, k}^{\text{ARN}} (\gamma, \gamma^\prime)$ is free from the choice of $k$, $\xi$ and $\omega$.

	To locate the different positions between $\gamma$ and $\gamma^\prime$, we define the set $J(\gamma, \gamma^\prime) := \{j| \gamma_j \ne \gamma_j^\prime\} \subseteq J(k)$. From (\ref{mat:ARN_acc}), we have
	\begin{align*}
	\alpha_{\theta, k}^{\text{ARN}} (\gamma, \gamma^\prime) & = \min \left\{1, \frac{\pi(\gamma^\prime)p^{\text{RN}}_{\xi \eta^\text{opt}}(k|\gamma^\prime)q_{\omega,k}^{\text{THIN}}(\gamma^\prime,\gamma)}{\pi(\gamma)p^{\text{RN}}_{\xi \eta^\text{opt}}(k|\gamma)q_{\omega, k}^{\text{THIN}}(\gamma,\gamma^\prime)} \right\} \\
	 & = \min \left\{1, \frac{\pi(\gamma^\prime)p^{\text{RN}}_{\xi \eta^\text{opt}}(k|\gamma^\prime)}{\pi(\gamma)p^{\text{RN}}_{\xi \eta^\text{opt}}(k|\gamma)} \right\}
	\end{align*}
	where the last equality follows since $q^{\text{THIN}}_{\omega, k}$ is symmetric. Substituting $p^{\text{RN}}_{\xi \eta}$ yields
	\begin{align*}
	\alpha_{\theta, k}^{\text{ARN}} (\gamma, \gamma^\prime) & = \min \left\{1, \frac{\pi(\gamma^\prime)}{\pi(\gamma)} \prod_{j=1}^{p} \frac{(\xi A_j)^{(1-\gamma^\prime_j)k_j} (1-\xi A_j)^{(1-\gamma^\prime_j)(1-k_j)} (\xi D_j)^{\gamma^\prime_jk_j} (1-\xi D_j)^{\gamma^\prime_j(1-k_j)}}{(\xi A_j)^{(1-\gamma_j)k_j} (1-\xi A_j)^{(1-\gamma_j)(1-k_j)} (\xi D_j)^{\gamma_jk_j} (1-\xi D_j)^{\gamma_j(1-k_j)}}  \right\} \\
	& = \min \left\{1, \frac{\pi(\gamma^\prime)}{\pi(\gamma)} \prod_{j \in J(\gamma, \gamma^\prime)} \frac{ (A_j)^{(1-\gamma^\prime_j)}  (D_j)^{\gamma^\prime_j} }{(A_j)^{(1-\gamma_j)}  (D_j)^{\gamma_j} }  \right\}.
	\end{align*}
	The value of $\alpha_{\theta, k}^{\text{ARN}} (\gamma, \gamma^\prime)$ only depends on the choice of $\eta$ and it does not involve the terms $\xi$, $\omega$ and $k$. Therefore, the proposition is proved. In addition, this is also the Metropolis-Hastings acceptance probability for ASI of proposing $\gamma^\prime$ to $\gamma$ following the definition of the ASI sampler.
	
\end{proof}

Now we formally prove Theorem \ref{them:asi_arn_equal}.

\begin{proof}[Proof of Theorem \ref{them:asi_arn_equal}]

    Rewriting the transition kernel of $p^{\mathrm{ARN}}_{(\xi \eta,\omega)}$ gives
    \begin{align}
	p^{\mathrm{ARN}}_{(\xi\eta,\omega)} (\gamma, \gamma^\prime) & = \sum_{k \in \mathcal{K}} p_{\xi\eta}^\text{RN} (k|\gamma) q^{\text{THIN}}_{\omega, k}(\gamma, \gamma^\prime)  \alpha^{\text{ARN}}_{(\xi\eta,\omega), k}(\gamma \gamma^\prime) \notag \\
	& =  \alpha^{\text{ASI}}_{\zeta \eta}(\gamma, \gamma^\prime) \sum_{k \in \mathcal{K}} p_{\xi \eta}^\text{RN} (k|\gamma) q^{\text{THIN}}_{\omega, k}(\gamma, \gamma^\prime) \label{mat:apx_asi_arn_1}
	\end{align}
    for which the last line follows from Proposition \ref{prop:ARN_acc_rate_equ} and the fact that $\zeta = \xi \times \omega$.

	Recall the definitions of the conditional distribution of $k$ in (\ref{mat:RN_prop1}) and the within-neighbourhood proposal in (\ref{mat:ARN_prop3}).  Together with $\eta = (A, D)$, we have
	\begin{align*}
	p^{\text{RN}}_{\xi\eta} (k|\gamma) & = \prod_{j = 1}^p p^{\text{RN}}_{\xi\eta, j}(k_j|\gamma_j) \\ 
	& = \prod_{j = 1}^p (\xi A_j)^{(1-\gamma_j)k_j} (1-\xi A_j)^{(1-\gamma_j)(1-k_j)} (\xi D_j)^{\gamma_j k_j} (1-\xi D_j)^{\gamma_j(1-k_j)}
	\end{align*}
	and
	\begin{align*}
	q^{\text{THIN}}_{\omega, k}(\gamma,\gamma^\prime)	& = \left(\frac{\omega}{1-\omega}\right)^{d_H(\gamma, \gamma^\prime)} (1-\omega)^{\sum_{j=1}^p k_j} \\
	& = \left(\frac{\omega}{1-\omega}\right)^{d_H(\gamma, \gamma^\prime)} \prod_{j = 1}^p (1-\omega)^{k_j}
	\end{align*}
	where the last line follows since $(1-\omega)^{k_j} = 1$ if $k_j = 0$. Substituting the above into (\ref{mat:apx_asi_arn_1}) yields
	\begin{align}
	p^{\mathrm{ARN}}_{(\xi\eta,\omega)} (\gamma, \gamma^\prime)
	& =  \alpha^{\text{ASI}}_{\zeta\eta}(\gamma, \gamma^\prime)  \sum_{k \in \mathcal{K}}  \left[  \prod_{j=1}^p (\xi A_j)^{(1-\gamma_j)k_j} (1-\xi A_j)^{(1-\gamma_j)(1-k_j)} \right. \notag \\
	& \qquad \left. (\xi D_j)^{\gamma_jk_j}  (1-\xi D_j)^{\gamma_j(1-k_j)} \left(\frac{\omega}{1-\omega}\right)^{d_H(\gamma, \gamma^\prime)}  (1-\omega)^{k_j} \right]. \label{mat:apx_asi_arn_2}
	\end{align}

	Let $J(\gamma, \gamma^\prime)$ be a set which consists of the indices $j$ for which $\gamma_j \neq \gamma_j^\prime$ and $J(\gamma, \gamma^\prime)^c$ be the complement to $J(\gamma, \gamma^\prime)$. By definition $k_j$ must be $1$ when $j \in J(\gamma, \gamma^\prime)$ and $k_j$ can take values either $0$ or $1$ when $j \in J(\gamma, \gamma^\prime)^c$. Continuing from (\ref{mat:apx_asi_arn_2}), we divide $j=1$ to $p$ into two groups, $J(\gamma, \gamma^\prime)$ and $J(\gamma, \gamma^\prime)^c$, and obtain
	\begin{align*}
	p^{\mathrm{ARN}}_{(\xi\eta,\omega)} (\gamma, \gamma^\prime) &  =  \alpha^{\text{ASI}}_{\zeta\eta}(\gamma, \gamma^\prime)  \left(\frac{\omega}{1-\omega}\right)^{|J(\gamma, \gamma^\prime)|}
	\prod_{j \in J(\gamma, \gamma^\prime)} \left(  (\xi A_j)^{(1-\gamma_j)} (\xi D_j)^{\gamma_j} (1-\omega) \right) \times \\
	& \qquad  \prod_{j \in J(\gamma, \gamma^\prime)^c} \Bigg[ \sum_{k_j \in \{0,1\}} (\xi A_j)^{(1-\gamma_j)k_j} (1-\xi A_j)^{(1-\gamma_j)(1-k_j)}  \\
	& \qquad  (\xi D_j)^{\gamma_j k_j} (1-\xi D_j)^{\gamma_j(1-k_j)} (1-\omega)^{k_j} \Bigg] \\
	& = \alpha^{\text{ASI}}_{\zeta\eta}(\gamma, \gamma^\prime) \prod_{j \in J(\gamma, \gamma^\prime)}  \left(  (\xi A_j)^{(1-\gamma_j)} (\xi D_j)^{\gamma_j} \cdot \omega \right) \times \\
	& \qquad  \prod_{j \in J(\gamma,\gamma^\prime)^c} \underbrace{ \left( (1-\xi A_j)^{1-\gamma_j} (1-\xi D_j)^{\gamma_j} +    (\xi A_j)^{1-\gamma_j}  (\xi D_j)^{\gamma_j} (1-\omega) \right)}_{I_j}
	\end{align*}
	
	We am going to further investigate the terms $I_j$ for $j \in J(\gamma,\gamma^\prime)^c$. Clearly that, if $\gamma_j = 1$, then
	\begin{align*}
	I_j & = (1-\xi D_j) + \xi D_j (1-\omega) \\
	& = 1 - \omega \xi D_j,
	\end{align*}
	similarly that when $\gamma_j = 0$, we have
	\begin{align*}
	I_j & = (1-\xi A_j) + \xi A_j (1 - \omega) \\
	& = 1 - \omega \xi A_j.
	\end{align*}
	
	Putting everything back to $p^{\mathrm{ARN}}_{\xi\eta,\omega}$, and reconstructing the product in $j$ from $1$ to $p$ gives
	\begin{align*}
	p^{\mathrm{ARN}}_{\xi\eta,\omega} (\gamma, \gamma^\prime) &  =  \alpha^{\text{ASI}}_{\zeta\eta}(\gamma, \gamma^\prime) \times 
	\prod_{j \in J(\gamma, \gamma^\prime)}  (\xi \omega A_j)^{(1-\gamma_j)} (\xi \omega D_j)^{\gamma_j}  \times \prod_{j \in J(\gamma, \gamma^\prime)^c}  (1-\xi \omega A_j)^{(1-\gamma_j)}  (1-\xi \omega D_j)^{\gamma_j} \\
	& =  \alpha^{\text{ASI}}_{\zeta\eta}(\gamma, \gamma^\prime) \times \prod_{j =1}^p \left[ (\xi \omega A_j)^{(1-\gamma_j)\gamma_j^\prime} (\xi \omega D_j)^{\gamma_j(1-\gamma_j^\prime)} (1-\xi \omega A_j)^{(1-\gamma_j)(1-\gamma_j^\prime)}  (1-\xi \omega D_j)^{\gamma_j \gamma_j^\prime} \right].
	\end{align*}
	Rewriting the above in terms of $\zeta = \xi\times\omega$ yields
	\begin{align*}
	p^{\text{ARN}}_{(\xi\eta,\omega)} (\gamma, \gamma^\prime) & =   \alpha^{\text{ASI}}_{\zeta\eta}(\gamma, \gamma^\prime) \times \prod_{j =1}^p \left[  (\zeta A_j)^{(1-\gamma_j)\gamma_j^\prime} (\zeta D_j)^{\gamma_j(1-\gamma_j^\prime)} (1-\zeta A_j)^{(1-\gamma_j)(1-\gamma_j^\prime)}  (1-\zeta D_j)^{\gamma_j \gamma_j^\prime}\right] \\
	& =  p^{\text{ASI}}_{\zeta\times\eta} (\gamma, \gamma^\prime)
	\end{align*}
	as required.
\end{proof}

\subsection{Proof of Corollary \ref{coro:arns_equal}}

\begin{proof}[Proof of Corollary \ref{coro:arns_equal}]
It is clear that the argument holds from the Theorem \ref{them:asi_arn_equal} and its proof.
\end{proof}

\subsection{Proof of Proposition \ref{prop:PARNI_acc}}

\begin{proof}

We define $J(k)$ to be a set that consists of the positions that $k_j = 1$ (\textit{i.e.} $J(k) = \{j|k_j = 1\}$) and $J(\gamma, \gamma^\prime)$ to be a set which consists of the different variables (\textit{i.e.} $J(\gamma, \gamma^\prime) = \{j| \gamma_j \ne \gamma_j^\prime\}$). They obey the relationship $J(\gamma, \gamma^\prime) \subseteq J(k)$ if $\gamma^\prime \in N(\gamma, k)$. The $j$-th conditional distribution of $k_j$, $p^{\text{RN}}_{\eta,j}$, satisfies 
\begin{align*}
	p^{\text{RN}}_{\eta,j}(k_j|\gamma_j) = p^{\text{RN}}_{\eta,j}(k_j|\gamma^\prime_j)
\end{align*} 
if $j$ is outside of $J(\gamma,\gamma^\prime)$. This is because $\gamma_j = \gamma^\prime_j$ for $j \in J(\gamma,\gamma^\prime)$. 

Start with simplifying the ratio $p^{\text{RN}}_{\eta}(k|\gamma^\prime)/p^{\text{RN}}_{\eta}(k|\gamma)$. Following the similar steps in the proof of Proposition \ref{prop:ARN_acc_rate_equ} and suppose $\gamma_j = 1-\gamma_j^\prime$ for all $j \in J(\gamma,\gamma^\prime)$ and $\eta = (A,D)$, we can show the following
\begin{align}
    \frac{p^{\text{RN}}_\eta(k|\gamma^\prime)}{p^{\text{RN}}_\eta(k|\gamma)} & = \prod_{j \in J(\gamma, \gamma^\prime)} \frac{A_j^{1-\gamma_j} D_j^{\gamma_j}}{A_j^{\gamma_j} D_j^{1-\gamma_j}} \notag \\
    & = \prod_{j \in J(\gamma, \gamma^\prime)} \left(\frac{A_j}{D_j}\right)^{1-2\gamma_j} = \prod_{j \in J(\gamma, \gamma^\prime)} \left(\frac{A_j}{D_j}\right)^{2\gamma^\prime_j-1}. \label{mat:apx_parni_con_k_odd}
\end{align}

Next step is simplifying the second ratio $q^{\text{PARNI}}_{\theta, k}(\gamma^\prime, \gamma) / q^{\text{PARNI}}_{\theta, k}(\gamma, \gamma^\prime)$ and showing that this term can be decomposed into 3 parts. Since the sampling process is sequential, the model $\gamma(r)$ is proposed from $\gamma(r-1)$ at time $r$. For reversed move, the model $\gamma^\prime(r)$ is proposed from $\gamma^{\prime}(r-1)$ at time $r$. Moreover, $\gamma(0)$ and $\gamma^\prime(p_k)$ are the current state $\gamma$ and $\gamma^\prime(p_k)$ and $\gamma^{\prime}(0)$ are the final proposal $\gamma^\prime$. We correlate $r$ and $r^\prime = p_k-r+1$ since $\gamma(r) = \gamma^\prime(r^\prime-1)$ and $K_{r} = K^\prime_{r^\prime}$. We consider the ratio $q^{\text{PARNI}}_{\theta,K^{\prime}_{r^\prime}}(\gamma^\prime(r^\prime-1), \gamma^\prime(r^\prime)) / q^{\text{PARNI}}_{\theta,K_r}(\gamma(r-1), \gamma(r))$. From (\ref{mat:PARNI_prop}) and therefore have
\begin{align*}
    \frac{q^{\text{PARNI}}_{\theta,K_{r^\prime}}(\gamma^\prime(r^\prime-1), \gamma^{\prime}(r^\prime))}{q^{\text{PARNI}}_{\theta,K_r}(\gamma(r-1), \gamma(r))} & =  \frac{g\left(\frac{\pi(\gamma^{\prime}(r^\prime))p^{\text{RN}}_{\eta}(e(K_{r^\prime})|\gamma^{\prime}(r^\prime))}{\pi(\gamma^\prime(r^\prime-1))p^{\text{RN}}_{\eta}(e(K_{r^\prime})|\gamma^\prime(r^\prime-1))}\right)/Z^\prime(r^\prime)}{g\left(\frac{\pi(\gamma(r))p^{\text{RN}}_{\eta}(e(K_{r})|\gamma(r))}{\pi(\gamma(r-1))p^{\text{RN}}_{\eta}(e(K_{r})|\gamma(r-1))}\right)/Z(r)}.
\end{align*}
Since $g$ is a balancing function and satisfies $g(t) = tg(1/t)$ for any positive $t$ and $\gamma(r) = \gamma^\prime(r^\prime+1)$, we have
\begin{align*}
    \frac{q^{\text{PARNI}}_{\theta,K_{r^\prime}}(\gamma^\prime(r^\prime-1), \gamma^{\prime}(r^\prime))}{q^{\text{PARNI}}_{\theta,K_r}(\gamma(r-1), \gamma(r))} & = \frac{\pi(\gamma(r-1))p^{\text{RN}}_{\eta}(e(K_{r})|\gamma(r-1))}{\pi(\gamma(r))p^{\text{RN}}_{\eta}(e(K_{r})|\gamma(r))} \cdot \frac{Z(r)}{Z^\prime(r^\prime)}.
\end{align*}
The product of the above ratio from $r=1$ to $p_k$ yields the term $q^{\text{PARNI}}_{\theta, k}(\gamma^\prime, \gamma) / q^{\text{PARNI}}_{\theta, k}(\gamma, \gamma^\prime)$ as follows
\begin{align*}
    \frac{q^{\text{PARNI}}_{\theta, k}(\gamma^\prime, \gamma)}{ q^{\text{PARNI}}_{\theta, k}(\gamma, \gamma^\prime)} & =  \frac{\prod_{r^\prime=1}^{p_k} q^{\text{PARNI}}_{\theta,K_{r^\prime}}(\gamma^\prime(r^\prime-1), \gamma^{\prime}(r^\prime))}{\prod_{r=1}^{p_k} q^{\text{PARNI}}_{\theta,K_r}(\gamma(r-1), \gamma(r))} \\
    & = \prod_{r=1}^{p_k} \frac{\pi(\gamma(r-1))p^{\text{RN}}_{\eta}(e(K_{r})|\gamma(r-1))}{\pi(\gamma(r))p^{\text{RN}}_{\eta}(e(K_{r})|\gamma(r))} \cdot \frac{Z(r)}{Z^\prime(r^\prime)} \\
    & = \underbrace{ \left( \prod_{r=1}^{p_k} \frac{\pi(\gamma(r-1))}{\pi(\gamma(r))} \right) }_{\text{I}} \cdot \underbrace{ \left( \prod_{r=1}^{p_k} \frac{p^{\text{RN}}_{\eta}(e(K_{r})|\gamma(r-1))}{p^{\text{RN}}_{\eta}(e(K_{r})|\gamma(r))} \right) }_{\text{II}} \cdot \underbrace{ \left( \prod_{r=1}^{p_k} \frac{Z(r)}{Z^\prime(r)} \right) }_{\text{III}}
\end{align*}
since $r^\prime = p_k-r$.

The first term $\text{I}$ is equal to
\begin{align*}
     \text{I} = \frac{\pi(\gamma)}{\pi(\gamma(1))} \frac{\pi(\gamma(1))}{\pi(\gamma(2))} \cdots \frac{\pi(\gamma(p_k-2))}{\pi(\gamma(p_k-1))} \frac{\pi(\gamma(p_k-1))}{\pi(\gamma^\prime)}.
\end{align*}
Most terms can be cancelled out and this leaves the first numerator and the last denominator
\begin{align*}
    \text{I} = \frac{\pi(\gamma)}{\pi(\gamma^\prime)}.
\end{align*}

We now deal with the second term $\text{II}$. Substituting the values gives
\begin{align*}
    \frac{p^{\text{RN}}_{\eta}(e(K_{r})|\gamma(r-1))}{p^{\text{RN}}_{\eta}(e(K_{r})|\gamma(r)))} = \left(\frac{A_{K_{r}}}{D_{K_{r}}}\right)^{1-2\gamma(r)_{K_{r}}}.
\end{align*}
We know that the positions $K_1,\dots,K_{p_k}$ are distinct and the vectors $e(K_1),\dots,e(K_{p_k})$ can recover the auxiliary variable $k$. Therefore, we obtain
\begin{align*}
    \text{II} = \prod_{j \in J(\gamma, \gamma^\prime)} \left(\frac{A_j}{D_j}\right)^{1-2\gamma_j}.
\end{align*}

Following the above arguments, the product of sequence $j = 1, \dots, p_k$ can be simplified 
\begin{align}
    \frac{q^{\text{PARNI}}_{\theta, k}(\gamma^\prime, \gamma)}{q^{\text{PARNI}}_{\theta, k}(\gamma, \gamma^\prime)} = \underbrace{\frac{\pi(\gamma)}{\pi(\gamma^\prime)}}_{\text{I}} \cdot \underbrace{\prod_{j \in J(\gamma, \gamma^\prime)}  \left(\frac{A_j}{D_j}\right)^{1-2\gamma_j} }_{\text{II}} \cdot \underbrace{\prod_{r = 1}^{p_k} \frac{Z(r)}{Z^\prime(r)}}_{\text{III}} \label{mat:apx_parni_mh_acc}
\end{align}

The Metropolis-Hastings acceptance probability in (\ref{mat:PARNIacc}) is 
\begin{align*}
    \alpha^{\text{PARNI}}_{\theta, k}(\gamma, \gamma^\prime) & = \left\{ 1, \frac{\pi(\gamma^\prime)p^{\text{RN}}_{\eta}(k|\gamma^\prime)q^{\text{PARNI}}_{\theta, k}(\gamma^\prime, \gamma)}{\pi(\gamma)p^{\text{RN}}_{\eta}(k|\gamma)q^{\text{PARNI}}_{\theta, k}(\gamma, \gamma^\prime)} \right\} \\
    & = \min\left\{1, \left(\frac{\pi(\gamma^\prime)}{\pi(\gamma)}\cdot \text{I}\right) \cdot \left(\frac{p^{\text{RN}}_{\eta}(k|\gamma^\prime)}{p^{\text{RN}}_{\eta}(k|\gamma)}\cdot\text{II}\right) \cdot \text{III} \right\}
\end{align*}
From (\ref{mat:apx_parni_con_k_odd}) and (\ref{mat:apx_parni_mh_acc}), we have
\begin{align*}
    & \frac{\pi(\gamma^\prime)}{\pi(\gamma)}\cdot \text{I} = 1 \\
    & \frac{p^{\text{RN}}_{\eta}(k|\gamma^\prime)}{p^{\text{RN}}_{\eta}(k|\gamma)}\cdot\text{II} = 1
\end{align*}
and we therefore obtain
\begin{align*}
    \alpha^{\text{PARNI}}_{\theta, k}(\gamma, \gamma^\prime) & = \text{III} \\
    & = \prod_{j=1}^{p_k} \frac{Z(j)}{Z^\prime(j)}
\end{align*}
as required.

\end{proof}

\subsection{Proof of Lemma \ref{lemma:PARNI_simerg}}

The proof of Lemma \ref{lemma:PARNI_simerg} is structured as proof of Lemma 1 in \cite{griffin2021search}.

\begin{proof}[Proof of Lemma \ref{lemma:PARNI_simerg}]

We first introduce some preliminary work. Since $\theta = (\eta, \omega) \in \Delta_\epsilon^{2p+1}$, we have
\begin{align}
     \epsilon^p \leq p^{\text{RN}}_{\eta}(k|\gamma) \leq (1-\epsilon)^p
\end{align}
for any $k$ and $\gamma$. Similar arguments implies that
\begin{align}
    \left(\frac{\epsilon}{1-\epsilon} \right)^p \leq \frac{p^{\text{RN}}_{\eta}(k|\gamma^\prime)}{p^{\text{RN}}_{\eta}(k|\gamma)}  \leq \left(\frac{1-\epsilon}{\epsilon} \right)^p. \label{mat:apx_parni_sue_1}
\end{align}
From assumption (A.2), we know that there exists a constant $\Pi$ such that
\begin{align}
    \frac{1}{\Pi} \leq \frac{\pi(\gamma^\prime)}{\pi(\gamma)} \leq \Pi. \label{mat:apx_parni_sue_2}
\end{align}
Let $t_{\gamma, \gamma^\prime, k}$ be
\begin{align}
    t_{\gamma, \gamma^\prime, k} = \frac{\pi(\gamma^\prime)}{\pi(\gamma)} \cdot \frac{p^{\text{RN}}_{\eta}(k|\gamma^\prime)}{p^{\text{RN}}_{\eta}(k|\gamma)}.
\end{align}
Using (\ref{mat:apx_parni_sue_1}) and (\ref{mat:apx_parni_sue_2}) leads to
\begin{align}
   \frac{1}{\Pi} \cdot \left(\frac{\epsilon}{1-\epsilon} \right)^p \leq t_{\gamma, \gamma^\prime, k}  \leq \Pi \cdot \left(\frac{1-\epsilon}{\epsilon} \right)^p,
\end{align}
for any $\gamma,\gamma^\prime \in \Gamma$ and $k \in \mathcal{K}$, thus,
\begin{align*}
   g\left(\frac{1}{\Pi} \cdot \left(\frac{\epsilon}{1-\epsilon} \right)^p\right) \leq g(t_{\gamma, \gamma^\prime, k}) \leq g \left( \Pi \cdot \left(\frac{1-\epsilon}{\epsilon} \right)^p \right)
\end{align*}
since $g$ is a non-decreasing function. We define the following qualities
\begin{align*}
    g^\uparrow & := g \left( \Pi \cdot \left(\frac{1-\epsilon}{\epsilon} \right)^p \right) \\
    g^\downarrow & := g\left(\frac{1}{\Pi} \cdot \left(\frac{\epsilon}{1-\epsilon} \right)^p\right)
\end{align*}
Therefore, for all normalising constants $Z(r)$, it is bounded between $Z^{\downarrow}$ and $Z^{\uparrow}$ where
\begin{align}
    Z^{\downarrow} & := 2\epsilon g^{\downarrow} \label{mat:apx_parni_sue_6}\\
    Z^{\uparrow} & := 2(1-\epsilon) g^{\uparrow} \label{mat:apx_parni_sue_4}
\end{align}
due to the fact that $\omega \in(\epsilon, 1- \epsilon)$.

Suppose $k$ and corresponding $K$ are given, we now bound each individual kernel $q^{\text{PARNI}}_{\theta, K_r}(\gamma, \gamma^\prime)$ for any $\gamma$ and $r$ from $1$ to $p_k$. Staring from the definition in (\ref{mat:PARNI_prop}),
\begin{align}
    q^{\text{PARNI}}_{\theta, K_r}(\gamma, \gamma^\prime) = & g\left(\frac{\pi(\gamma^\prime)p^{\text{RN}}_{\eta}(e(K_r)|\gamma^\prime)}{\pi(\gamma)p^{\text{RN}}_{\eta}(e(K_r)|\gamma)}\right) q^{\text{THIN}}_{\omega, e(K_r)}(\gamma, \gamma^\prime) / Z(r) \notag \\
    \geq & \frac{\epsilon g^\downarrow}{Z(r)} \notag \\
    \geq & \frac{\epsilon g^\downarrow}{2(1-\epsilon)g^\uparrow} \label{mat:apx_parni_sue_5}
\end{align}
where the last is followed by (\ref{mat:apx_parni_sue_4}). We next consider the full update kernel of PARNI
\begin{align}
    q_{\theta, k}^{\text{PARNI}}(\gamma, \gamma^\prime) = \prod_{r = 1}^{p_k} q^{\text{PARNI}}_{\theta, K_r}(\gamma(r-1), \gamma(r))
\end{align}
where $\gamma(0) = \gamma$ and $\gamma(p_k) = \gamma^\prime$. From (\ref{mat:apx_parni_sue_5}), the full update kernel is bounded as follows
\begin{align}
    q_{\theta, k}^{\text{PARNI}}(\gamma, \gamma^\prime) \geq \left( \frac{\epsilon g^\downarrow}{2(1-\epsilon)g^\uparrow} \right)^{p_k} \notag \\
    \geq \left( \frac{\epsilon g^\downarrow}{2(1-\epsilon)g^\uparrow} \right)^{p}
\end{align}
since $p_k \leq p$ for all $k$. We also bound the Metropolis-Hastings acceptance probability in (\ref{mat:PARNIacc}) from below
\begin{align*}
    \alpha_{\theta, k}^\text{PARNI}(\gamma, \gamma^\prime) & = \min \left\{1, \frac{\pi(\gamma^\prime) p^{\text{RN}}_\eta(k|\gamma^\prime) q^{\text{PARNI}}_{\theta, k}(\gamma^\prime,\gamma)}{\pi(\gamma) p^{\text{RN}}_\eta(k|\gamma) q^{\text{PARNI}}_{\theta, k}(\gamma, \gamma^\prime)}\right\} \\
    & \geq \pi(\gamma^\prime) p^{\text{RN}}_\eta(k|\gamma^\prime) q^{\text{PARNI}}_{\theta, k}(\gamma^\prime,\gamma) \\ 
    & \geq \pi_m  \left( \frac{\epsilon g^\downarrow}{2(1-\epsilon)g^\uparrow} \right)^p \epsilon^p
\end{align*}
where $\pi_m = \min_\gamma\pi(\gamma)$.

Finally, we can chose $b$ such that
\begin{align*}
    b = \pi_m  \left( \frac{\epsilon^2 g^\downarrow}{2(1-\epsilon)g^\uparrow} \right)^{2p}
\end{align*}
and therefore
\begin{align*}
    p^{\text{PARNI}}_{\theta}(\gamma,\gamma^\prime) & = \sum_{k \in \mathcal{K}} p^{\text{PARNI}}_{\theta, k}(\gamma,\gamma^\prime) \\
     & = \sum_{k \in \mathcal{K}}  q^{\text{PARNI}}_{\theta, k}(\gamma,\gamma^\prime) \alpha^{\text{PARNI}}_{\theta, k}(\gamma,\gamma^\prime) \\
    & \geq \sum_{k \in \mathcal{K}}  b \geq b.
\end{align*}
We say that $\Gamma$ is $(1, b, \pi(\cdot))$-small, and the simultaneous uniform ergodicity of the chain has been established following Theorem 8 of \cite{roberts2004general}.

For $L$ multiple chains, the arguments are similar, but instead, the target distribution is now $\pi^{\otimes L}(\gamma^{\otimes L})$ for $\gamma^{\otimes L} \in \Gamma^{\otimes L}$ and $\Gamma^{\otimes L}$ is $(1, b^{\otimes L}, \pi^{\otimes L}(\cdot))$-small for which
\begin{align*}
    b^{\otimes L} = \pi_m^L  \left( \frac{\epsilon^2 g^\downarrow}{2(1-\epsilon)g^\uparrow} \right)^{2pL}.
\end{align*}

\end{proof}

\subsection{Proof of Lemma \ref{lemma:PARNI_dimadap}}

Before proving the lemma, we require the following inequalities and its generalised version
\begin{lemma} \label{lemma:apx_bound_series}
\begin{align}
    \bigg|\prod_{j=1}^p a_j - \prod_{j=1}^p b_j\bigg| \leq \sum_{j=1}^p |a_j - b_j| \label{mat:inequality_1}
\end{align}
for all $a_j$, $b_j \in [0,1]$.
\end{lemma}

\begin{proof}[Proof of Lemma \ref{lemma:apx_bound_series}]

    Since $a_j$ and $b_j \in [0,1]$ for all $j$, we have 
    \begin{align}
        \prod_{j=1}^p a_j \leq \sum_{j=1}^p a_j, \quad \prod_{j=1}^p b_j \leq \sum_{j=1}^p b_j.
    \end{align}
    Starting from the right hand size of (\ref{mat:inequality_1}) and combining with the above gives
    \begin{align}
        \bigg|\prod_{j=1}^p a_j - \prod_{j=1}^p b_j\bigg| & \leq \bigg|\sum_{j=1}^p a_j - \sum_{j=1}^p b_j\bigg| \notag \\
        & \leq \sum_{j=1}^p |a_j - b_j|
    \end{align}
    for which the last line follows from the triangle inequality.
    
\end{proof}

We can generalise the above lemma and obtain the following.
\begin{lemma} \label{lemma:apx_bound_series_2}
If $a_j$, $b_j \leq C$ for some constant $C > 0$, then
\begin{align}
    \bigg|\prod_{j=1}^p a_j - \prod_{j=1}^p b_j\bigg| \leq C_1 \sum_{j=1}^p |a_j - b_j|
\end{align}
for some constant $C_1$. $C_1$ can be chosen to be $C^{p-1}$.
\end{lemma}

\begin{proof}
\begin{align*}
    \bigg|\prod_{j=1}^p a_j - \prod_{j=1}^p b_j\bigg| & = C^p \bigg|\prod_{j=1}^p \frac{a_j}{C} - \prod_{j=1}^p \frac{b_j}{C}\bigg| := A
\end{align*}
As $a_j$, $b_j \leq C$, $a_j/C$ and $b_j/C$ are smaller than 1.
\begin{align*}
    A & \leq C^p \sum_{j=1}^p \bigg|\frac{a_j}{C} - \frac{b_j}{C} \bigg| \\
    & = C^{p-1} \sum_{j=1}^p |a_j - b_j|
\end{align*}
as required.
\end{proof}

The following lemma shows the diminishing rate of the proposal thinning parameter $\omega$ for both schemes (the PARNI-KW and PARNI-RM proposals).

\begin{lemma}\label{lemma:omega_dimishing}
The diminishing rate of adaptive parameter $\omega$ in both (\ref{mat:update_omega_RM}) and (\ref{mat:update_omega_KW}) between two consecutive iterations satisfies
\begin{align}
    |\omega^{(i+1)} - \omega^{(i)}| = \mathcal{O}(i^{-\lambda}) \label{mat:kwdimishing}
\end{align}
for some $\lambda > 0$. In particular, setting $\phi_i = i^{-0.7}$, $a_i = i^{-1}$ and $c_i = i^{-0.5}$ as suggested, (\ref{mat:kwdimishing}) holds for $\lambda = 0.5$.
\end{lemma}

\begin{proof}[Proof of Lemma \ref{lemma:omega_dimishing}]
    
    The update rule of (\ref{mat:update_omega_RM}) immediately leads to 
    \begin{align*}
        |\omega^{(i+1)} - \omega^{(i)}| = \mathcal{O}(i^{-\lambda})
    \end{align*}
    for $\lambda = 0.7$.
    
    For the Kiefer-Wolfowitz updating law in (\ref{mat:update_omega_KW}), the values of tuning parameters $\omega$ adopted involve diminishing sequence $c_i$.  We are therefore interested in 
    \begin{align}
        \bigg||\omega^{(i+1)} \pm c_{i+1}| - |\omega^{(i)} \pm c_i|\bigg|. \label{mat:omega_kw_dimishing}
    \end{align}
    By applying the triangle inequality, we obtain
    \begin{align*}
    \bigg||\omega^{(i+1)} \pm c_{i+1}| - |\omega^{(i)} \pm c_i|\bigg| \leq & \bigg|(\omega^{(i+1)} \pm c_{i+1}) - (\omega^{(i)} \pm c_i)\bigg| \\
    \leq & \bigg|\omega^{(i+1)} + c_{i+1} - \omega^{(i)} + c_i\bigg| \\
    \leq & \bigg|\omega^{(i+1)} - \omega^{(i)}\bigg| + \bigg|c_{i+1}  + c_i\bigg|
    := S_1 + S_2,
\end{align*}
Starting from the first term $S_1$ and rearranging (\ref{ESJDupdateKW}), we obtain
\begin{align*}
    S_1 \leq & \Bigg| a_i\left(\frac{\text{ASJD}^{+,(i)} - \text{ASJD}^{-,(i)}}{2c_i}\right)\Bigg| \\
    \leq & p \Bigg| \frac{a_i}{c_i}\Bigg| \\
    = & \mathcal{O}(i^{-(\phi_a + \phi_c)}) = \mathcal{O}(i^{-0.5})
\end{align*}
where the second line follows from the fact that the expected jumping distances are bounded above by $p$.

Substituting the definition of $c_i$ into $S_2$ yields
\begin{align*}
    S_2 \leq & \bigg|(i+1)^{-\phi_c}  + (i)^{-\phi_c}\bigg| \\
    \leq & \bigg|2i^{-\phi_c}\bigg| \\
    = & \mathcal{O}(i^{-\phi_c}) = \mathcal{O}(i^{-0.5}).
\end{align*}

Since both terms $S_1$ and $S_2$ are of the same order of $\mathcal{O}(i^{-0.5})$, the equation (\ref{mat:omega_kw_dimishing}) is also $\mathcal{O}(i^{-0.5})$, which completes the proof.

\end{proof}

We also require the following lemma to bound transition kernels by proposal kernels. The lemma and its proof is inspired by Lemma 4.21 in \cite{latuszynski2013adaptive}. 

\begin{lemma} \label{lemma:apx_PARNI_tran_prop}

The sub-proposal kernel in (\ref{mat:PARNI_subpropkernel}) and sub-transition kernel in (\ref{mat:PARNI_sub_trankernel}) obey the following relationship:
\begin{align}
    \sup_{\gamma \in \Gamma} \sup_{\gamma^\prime \in \Gamma}  \sup_{k \in \mathcal{K}} &
     \bigg|  p^{\text{PARNI}}_{\theta^{(i+1)}, k }(\gamma, \gamma^\prime) - p^{\text{PARNI}}_{\theta^{(i)}, k}(\gamma, \gamma^\prime) \bigg| \notag  \\
     & \leq C \sup_{\gamma \in \Gamma} \sup_{\gamma^\prime \in \Gamma} \sup_{k \in \mathcal{K}}   \bigg| \psi^{\text{PARNI}}_{\theta^{(i+1)}, k}(\gamma, \gamma^\prime) -  \psi^{\text{PARNI}}_{\theta^{(i)}, k}(\gamma, \gamma^\prime)\bigg|
     \label{mat:apx_PARNI_tran_prop}
\end{align}
for some constant $C < \infty$.
\end{lemma}

\begin{proof}[Proof of Lemma \ref{lemma:apx_PARNI_tran_prop}]
   
   Let the left-hand side and right-hand side of (\ref{mat:apx_PARNI_tran_prop}) be $L$ and $R$ respectively, namely
   \begin{align}
       L & = \sup_{\gamma \in \Gamma} \sup_{\gamma^\prime \in \Gamma}  \sup_{k \in \mathcal{K}}
     \bigg|  p^{\text{PARNI}}_{\theta^{(i+1)}, k }(\gamma, \gamma^\prime) - p^{\text{PARNI}}_{\theta^{(i)}, k}(\gamma, \gamma^\prime) \bigg| \\
     R & = \sup_{\gamma \in \Gamma} \sup_{\gamma^\prime \in \Gamma} \sup_{k \in \mathcal{K}}  \bigg| \psi^{\text{PARNI}}_{\theta^{(i+1)}, k }(\gamma, \gamma^\prime) -  \psi^{\text{PARNI}}_{\theta^{(i)}, k}(\gamma, \gamma^\prime)\bigg|.
   \end{align}

   Starting from the definition of sub-proposal kernel
   \begin{align*}
      p^{\text{PARNI}}_{\theta, k}(\gamma, \gamma^\prime)
      = &  \psi^{\text{PARNI}}_{\theta, k}(\gamma, \gamma^\prime) \alpha^{\text{PARNI}}_{\theta, k}(\gamma, \gamma^\prime) + \mathbb{I}\{\gamma = \gamma^\prime\} \sum_{\gamma^\prime \in \Gamma} \psi^{\text{PARNI}}_{\theta, k}(\gamma, \gamma^\prime)(1- \alpha^{\text{PARNI}}_{\theta, k}(\gamma, \gamma^\prime))
   \end{align*}
   and substituting it into the left-hand side of (\ref{mat:apx_PARNI_tran_prop}), we obtain
   \begin{align*}
       \bigg|  p^{\text{PARNI}}_{\theta^{(i+1)}, k}(\gamma, \gamma^\prime) & - p^{\text{PARNI}}_{\theta^{(i)}, k}(\gamma, \gamma^\prime) \bigg| \\
       \leq & \bigg| \psi^{\text{PARNI}}_{\theta^{(i+1)}, k}(\gamma, \gamma^\prime) \alpha^{\text{PARNI}}_{\theta^{(i+1)}, k}(\gamma, \gamma^\prime) - \psi^{\text{PARNI}}_{\theta^{(i)}, k}(\gamma, \gamma^\prime) \alpha^{\text{PARNI}}_{\theta^{(i)}, k}(\gamma, \gamma^\prime)   \bigg| \\ 
       & + \mathbb{I}\{\gamma = \gamma^\prime\} \sum_{\gamma^\prime \in \Gamma} \bigg| \psi^{\text{PARNI}}_{\theta^{(i+1)}, k}(\gamma, \gamma^\prime)(1- \alpha^{\text{PARNI}}_{\theta^{(i+1)}, k}(\gamma, \gamma^\prime)) \\
       & \qquad \qquad - \psi^{\text{PARNI}}_{\theta^{(i)}, k}(\gamma, \gamma^\prime)(1- \alpha^{\text{PARNI}}_{\theta^{(i)}, k}(\gamma, \gamma^\prime)) \bigg| := \text{I} + \text{II} 
   \end{align*}
   
   Starting with the term $I$ and substituting in the definition of Metropolis-Hastings acceptance probability in (\ref{mat:PARNIacc}) gives
   \begin{align}
       \text{I} = & \bigg| \min\{\psi^{\text{PARNI}}_{\theta^{(i+1)}, k}(\gamma, \gamma^\prime), \frac{\pi(\gamma^\prime)}{\pi(\gamma)}\psi^{\text{PARNI}}_{\theta^{(i+1)}, k}(\gamma^\prime, \gamma)\} \notag \\ 
       & \qquad - \min\{\psi^{\text{PARNI}}_{\theta^{(i)}, k}(\gamma, \gamma^\prime), \frac{\pi(\gamma^\prime)}{\pi(\gamma)}\psi^{\text{PARNI}}_{\theta^{(i)}, k}(\gamma^\prime, \gamma)\} \bigg| \label{mat:apx_parni_tran_prop_1}
   \end{align}
   Using $|\min\{a,b\} - \min\{c,d\}| < |a-c| + |b-d|$ to further split $\text{I}$
   \begin{align}
       \text{I} = & |\psi^{\text{PARNI}}_{\theta^{(i+1)}, k}(\gamma, \gamma^\prime) - \psi^{\text{PARNI}}_{\theta^{(i)}, k}(\gamma, \gamma^\prime)| \notag \\
       & \qquad + \frac{\pi(\gamma^\prime)}{\pi(\gamma)} |\psi^{\text{PARNI}}_{\theta^{(i+1)}, k}(\gamma^\prime, \gamma) - \psi^{\text{PARNI}}_{\theta^{(i)}, k}(\gamma^\prime, \gamma)| \notag \\
       \leq & (1+\Pi)  R \label{mat:apx_parni_tran_prop_2}
   \end{align}
   where the last line follows from the assumption (A.2).
   An analogous calculation gives
   \begin{align}
       \text{II} \leq (2 + K) R,
   \end{align}
   which together with (\ref{mat:apx_parni_tran_prop_2}) implies that
   \begin{align*}
       \bigg|  p^{\text{PARNI}}_{\theta^{(i+1)}, k }(\gamma, \gamma^\prime) - p^{\text{PARNI}}_{\theta^{(i)}, k}(\gamma, \gamma^\prime) \bigg| \leq (3+2\Pi) R
   \end{align*}
   for any values of $\gamma$, $\gamma^\prime$, $k$ and $O$, hence, it is enough to prove
   \begin{align}
       L \leq C \cdot R
   \end{align}
   for $C = (3+2\Pi)$ as required.

\end{proof}

The proof of Lemma \ref{lemma:PARNI_dimadap} is structured similarly to the proof of Lemma 2 in \cite{griffin2021search}.

\begin{proof}[Proof of Lemma \ref{lemma:PARNI_dimadap}]

The proof is structured as follows: firstly, we re-write the problem as a sum of sub-transition kernels and bound the sub-transition kernels by sub-proposal kernels; secondly, we break sub-proposal kernels into various parts; lastly, we bound each part individually and hence bound the proposal kernels.

Starting from the total variation norm in (\ref{mat:PARNI_dimadap_tv}) and substituting in the definition of $P^{\text{PARNI}}_{\theta}$ in (\ref{mat:PARNI_trankernel}), we have
\begin{align*}
   \sup_{\gamma \in \Gamma} \|P^{\text{PARNI}}_{\theta^{(i+1)}}(\gamma, \cdot) - P^{\text{PARNI}}_{\theta^{(i)}}(\gamma, \cdot) \|_{TV} 
   = & \sup_{\gamma \in \Gamma} \sum_{\gamma^\prime \in \Gamma} \sum_{k \in \mathcal{K}} \bigg|p^{\text{PARNI}}_{\theta^{(i+1)}, k }(\gamma, \gamma^\prime) - p^{\text{PARNI}}_{\theta^{(i)}, k}(\gamma, \gamma^\prime) \bigg| \\
   \leq & C_1 \sup_{\gamma \in \Gamma} \sup_{\gamma^\prime \in \Gamma}  \sup_{k \in \mathcal{K}} 
   \bigg|p^{\text{PARNI}}_{\theta^{(i+1)}, k }(\gamma, \gamma^\prime) - p^{\text{PARNI}}_{\theta^{(i)}, k}(\gamma, \gamma^\prime) \bigg| \\
   := & \text{ I}
\end{align*}
for some constant $C_1 < \infty$. Using Lemma \ref{lemma:apx_PARNI_tran_prop}, the problem is reduced to bounding the largest variations in two consecutive proposal kernels
\begin{align}
    \text{I} \leq & C_2 \sup_{\gamma \in \Gamma} \sup_{\gamma^\prime \in \Gamma}  \sup_{k \in \mathcal{K}}  
   \bigg|\psi^{\text{PARNI}}_{\theta^{(i+1)}, k }(\gamma, \gamma^\prime) - \psi^{\text{PARNI}}_{\theta^{(i)}, k}(\gamma, \gamma^\prime) \bigg|
\end{align}
for some constant $C_2 < \infty$. Plugging in the definition of $\psi^{\text{PARNI}}_{\theta, k}$ into (\ref{mat:PARNI_subpropkernel}), therefore,
\begin{align}
    \text{I} \leq C_3 \sup_{\gamma \in \Gamma} \sup_{\gamma^\prime \in \Gamma}  \sup_{k \in \mathcal{K}}  
   \bigg|p^{\text{RN}}_{\eta^{(i+1)}}(k|\gamma)q^{\text{PARNI}}_{\theta^{(i+1)}, k }(\gamma, \gamma^\prime) - p^{\text{RN}}_{\eta^{(i)}}(k|\gamma) q^{\text{PARNI}}_{\theta^{(i)}, k}(\gamma, \gamma^\prime) \bigg| \label{mat:apx_parni_dimiadap_1}
\end{align}
for some constant $C_3 < \infty$. Applying Lemma \ref{lemma:apx_bound_series} to (\ref{mat:apx_parni_dimiadap_1}), we obtain
\begin{align*}
    \text{I} \leq & C_3 \sup_{\gamma \in \Gamma} \sup_{\gamma^\prime \in \Gamma}  \sup_{k \in \mathcal{K}}  \Bigg( \underbrace{\bigg|p^{\text{RN}}_{\eta^{(i+1)}}(k|\gamma)- p^{\text{RN}}_{\eta^{(i)}}(k|\gamma)\bigg|}_{\text{II}}+
     \underbrace{ \bigg|q^{\text{PARNI}}_{\theta^{(i+1)}, k}(\gamma, \gamma^\prime)  -  q^{\text{PARNI}}_{\theta^{(i)}, k}(\gamma, \gamma^\prime) \bigg|}_{\text{III}} \Bigg).
\end{align*}

Now we move our attention to the next part of the proof where we are going to bound terms $\text{II}$ and $\text{III}$ respectively.

Starting with $\text{II}$, recall the definition of $p^\text{RN}_{\eta}$ in (\ref{mat:RN_prop1}) and $\eta = (A, D)$, we have
\begin{align*}
    p_{\eta}(k|\gamma) & = \prod_{j=1}^p p_{\eta,j}(k_j|\gamma_j) \\
    & = \prod_{j=1}^p  \left( A_j \right)^{(1-\gamma_j)k_j} \left( 1-A_j \right)^{(1-\gamma_j)(1-k_j)} \left( D_j \right)^{\gamma_j(1-k_j)} \left( 1-D_j \right)^{\gamma_j k_j}.
\end{align*}
Following similar arguments to the proof of Lemma 2 of \cite{griffin2021search}, we obtain
\begin{align*}
    \text{II} \leq & \sum_{j=1}^p \max \left\{ |A_j^{(i+1)} - A_j^{(i)}|,  |D_j^{(i+1)} - D_j^{(i)}| \right\} \\
    \leq & p \max \left\{ \max_j \left\{|A_j^{(i+1)} - A_j^{(i)}|\right\}, \max_j \left\{|D_j^{(i+1)} - D_j^{(i)}| \right\}\right\}.
\end{align*}
From the definitions of $A_j$ and $D_j$, we have
\begin{align}
    |A_j^{(i+1)} - A_j^{(i)}| & = \Bigg| \min \left\{1, \frac{ \tilde{\pi}^{(i+1)}_j}{(1-\tilde{\pi}^{(i+1)}_j)} \right\} - \min \left\{1, \frac{ \tilde{\pi}^{(i)}_j}{(1-\tilde{\pi}^{(i)}_j)} \right\} \Bigg| \\
    |D_j^{(i+1)} - D_j^{(i)}| & = \Bigg| \min \left\{1, \frac{(1-\tilde{\pi}^{(i+1)}_j)}{ \tilde{\pi}^{(i+1)}_j} \right\} - \min \left\{1, \frac{(1-\tilde{\pi}^{(i)}_j)}{ \tilde{\pi}^{(i)}_j} \right\} \Bigg|.
\end{align}
The pseudo-code of the PARNI sampler in (\ref{alg:PARNI}) states that $\tilde{\pi}^{(i)}_j = \pi_0 + (1-2\pi_0) \hat{\pi}^{(i)}_j$ and $\hat{\pi}^{(i)}_j$ are the Rao-Blackwellised estimates of posterior inclusion probabilities $\pi_j$. Note that
\begin{align}
    \hat{\pi}^{(i)}_j = \frac{1}{i} \sum_{n=1}^i \Pr(\gamma_j = 1|y, \gamma^{(n)}_{-j})
\end{align}
for all $j$ from $1$ to $p$, and therefore
\begin{align*}
    |\hat{\pi}^{(i+1)}_j - \hat{\pi}^{(i)}_j| & = \bigg| \frac{i}{i+1}\hat{\pi}^{(i)}_j + \frac{1}{i+1} \Pr(\gamma_j = 1|y, \gamma^{(i+1)}_{-j}) - \hat{\pi}^{(i)}_j \bigg| \\
    & \leq \bigg| \frac{i}{i+1}\hat{\pi}^{(i)}_j- \hat{\pi}^{(i)}_j \bigg|  + \frac{1}{i+1} \Pr(\gamma_j = 1|y, \gamma^{(i+1)}_{-j}) \\
    & \leq \frac{2}{i+1}
\end{align*}
Note that $f_{\pi_0}(x) = \min\{1, (\pi_0 + (1-2\pi_0x))/(\pi_0 + (1-2\pi_0(1-x))\}$ is Lipshitz with constant $1/\pi_0$, we obtain
\begin{align*}
    \Bigg| \min \left\{1, \frac{ \tilde{\pi}^{(i+1)}_j}{(1-\tilde{\pi}^{(i+1)}_j)} \right\} - \min \left\{1, \frac{ \tilde{\pi}^{(i)}_j}{(1-\tilde{\pi}^{(i)}_j)} \right\} \Bigg| & \leq \frac{1}{\pi_0} |\hat{\pi}^{(i+1)}_j - \hat{\pi}^{(i)}_j| \\
    & \leq \frac{1}{\pi_0} \cdot \frac{2}{i+1}.
\end{align*}
A similar conclusion can be drawn for each $D_j$, meaning
\begin{align}
    \text{II} \leq C_4 i^{-1}
\end{align}
for some constant $C_4 < \infty$.

The second term, $\text{III}$, is more complicated.  We start by substituting sub proposal kernels in (\ref{mat:PARNI_prop}) to $\text{III}$ yielding
\begin{align}
    \text{III} = \bigg| \prod_{r=1}^{p_k} q_{\theta^{(i+1)}, K_r}^{\text{PARNI}} (\gamma(r-1), \gamma(r)) - \prod_{r=1}^{p_k} q_{\theta^{(i)}, K_r}^{\text{PARNI}} (\gamma(r-1), \gamma(r)) \bigg| \label{mat:apx_parni_dimiadap_2}
\end{align}
where $\gamma(0) = \gamma$ and $\gamma(p_k) = \gamma^\prime$. Applying Lemma \ref{lemma:apx_bound_series} to (\ref{mat:apx_parni_dimiadap_2}), we have
\begin{align}
    \text{III} & \leq  \sum_{r=1}^{p_k} \bigg| q_{\theta^{(i+1)}, K_r}^{\text{PARNI}} (\gamma(r-1), \gamma(r)) - q_{\theta^{(i)}, K_r}^{\text{PARNI}} (\gamma(r-1), \gamma(r)) \bigg| \notag \\
    & \leq p\max_{r} \bigg| q_{\theta^{(i+1)}, K_r}^{\text{PARNI}} (\gamma(r-1), \gamma(r)) - q_{\theta^{(i)}, K_r}^{\text{PARNI}} (\gamma(r-1), \gamma(r)) \bigg| \notag  \\
    & := \text{IV}. \notag
\end{align}

The sub-proposal kernels typically contain two moves, either flipping position $K_r$ or keeping it.  Term $\text{IV}$ is smaller than the maximum probability of these two moves. Let $\text{V}$ be the absolute difference in flipping and $\text{VI}$ be the absolute difference in keeping, then
\begin{align*}
    \text{IV} & \leq q \max_r \{ \max \left\{\text{V} , \text{VI} \right\} \}.
\end{align*}
We next consider terms $\text{V}$ and $\text{VI}$. Starting with $\text{V}$ and substituting (\ref{mat:PARNI_prop}) in $\text{V}$ gives
\begin{align*}
    \text{V} & = \left| \frac{\omega^{(i+1)} g \left( \frac{\pi(\gamma(r))}{\pi(\gamma(r-1))}  \left(\frac{A^{(i+1)}_{K_r}}{D^{(i+1)}_{K_r}}\right)^{ \gamma(r)_{K_r} - \gamma(r-1)_{K_r}} \right)}{Z^{(i+1)}(r)} - \frac{\omega^{(i)} g \left( \frac{\pi(\gamma(r))}{\pi(\gamma(r-1))}  \left(\frac{A^{(i)}_{K_r}}{D^{(i)}_{K_r}}\right)^{ \gamma(r)_{K_r} - \gamma(r-1)_{K_r}} \right)}{Z^{(i)}(r)}\right|.
    \end{align*}
    Reduce the fractions to a common denominator to yield
    \begin{align*}
    \text{V} & = \left| \frac{\omega^{(i+1)} g \left( \frac{\pi(\gamma(r))}{\pi(\gamma(r-1))}  \left(\frac{A^{(i+1)}_{K_r}}{D^{(i+1)}_{K_r}}\right)^{ \gamma^d(r) } \right)Z^{(i)}(r)  - \omega^{(i)} g \left( \frac{\pi(\gamma(r))}{\pi(\gamma(r-1))}  \left(\frac{A^{(i)}_{K_r}}{D^{(i)}_{K_r}}\right)^{ \gamma^d(r)}\right)Z^{(i+1)}(r) }{Z^{(i+1)}(r) Z^{(i)}(r)}  \right| \\
    & \leq \frac{1}{(Z^\downarrow)^2} \left| \omega^{(i+1)} g \left( \frac{\pi(\gamma(r))}{\pi(\gamma(r-1))} \left(\frac{A^{(i+1)}_{K_r}}{D^{(i+1)}_{K_r}}\right)^{ \gamma^d(r) } \right)Z^{(i)}(r)  - \omega^{(i)} g \left( \frac{\pi(\gamma(r))}{\pi(\gamma(r-1))} \left(\frac{A^{(i)}_{K_r}}{D^{(i)}_{K_r}}\right)^{ \gamma^d(r) }\right) Z^{(i+1)}(r)  \right| 
    \end{align*}
    where $\gamma^d(r) = \gamma(r)_{K_r} - \gamma(r-1)_{K_r}$ and the last line follows from (\ref{mat:apx_parni_sue_6}) in the proof of Lemma \ref{lemma:PARNI_simerg} where all normalising constants can be bounded above and below. Using Lemma \ref{lemma:apx_bound_series_2}, we obtain
    \begin{align*}
    \text{V} & \leq  C_5 \underbrace{\left| \omega^{(i+1)} - \omega^{(i)} \right|}_{:=\text{VII}} + \\
    & \qquad C_6 \underbrace{ \left|  g \left( \frac{\pi(\gamma(r))}{\pi(\gamma(r-1))} \left(\frac{A^{(i+1)}_{K_r}}{D^{(i+1)}_{K_r}}\right)^{ \gamma^d(r) } \right) - g \left( \frac{\pi(\gamma(r))}{\pi(\gamma(r-1))} \left(\frac{A^{(i)}_{K_r}}{D^{(i)}_{K_r}}\right)^{ \gamma^d(r) }\right) \right|}_{:=\text{VIII}} + \\
    & \qquad C_7 \underbrace{\left| Z^{(n+1)}(r) - Z^{(n)}(r)  \right|}_{:=\text{IX}}
\end{align*}
for some constants $C_5, C_6, C_7 < \infty$. We can apply similar arguments to $\text{VI}$ giving
\begin{align*}
    \text{VI} & = \left| \frac{(1-\omega^{(i+1)}) g(1)}{Z^{(i+1)}(r)} - \frac{(1-\omega^{(i)}) g(1)}{Z^{(i)}(r)}\right| \\
    & \leq C_8 \underbrace{\left| \omega^{(i+1)} -  \omega^{(i)}\right|}_{\text{VII}} + C_{9} \underbrace{ \left| Z^{(i+1)}(r) - Z^{(i)}(r) \right|}_{\text{IX}}
\end{align*}
for some constants $C_8, C_{9} < \infty$. Starting with $\text{IX}$, by substituting in the definitions in (\ref{mat:PARNInorcon}), we have
\begin{align*}
    \text{IX} & = \left| Z^{(i+1)}(r) - Z^{(i)}(r) \right|  \\
    & = \left| \left( \omega^{(i+1)} g \left( \frac{\pi(\gamma(r))}{\pi(\gamma(r-1))} \left(\frac{A^{(i+1)}_{K_r}}{D^{(i+1)}_{K_r}}\right)^{ \gamma^d(r) } \right) + (1-\omega^{(i+1)}) g(1) \right) \right.  \\
    & - \left. \left( \omega^{(i)}  g \left( \frac{\pi(\gamma(r))}{\pi(\gamma(r-1))} \left(\frac{A^{(i)}_{K_r}}{D^{(i)}_{K_r}}\right)^{ \gamma^d(r) } \right) + (1-\omega^{(i)})  g(1)\right) \right|,
\end{align*}
and applying the triangle inequality yields
\begin{align*}
    \text{IX} \leq& \left| \omega^{(i+1)} g \left( \frac{\pi(\gamma(r))}{\pi(\gamma(r-1))} \left(\frac{A^{(i+1)}_{K_r}}{D^{(i+1)}_{K_r}}\right)^{ \gamma^d(r) } \right) -  \omega^{(i)}  g \left( \frac{\pi(\gamma(r))}{\pi(\gamma(r-1))} \left(\frac{A^{(i)}_{K_r}}{D^{(i)}_{K_r}}\right)^{ \gamma^d(r) } \right)\right|
     + g(1) \Bigg| \omega^{(i+1)} - \omega^{(i)}   \Bigg| \\
    \leq&  C_{10} \underbrace{ \Bigg| \omega^{(i+1)} - \omega^{(i)} \Bigg|}_{\text{VII}}
    + C_{11}  \underbrace{ \left| g \left( \frac{\pi(\gamma(r))}{\pi(\gamma(r-1))} \left(\frac{A^{(i+1)}_{K_r}}{D^{(i+1)}_{K_r}}\right)^{ \gamma^d(r) } \right) -   g \left( \frac{\pi(\gamma(r))}{\pi(\gamma(r-1))} \left(\frac{A^{(i)}_{K_r}}{D^{(i)}_{K_r}}\right)^{ \gamma^d(r) } \right)  \right|}_{\text{VIII}}
\end{align*}
for some constants $C_{10}, C_{11} < \infty$.  The last line follows after applying Lemma \ref{lemma:apx_bound_series}.

The diminishing adaptation of tuning parameter $\omega$ is shown in Lemma \ref{lemma:omega_dimishing} where 
\begin{align}
   \text{VII} \leq C_{12} i^{-\lambda}
\end{align}
for some constant $C_{12} < \infty$ and $\lambda = 0.5$.

We now consider term \text{VIII}. From assumption (A.1), we have
\begin{align*}
    g(t_2) - g(t_1) \leq C_g (t_2 - t_1)
\end{align*}
and therefore
\begin{align}
    |g(t_2) - g(t_1)| \leq C_g |t_2 - t_1|
\end{align}
for any $t_1,t_2 > 0$. We then have
\begin{align*}
    \text{VIII} & \leq C_g \left|  \frac{\pi(\gamma(r))}{\pi(\gamma(r-1))} \left(\frac{A^{(i+1)}_{K_r}}{D^{(i+1)}_{K_r}}\right)^{ \gamma^d(r) } -   \frac{\pi(\gamma(r))}{\pi(\gamma(r-1))} \left(\frac{A^{(i)}_{K_r}}{D^{(i)}_{K_r}}\right)^{ \gamma^d(r) } \right| \\
    & \leq \Pi C_g \left|   \left(\frac{A^{(i+1)}_{K_r}}{D^{(i+1)}_{K_r}}\right)^{ \gamma^d(r) } -    \left(\frac{A^{(i)}_{K_r}}{D^{(i)}_{K_r}}\right)^{ \gamma^d(r) } \right|,
\end{align*}
where the last line follows from Assumption (A.2). Considering two possible values of $\gamma^d(r)$, namely $1$ and $-1$, we show that $\text{VIII}$ is bounded by the maximum of those values
\begin{align*}
    \text{VIII}  & \leq \Pi C_g \max \left\{ \left|    \frac{A^{(i+1)}_{K_r}}{D^{(i+1)}_{K_r}}-    \frac{A^{(i)}_{K_r}}{D^{(i)}_{K_r}}\right| , \left|    \frac{D^{(i+1)}_{K_r}}{A^{(i+1)}_{K_r}} -    \frac{D^{(i)}_{K_r}}{A^{(i)}_{K_r}} \right| \right\}.
\end{align*}
Next, multiplying the common denominator yields
\begin{align*}
     \text{VIII} & \leq \Pi C_g \max \left\{ \Bigg|    \frac{A_{K_r}^{(i+1)}D^{(i)}_{K_r} - A^{(i)}_{K_r} D^{(i+1)}_{K_r}} {D^{(i+1)}_{K_r}D^{(i)}_{K_r}}\Bigg| , \Bigg|    \frac{D^{(i+1)}_{K_r}A^{(i)}_{K_r} - D^{(i)}_{K_r} A^{(i+1)}_{K_r} }{A^{(i+1)}_{K_r}A^{(i)}_{K_r}}\Bigg| \right\}  \\
     & \leq \frac{ \Pi C_g}{\pi_0^2} \max \left\{ \Bigg|    A^{(i+1)}_{K_r}D^{(i)}_{K_r} - A^{(i)}_{K_r} D^{(i+1)}_{K_r}\Bigg| , \Bigg|    D^{(i+1)}_{K_r}A^{(i)}_{K_r} - D^{(i)}_{K_r} A^{(i+1)}_{K_r} \Bigg| \right\},
\end{align*}
which holds because $\pi_0 \leq A_j, D_j \leq 1$ from the proof of Lemma \ref{lemma:PARNI_simerg}. Then, by applying Lemma \ref{lemma:apx_bound_series}, we have
\begin{align*}
     \text{VIII} & \leq \frac{ \Pi C_g}{\pi_0^2} \max \left\{ \Bigg| A^{(i+1)}_{K_r} - A^{(i)}_{K_r} \Bigg| + \Bigg| D^{(i+1)}_{K_r} - D^{(i)}_{K_r} \Bigg| , \Bigg| A^{(i+1)}_{K_r} - A^{(i)}_{K_r} \Bigg| + \Bigg| D^{(i+1)}_{K_r} - D^{(i)}_{K_r} \Bigg| \right\} \\
     & \leq \frac{\Pi C_g}{\pi_0^2}\left(\Bigg| A^{(i+1)}_{K_r} - A^{(i)}_{K_r} \Bigg| + \Bigg| D^{(i+1)}_{K_r} - D^{(i)}_{K_r} \Bigg| \right) \\
     & \leq C_{13}\left( \max_j \left\{\Bigg| A^{(i+1)}_{j} - A^{(i)}_{j} \Bigg| \right\} + \max_j \left\{\Bigg| D^{(i+1)}_{j} - D^{(i)}_{j} \Bigg| \right\}  \right)
\end{align*}
for some constant $C_{13} < \infty$.

As we have previously showed that
\begin{align}
    \max_j \left\{\Bigg| A^{(i+1)}_{j} - A^{(i)}_{j} \Bigg| \right\} &  \leq \frac{1}{\pi_0} \cdot \frac{2}{i+1} \\
    \max_j \left\{\Bigg| D^{(i+1)}_{j} - D^{(i)}_{j} \Bigg| \right\}&  \leq \frac{1}{\pi_0} \cdot \frac{2}{i+1},
\end{align}
this leads to
\begin{align*}
    \text{VIII} & \leq C_{14} \frac{1}{\pi_0} \cdot \frac{2}{i+1} \\
    & \leq C_{15} i^{-1}
\end{align*}
for some constants $C_{14},C_{15} < \infty$.

We complete the proof by stating that $\text{IV} \leq C_{16} i^{-\lambda}$ for some constant $C_{16} < \infty$, and hence $\text{I} \leq C_{17} i^{-\lambda}$ for some constant $C_{17} < \infty$ and $\lambda = 0.5$, which shows that  diminishing adaptation for the PARNI sampler has established.

\end{proof}

\subsection{Proof of Theorem \ref{theo:PARNI_ergslln}}

\begin{proof}

\emph{Simultaneous uniform ergodicity} together with \emph{diminishing adaption} are enough to show that $\pi$ is stationary for each kernel $P^{\text{PARNI}}_{\theta}$ and the adaptive algorithm is ergodic from Theorem 1 in \cite{roberts2007coupling}. Its multiple chain version is also ergodic with respect to $\pi^{\otimes L}$.

The proof of the Strong Law of Large Numbers (SLLN) contains two steps. Firstly, we show that each individual chain satsifies a SLLN, that is
\begin{align}
    \frac{1}{N} \sum_{i=0}^{N-1} f(\gamma^{l,(i)}) \to \pi(f) \quad \text{almost surely as $N \to \infty$}. \label{mat:apx_parni_erg_thm_1}
\end{align}
Then by averaging over $L$ parallel chains, we can show that the SLLN in (\ref{mat:PARNI_SLLN}) is satisfied for the multiple chain version.

We use Theorem 2.7 in \cite{fort2011convergence} to show that (\ref{mat:apx_parni_erg_thm_1}) holds. To do so, three conditions are required.
\begin{enumerate} [label=(Con.\arabic*)]
    \item \label{PARNI_ergslln_con1} The measurable function $V$ can be chosen to be the constant function $V \equiv 1$, where $V$-variation distance norm reduces to the total variation distance. It is obvious that the below is met if with $\lambda_\theta = 1/2$, $b_\theta = 1$, the measure $\nu_\theta$ is uniformly distributed on the sample space $\Gamma$, that is
    \begin{align*}
         \nu_\theta(\gamma) = \frac{1}{2^p}, 
    \end{align*}
    with $\delta_\theta = b$ (the lower bound for a single chain in Lemma \ref{lemma:PARNI_simerg}), then
    \begin{align*}
        P^{\text{PARNI}}_\theta V & \leq \frac{1}{2}V + 1 \\
        P^{\text{PARNI}}_\theta(\gamma, \cdot) & \geq b \nu_\theta(\cdot) I\{V \leq c_\theta\}(\gamma), \quad c_\theta = 2b_\theta (1-\lambda_\theta)^{-1} - 1 = 3
    \end{align*}
    is satisfied.
    \item \label{PARNI_ergslln_con2} Using the same parameters specified in \ref{PARNI_ergslln_con1}, the problem is reduced to
    \begin{align*}
        \sum_{i=1}^\infty i^{-1} \sup_{\gamma \in \Gamma}\|P^{\text{PARNI}}_{\theta^{(i+1)}}(\gamma, \cdot) - P^{\text{PARNI}}_{\theta^{(i)}}(\gamma, \cdot)\| < + \infty.
    \end{align*}
    This is satisfied since the PARNI sampler satisfies diminishing adaption by Lemma \ref{lemma:PARNI_dimadap}.
    \item \label{PARNI_ergslln_con3} Condition A5 in \emph{Fort et al.} \cite{fort2011convergence} is trivially satisfied with the parameters chosen in \ref{PARNI_ergslln_con1}.
\end{enumerate}
We have established \ref{PARNI_ergslln_con1}, \ref{PARNI_ergslln_con2}, and \ref{PARNI_ergslln_con3}, therefore by Corollary 2.8 in \cite{fort2011convergence}, (\ref{mat:apx_parni_erg_thm_1}) holds,
and so does (\ref{mat:PARNI_SLLN}).

\end{proof}

\newpage

\section{Additional numerical results}
\label{appendixlabel3}

This section will provide more numerical results in addition to Section \ref{sec:numerical_studies}.

\subsection{Sensitivity analysis on thinning parameter $\omega$ for simulated datasets}
\label{apx:yang_omega}

\begin{figure}
\includegraphics[width=0.8\columnwidth]{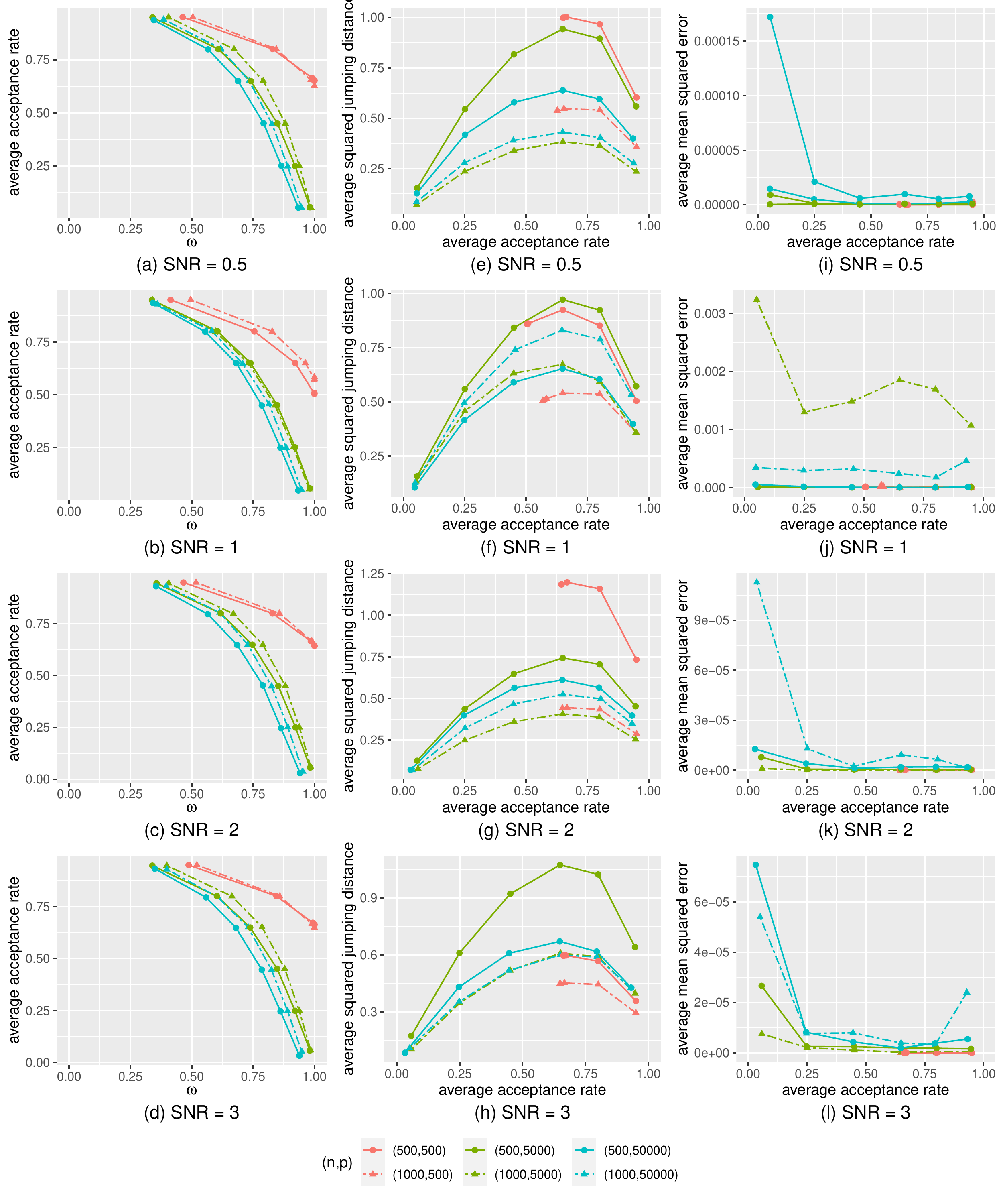}
\centering
\caption{Simulated data: plots of average squared jumping distance and average mean square error against average acceptance rate and $\omega$. (a) - (d) average acceptance rate against $\omega$ for simulated datasets with signal-to-noise ratio of 0.5, 1, 2 and 3; (e) - (h) average squared jumping distance against average acceptance rate for simulated datasets with signal-to-noise ratio of 0.5, 1, 2, 3; (i) - (j) average mean squared error against average acceptance rate for simulated datasets with signal-to-noise ratio of 0.5, 1, 2 and 3.}
\label{fig:yang_senana}
\end{figure}

We repeat the experiment which studies the optimal value of $\omega$ and optimal average acceptance rate in Section \ref{subsec:resl_dataset} on simulated datasets used in Section \ref{subsec:simulated_dataset}. Fig. \ref{fig:yang_senana} shows the effect of manipulating the target average acceptance rate on average squared jumping distance and average mean squared errors. We split the plots into 4 sets where each set of graphs corresponds to a level of signal-to-noise ratio. Panels (a)-(d) show the negative relationship of $\omega$ against average acceptance rate. Panels (e)-(h) plot the average squared jumping distance against the average acceptance rate. Finally, panels (i)-(l) show the average acceptance rate and the average mean squared errors. These plots suggest the similar conclusion in Section \ref{subsec:resl_dataset} for which the average acceptance rate of $0.65$ yields the largest average squared jumping distance. The smallest average mean squared errors are also located around the region of 0.65 average acceptance rate.

\subsection{Additional results from Kiefer-Wolfowitz adaptation scheme}
\label{apx:KW_feasibility}

This section is to examine whether applying the Kiefer-Wolfowitz adaptation scheme to the PARNI sampler would lead to the optimal scaling property of the chains. We repeatedly ran the PARNI-KW sampler for 1,500 iterations with 3 different initial values on the 24 simulated datasets of Section \ref{subsec:simulated_dataset} and the 8 real datasets of Section \ref{subsec:resl_dataset} and recorded the values of $\omega$. The trace plots of these $\omega$ values are given in Figs.  \ref{fig:yang_omega_trace_0_5}, \ref{fig:yang_omega_trace_1}, \ref{fig:yang_omega_trace_2} and \ref{fig:yang_omega_trace_3} for the simulated datasets and Fig.  \ref{fig:real_omega_trace} for the real datasets. The black horizontal lines in these plots indicate the empirical optimal values of $\omega$ gathered for each dataset from Fig. \ref{fig:opt_acc_rate} and \ref{fig:yang_senana}. The optimal values decrease along with the dimensionality of $p$ and they are also influenced by the correlation structure for which more complicated correlation structures imply smaller values of $\omega$. It appears that the values of $\omega$ are moving towards the black lines and converging to them regardless of initial values.

There is a significant trend that the $\omega$ values are approaching the region around the optimal values fairly quickly. Surprisingly, the parameter $\omega$ converges even faster in high-dimensional problems, for example, when $p=50,000$ in simulated datasets and the SNP dataset. But there is still a rare chance that the Kiefer-Wolfowitz scheme does not lead to the optimal choice. Some of $\omega$ values become trapped on the value of $1$. This issue is mainly caused by two reasons. Firstly, the ASJD estimators are often too noisy to capture the true signal in the expected squared jumping distances. These estimators can be viewed as simple Monte Carlo with only a few samples and therefore we may obtain estimates with extremely large estimation errors. Large errors introduce uncertainty into the true direction that should be updated and make $\omega$ take longer to converge or converge to a suboptimal value. Secondly, the use of the logistic transformation makes $\omega$ difficult to be updated in the boundary areas and therefore $\omega$ easily get trapped in values of 0 or 1.

Overall, the Kiefer-Wolfowitz adaptation scheme is relatively robust in tuning $\omega$ for the PARNI sampler, and we believe it can also be applied to other adaptive MCMC schemes in tuning the scaling parameters.

\begin{figure}
\includegraphics[width=\columnwidth]{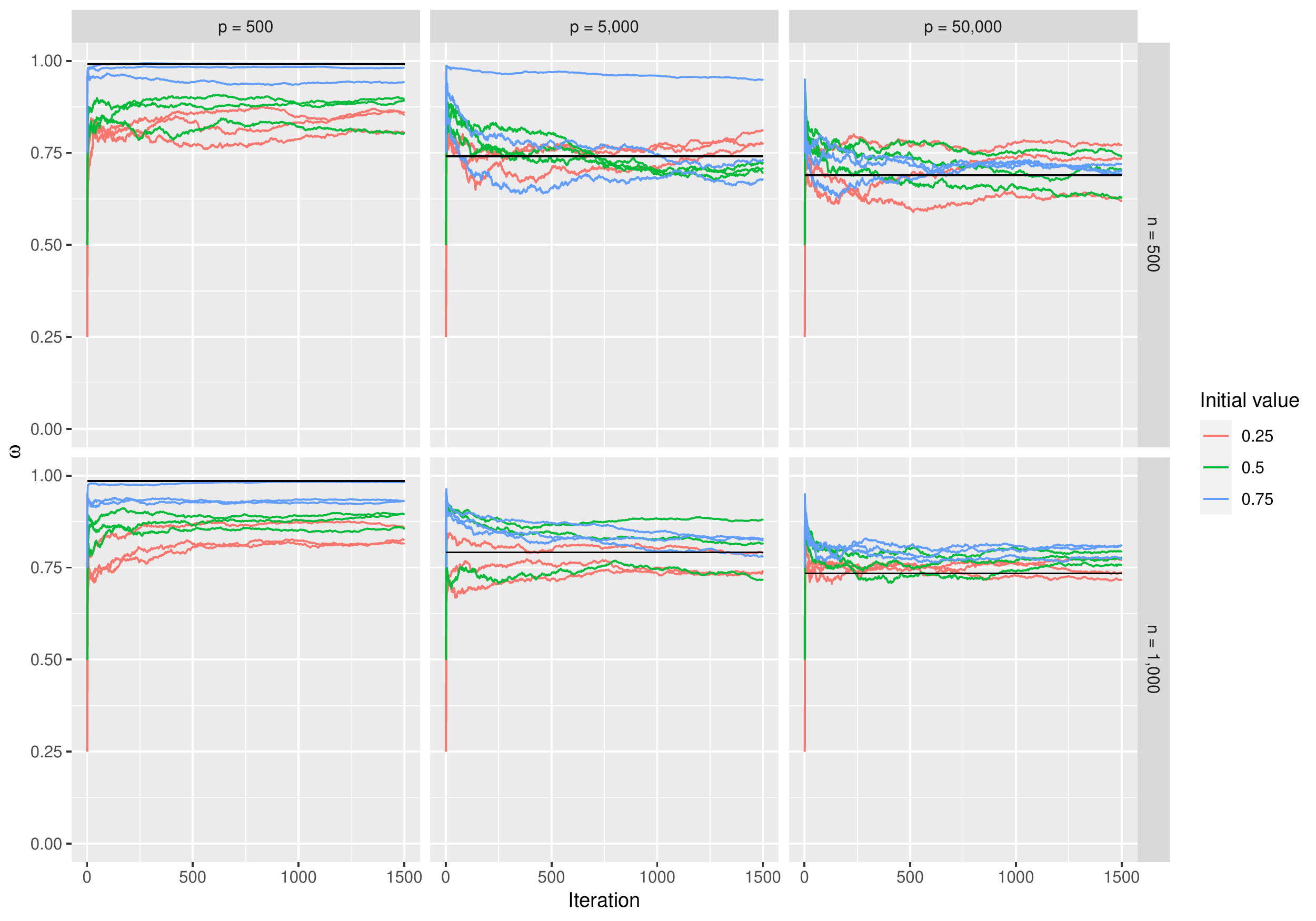}
\centering
\caption{Simulated data: trace plots of $\omega$ from the point-wise implementation of the Adaptive Random Neighbourhood Informed proposal with Kiefer-Wolfowitz update sampler for the first 1,500 iterations on simulated datasets with signal-to-noise ratio of 0.5 and three choices of initial values (0.25, 0.5 and 0.75). The black line indicates the optimal values of $\omega$ for each dataset.}
\label{fig:yang_omega_trace_0_5}
\end{figure}

\begin{figure}
\includegraphics[width=\columnwidth]{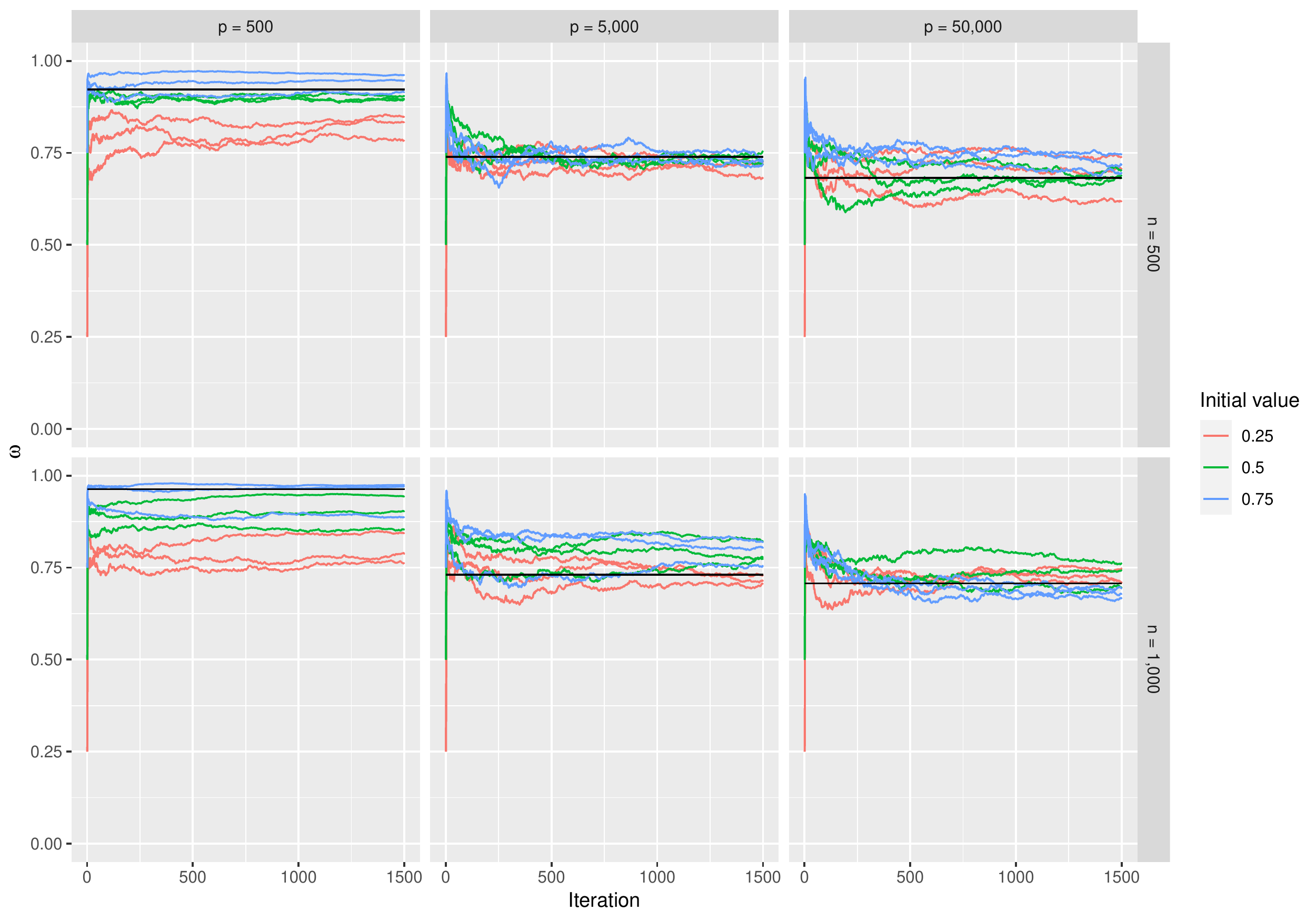}
\centering
\caption{Simulated data: trace plots of $\omega$ from the point-wise implementation of the Adaptive Random Neighbourhood Informed proposal with Kiefer-Wolfowitz update sampler for the first 1,500 iterations on simulated datasets with signal-to-noise ratio of 1 and three choices of initial values (0.25, 0.5 and 0.75). The black line indicates the optimal values of $\omega$ for each dataset.}
\label{fig:yang_omega_trace_1}
\end{figure}

\begin{figure}
\includegraphics[width=\columnwidth]{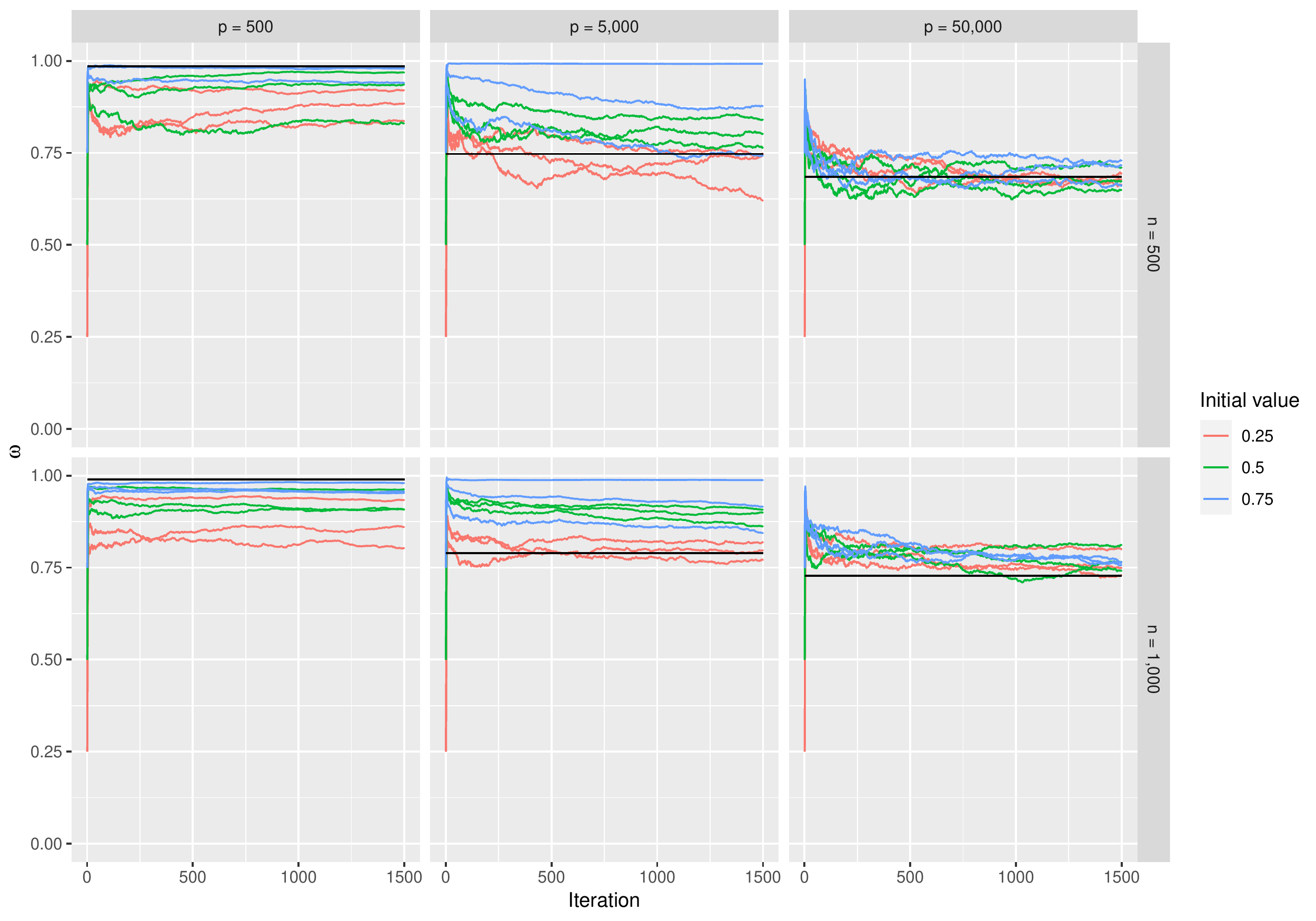}
\centering
\caption{Simulated data: trace plots of $\omega$ from the point-wise implementation of the Adaptive Random Neighbourhood Informed proposal with Kiefer-Wolfowitz update sampler for the first 1,500 iterations on simulated datasets with signal-to-noise ratio of 2 and three choices of initial values (0.25, 0.5 and 0.75). The black line indicates the optimal values of $\omega$ for each dataset.}
\label{fig:yang_omega_trace_2}
\end{figure}

\begin{figure}
\includegraphics[width=\columnwidth]{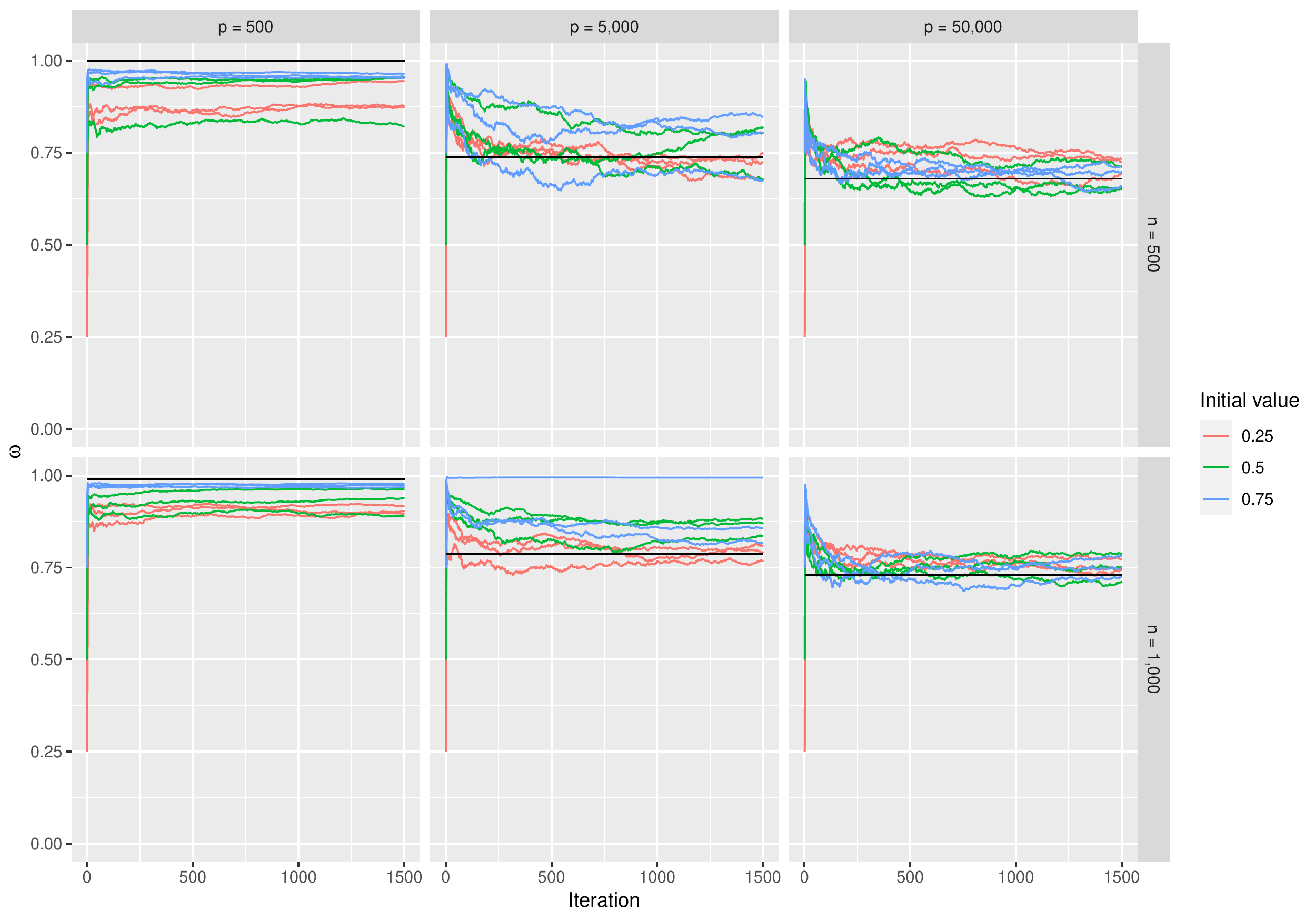}
\centering
\caption{Simulated data: trace plots of $\omega$ from the point-wise implementation of the Adaptive Random Neighbourhood Informed proposal with Kiefer-Wolfowitz update sampler for the first 1,500 iterations on simulated datasets with signal-to-noise ratio of 3 and three choices of initial values (0.25, 0.5 and 0.75). The black line indicates the optimal values of $\omega$ for each dataset.}
\label{fig:yang_omega_trace_3}
\end{figure}

\begin{figure}
\includegraphics[width=\columnwidth]{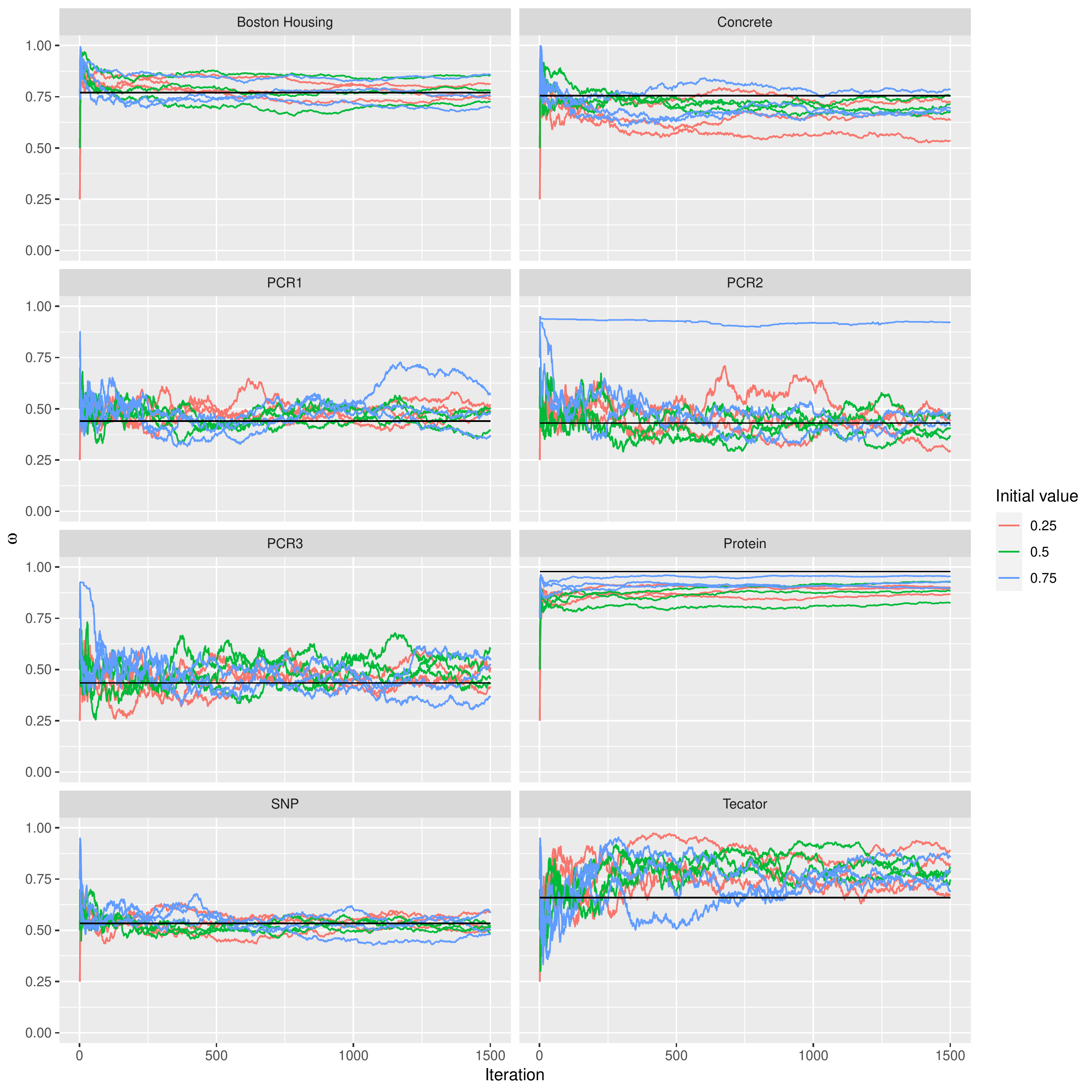}
\centering
\caption{Real data: trace plots of $\omega$ from the point-wise implementation of the Adaptive Random Neighbourhood Informed proposal with Kiefer-Wolfowitz update sampler for the first 1,500 iterations on real datasets with three choices of initial values (0.25, 0.5 and 0.75). The black line indicates the optimal values of $\omega$ for each dataset.}
\label{fig:real_omega_trace}
\end{figure}

\subsection{More results from simulated datasets}
\label{apx:yang_results}

In addition to Fig. \ref{fig:yang_trace}, Figs. \ref{fig:yang_trace_0_5}, \ref{fig:yang_trace_1} and \ref{fig:yang_trace_3} are trace plots of log posterior model probabilities from the Add-Delete-Swap, ASI, PARNI-RM and PARNI-KW schemes on the simulated datasets of Section \ref{subsec:simulated_dataset} when $\text{SNR} = 0.5, 1$ and $3$. Generally speaking, the PARNI-RM and PARNI-KW algorithms mix better than the Add-Delete-Swap and ASI schemes on all datasets. Except for the datasets for which the posterior distributions do not concentrate in a few models (when $\text{SNR} = 0.5$), the Add-Delete-Swap scheme always get struck on the empty model and struggles to include important variables and reach the high probability region within the first 1,500 iterations. The ASI algorithm mixes quite well when $p$ is relative small, but this algorithm is taking longer to converge and it is inefficient to jump between different models when $p$ reaches $50,000$. On the other hand, the PARNI-RM and PARNI-KW samplers only take dozens of iterations to converge properly in all settings. In conclusion, the plots suggest that both the PARNI schemes outperform Add-Delete-Swap and ASI in terms of the mixing time and convergence rate on the simulated datasets. They always propose models with high probability of being accepted and therefore sufficiently explore the sample space.

\begin{figure}
\includegraphics[width=\columnwidth]{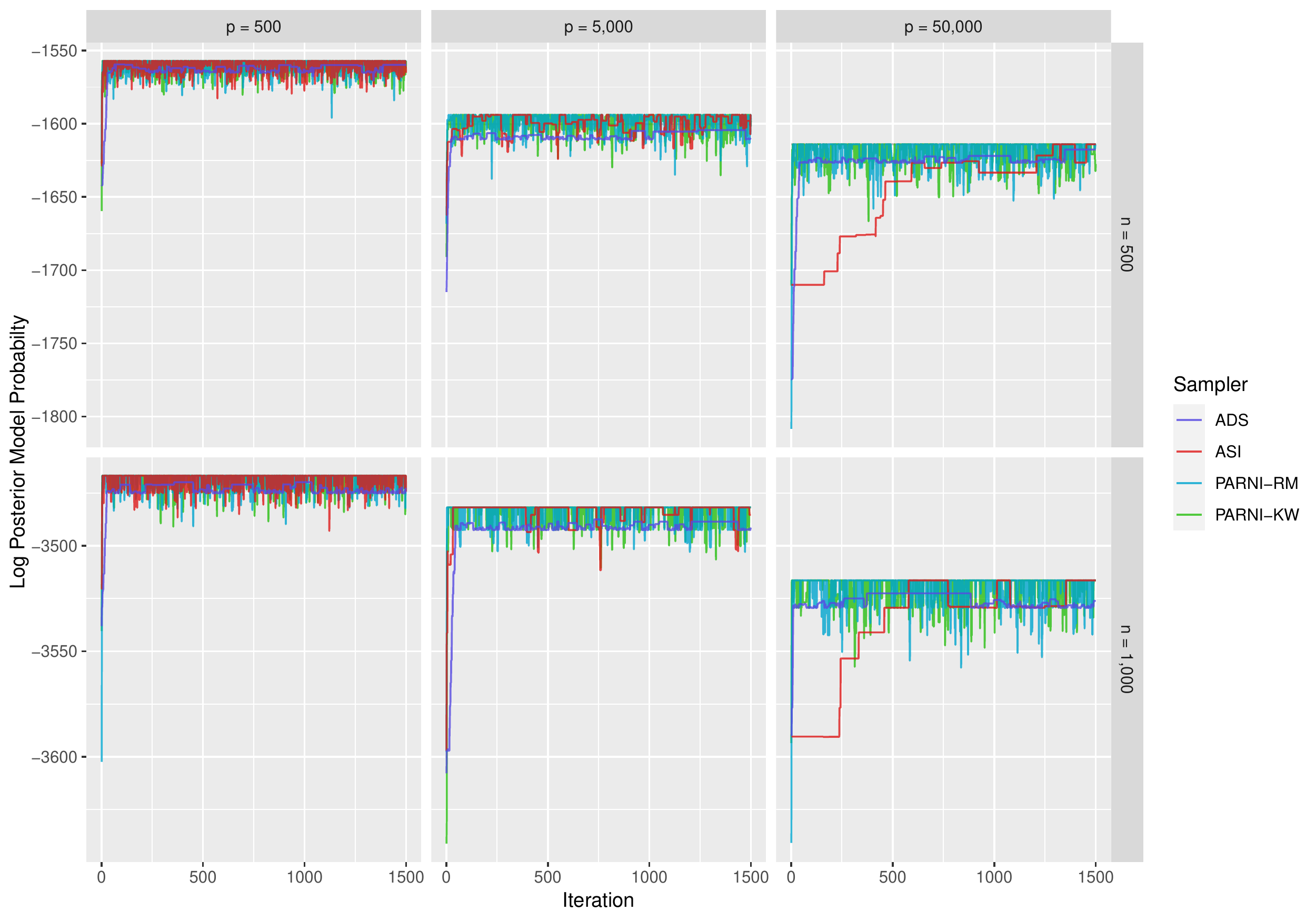}
\centering
\caption{Simulated data: trace plots of log posterior model probability from the Add-Delete-Swap (ADS), Adaptively Scaled Individual (ASI) adaptation, Pointwise implementation of Adaptive Random Neighbourhood Informed proposal with Robbins-Monro update (PARNI-RM) and Pointwise implementation of Adaptive Random Neighbourhood Informed proposal with Kiefer-Wolfowitz update (PARNI-KW) samplers for the first 1,500 iterations on simulated datasets with signal-to-noise ratio of 0.5.}
\label{fig:yang_trace_0_5}
\end{figure}

\begin{figure}
\includegraphics[width=\columnwidth]{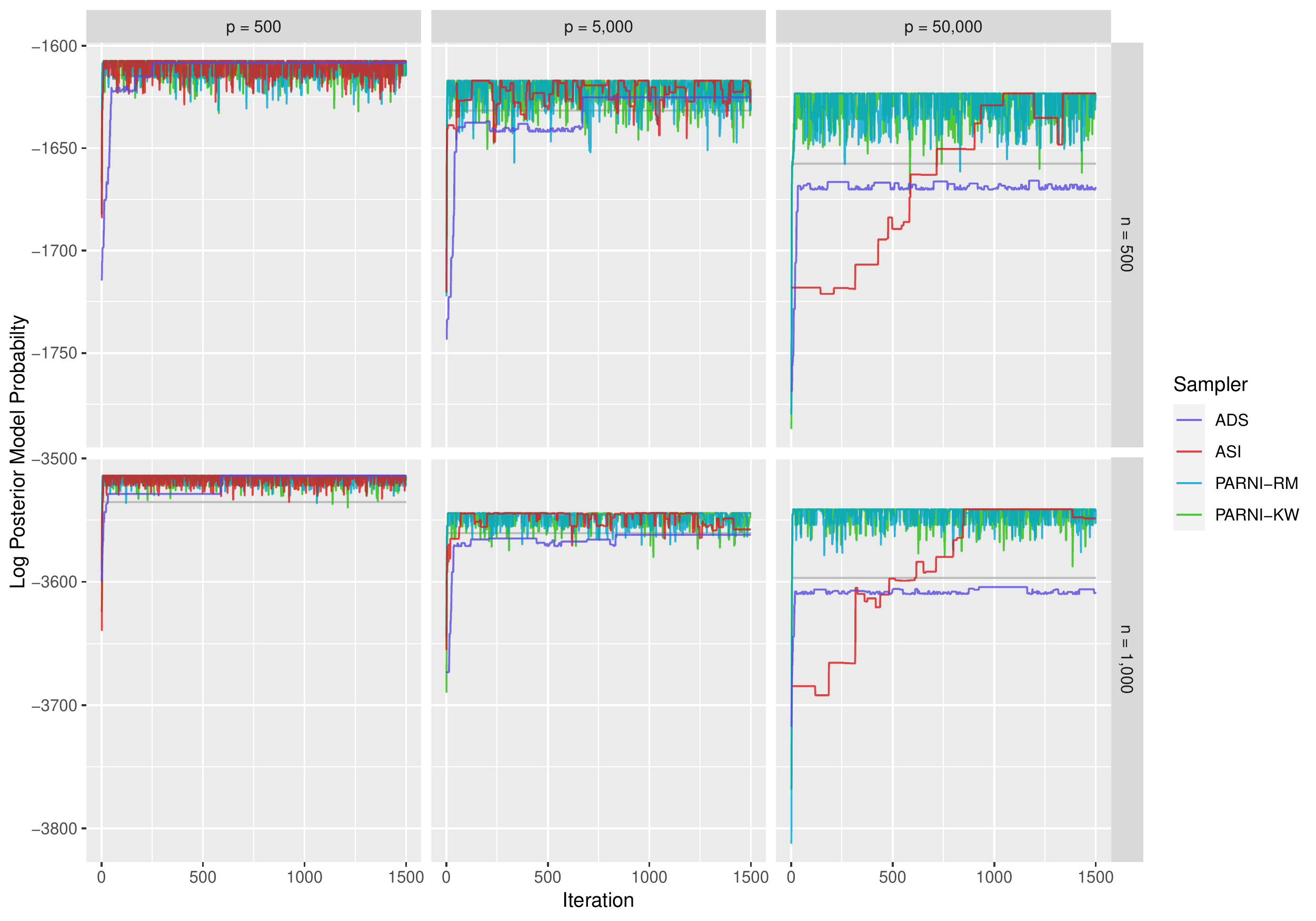}
\centering
\caption{Simulated data: trace plots of log posterior model probability from the Add-Delete-Swap (ADS), Adaptively Scaled Individual (ASI) adaptation, Pointwise implementation of Adaptive Random Neighbourhood Informed proposal with Robbins-Monro update (PARNI-RM) and Pointwise implementation of Adaptive Random Neighbourhood Informed proposal with Kiefer-Wolfowitz update (PARNI-KW) samplers for the first 1,500 iterations on simulated datasets with signal-to-noise ratio of 1.}
\label{fig:yang_trace_1}
\end{figure}

\begin{figure}
\includegraphics[width=\columnwidth]{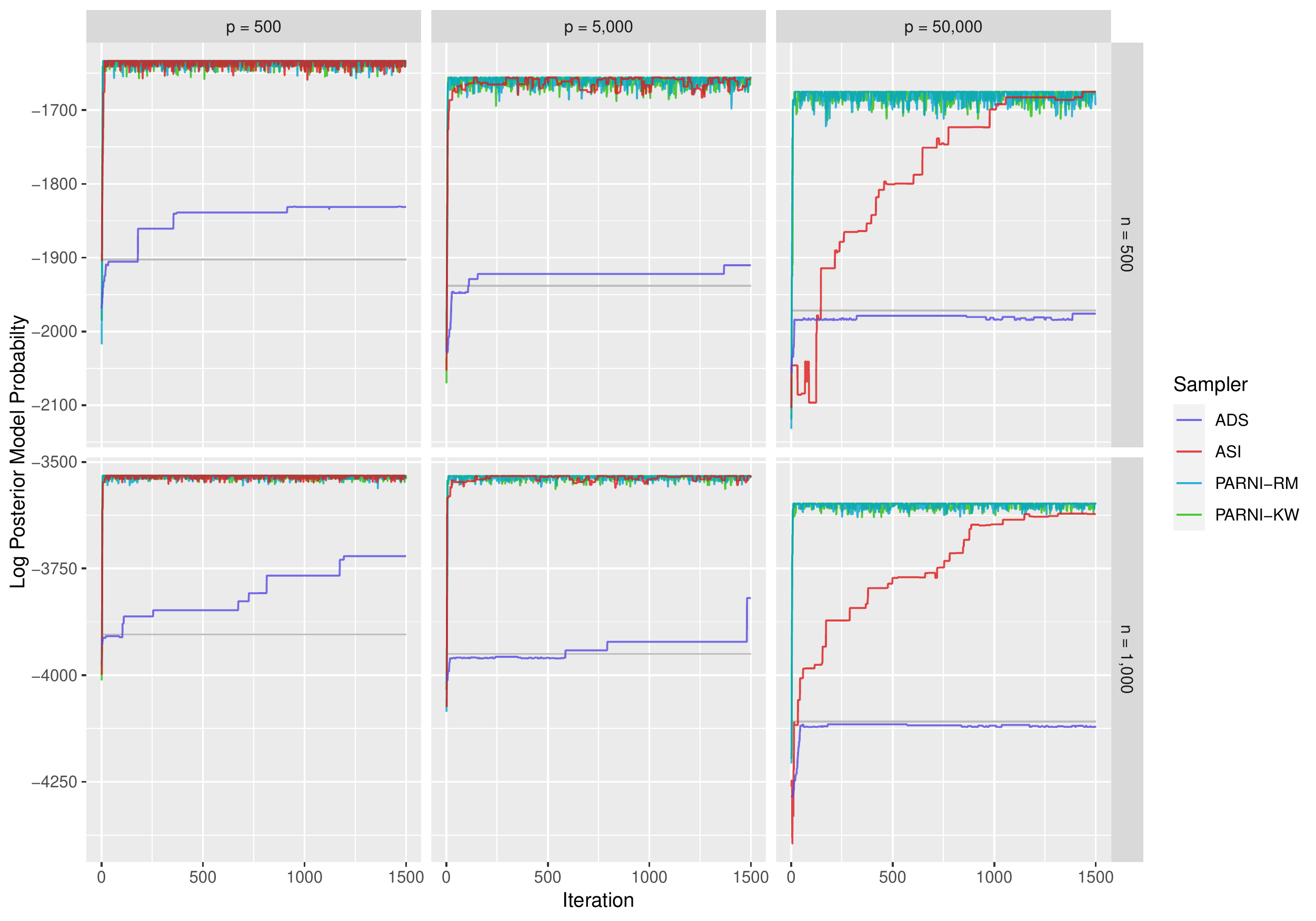}
\centering
\caption{Simulated data: trace plots of log posterior model probability from the Add-Delete-Swap (ADS), Adaptively Scaled Individual (ASI) adaptation, Pointwise implementation of Adaptive Random Neighbourhood Informed proposal with Robbins-Monro update (PARNI-RM) and Pointwise implementation of Adaptive Random Neighbourhood Informed proposal with Kiefer-Wolfowitz update (PARNI-KW) samplers for the first 1,500 iterations on simulated datasets with signal-to-noise ratio of 3.}
\label{fig:yang_trace_3}
\end{figure}

\end{document}